\documentclass[12pt]{article}
\usepackage{amsmath,amssymb,amsthm, amsfonts}
\usepackage{mathtools}
\usepackage[authoryear,round]{natbib}
\usepackage{rotating}
\usepackage[hidelinks]{hyperref}
\usepackage{tabularx}
\usepackage{booktabs}
\usepackage{titling}
\usepackage{adjustbox}
\usepackage{dsfont}
\usepackage{xurl}
\usepackage{hyperref}
\usepackage{graphicx, booktabs, multirow} 
\usepackage{setspace, enumitem}
\usepackage{float}
\usepackage{booktabs}
\usepackage{multirow}
\usepackage{caption}
\usepackage{subcaption}
\usepackage{cleveref}
\hypersetup{hidelinks, colorlinks=true, allcolors=black, pdfstartview=Fit, breaklinks=true}
\usepackage{lmodern, bm, mathrsfs, bbm, xcolor, soul}
\usepackage{appendix, authblk}
\numberwithin{equation}{section}
\allowdisplaybreaks[4]
\newtheorem{lemma}{Lemma}
\newtheorem{assumption}{Assumption}
\newtheorem{definition}{Definition}
\newtheorem{theorem}{Theorem}
\newtheorem{proposition}{Proposition}

\newtheorem{remark}{Remark}

\newcommand{\blind}{1}
\usepackage{algorithm}
\usepackage{algorithmic}
\usepackage{needspace}

\usepackage{amsmath,amssymb}
\newcommand{\algorithmeqnumber}[1]{%
  \refstepcounter{equation}%
  \label{#1}%
  \hfill\textup{(\theequation)}%
}
\begin{document}
\date{} 
\def\spacingset#1{\renewcommand{\baselinestretch}{#1}\small\normalsize}\spacingset{1} 

\if1\blind
{
  \title{\fontsize{20pt}{30pt}\bf
  Beyond Source-Level Transfer:\\ Sample-Level Learning for High-Dimensional Quantile Regression
}
  \author[a]{Xiangyu Xing} 
  \author[a]{Yingxuan Wang}
    \author[a] {Wangli Xu$^*$}
  \affil[a]{Center for Applied Statistics and School of Statistics, Renmin University of China,
Beijing, 100872, China}
  \maketitle
  \vspace{-1cm}
  \def\thefootnote{\fnsymbol{footnote}}
} \fi
\spacingset{1.9}
\if0\blind
{ 
  \bigskip\bigskip\bigskip
  \begin{center}
    {\fontsize{20pt}{30pt}\selectfont\bfseries Sample-Level 
    Transfer Learning for High-dimensional Quantile Regression}
  \end{center}
  \medskip
  \vspace{-0.8cm}
} \fi

\spacingset{1}
\bigskip
\begin{abstract}
This paper studies transfer learning for high-dimensional quantile regression with limited target samples and heterogeneous source domains. We propose a sample-level transfer learning method (SL-TL) that combines sample selection and importance weighting. 
Unlike existing source-level approaches that include or exclude entire source domains, SL-TL identifies transferable samples within each source domain and incorporates them through adaptive importance weights. 
The proposed weights rely only on one-dimensional densities associated with the quantile loss, avoiding the estimation of high-dimensional density ratios.
We establish error bounds for SL-TL estimators and characterize the effective sample size contributed by source domains. In the oracle setting, we show that SL-TL achieves faster $\ell_2$-convergence rates than existing competitors under the considered regimes.
For the unknown setting, we develop an implementable algorithm based on sample splitting and cross-fitting procedures.
Extensive simulations and a real data analysis demonstrate the finite-sample performance and robustness of SL-TL.
\end{abstract}		
			

\spacingset{1.9}

\def\thefigure{\arabic{figure}}
\def\thetable{\arabic{table}}

\renewcommand{\theequation}{\thesection.\arabic{equation}}


\section{Introduction}

High-dimensional regression is fundamental in modern statistical learning, where the number of covariates can be comparable to or much larger than the sample size \citep{candes2007dantzig,fan2008}. Sparsity is a common structural assumption, under which only a small subset of covariates has nonzero effects. This motivates sparsity-inducing methods such as LASSO, which performs coefficient shrinkage  and  variable selection \citep{LASSO1996}. Under suitable conditions, LASSO-type estimators remain consistent in high dimensions \citep{vande2008,bickel2009,meinshausen2009}. However, when the target sample is limited, exploiting sparsity alone may be insufficient for efficient estimation, motivating the use of related source data.

Transfer learning improves the performance of target domain estimation by borrowing information from related source domains. However, source domains may be heterogeneous, and naively pooling source and target data can cause negative transfer \citep{zhangnegative2026}. A key challenge is therefore to determine which source information is transferable and how to incorporate it. Existing methods for high-dimensional regression, such as Trans-Lasso \citep{Lisai2022} and Trans-SQR \citep{zhang2022transfer}, address this problem through source-level selection, in which each source domain is either entirely included or excluded. Although simple and theoretically tractable, such all-in-or-all-out strategies may waste useful information when a source is only partially informative and may increase the risk of negative transfer when the entire source is incorporated.

To address the limitations of source-level selection, \citet{zhao2026residual} propose residual importance weighted transfer learning (RIW-TL) for high-dimensional linear regression. RIW-TL assigns observation-specific weights based on density ratios dependent on residuals, allowing transferable samples to be selected within each source domain. Unlike conventional importance weighting \citep{FengJMLRQuantile,wang2026deep}, which requires estimating high-dimensional density ratios, RIW-TL relies only on one-dimensional densities and thus avoids this difficulty. This sample-level view  motivates us to investigate whether source observations can be evaluated and transferred individually in high-dimensional quantile regression.

Conditional quantile regression is important in applications where tail or heterogeneous effects are of interest, including biomedical and economic studies \citep{novak2017change,engle2004caviar}. Such problems may involve limited target domain samples and multiple related source datasets. Motivated by these settings, we propose a sample-level transfer learning framework (SL-TL) for high-dimensional quantile regression. For a given quantile level $\tau \in (0,1)$, SL-TL estimates the sparse target quantile regression coefficients by selectively borrowing information from multiple source domains. Our key contributions are summarized as follows:
\begin{itemize}
    \item [1.] 

    To mitigate negative transfer while retaining useful source information, we propose SL-TL based on a sample selection procedure and a sample-level importance weighting method tailored to quantile regression. Building on this framework, we consider two scenarios according to whether the selected sample subsets and their corresponding importance weights are known or unknown.

    \item [2.]
    We establish theoretical guarantees for  SL-TL under the known case and characterize the effective sample size contributed by each source domain. We further compare the $\ell_2$-convergence rate of oracle SL-TL with Target-only estimation and source-level transfer learning method in the single source setting. We show that the oracle SL-TL achieves a faster convergence rate across the regimes considered.

    \item [3.]
    We develop sample splitting and cross-fitted procedures for  SL-TL under the unknown case.
    We construct a B-spline projection estimator for error densities and establish a uniform error bound, which can achieve a faster convergence rate than the kernel estimator of \citet{zhao2026residual} under higher-order smoothness. Extensive simulations and a real data analysis are also provided.

\end{itemize}

The remainder of the paper is organized as follows. Section \ref{Sec: Model Estimation} introduces the high-dimensional quantile regression model, reviews existing Target-only and source-level estimators, and proposes our SL-TL estimators. Section \ref{sec: Theoretical Analysis} establishes error bounds and corresponding theoretical guarantees under various assumptions for different estimators. Section \ref{sec: Numerical Studies} evaluates the proposed estimators through simulations and a real data example. All technical proofs are provided in Supplementary Material.

Throughout the paper, $\mathbb{N}_{0}$, $\mathbb{N}$ and $\mathbb{R}$ denote the sets of nonnegative integers, positive integers and real numbers, respectively. For any $n\in\mathbb{N}$, let $[n]:=\{1,2,\dots,n\}$. For two real sequences $a_n$ and $b_n$, we write $a_n\lesssim b_n$ or equivalently $a_n=O_p(b_n)$ if $a_n\leq C{b}_{n}$ for some constant $C>0$ and all sufficiently large $n$ in probability. We write $a_n\asymp b_n$ if both $a_n \lesssim b_n$ and $a_n \gtrsim b_n$ hold. For $a,b\in\mathbb{R}$, let $a\land b:=\min\{a,b\}$ and $a\vee b:=\max\{a,b\}$. For $\boldsymbol{\alpha}=(\alpha_1,\ldots,\alpha_d)^\top\in\mathbb{R}^d$, define $\operatorname{supp}(\boldsymbol{\alpha}):=\left\{j\in[d]:\alpha_j\neq 0\right\}$.  
For $\mathcal{C}\subseteq[d]$, let $\boldsymbol{\alpha}_{\mathcal{C}}$ denote the subvector indexed by $\mathcal{C}$, and let $|\mathcal{C}|$ denote the cardinality of $\mathcal{C}$. Finally, define
$\|\boldsymbol{\alpha}\|_0:=\sum_{j=1}^d \mathds{1} \{\alpha_j\neq0\}$,
$\|\boldsymbol{\alpha}\|_1:=\sum_{j=1}^d|\alpha_j|$,
$\|\boldsymbol{\alpha}\|_2:=(\sum_{j=1}^d\alpha_j^2)^{1/2}$, and
$\|\boldsymbol{\alpha}\|_\infty:=\max_{1\leq j\leq d}|\alpha_j|$,
where $\mathds{1}\{\cdot\}$ denotes the indicator function.

\section{Methodology}
\label{Sec: Model Estimation}

\subsection{Quantile regression models}
\label{Subsec:QR-model}
In this paper, we consider a high-dimensional linear quantile regression model.
Let $\mathcal{D}^{(0)}=\{\boldsymbol{Z}_i^{(0)},i\in\mathcal{J}_0\}=\{((\boldsymbol{X}_i^{(0)})^\top,Y_i^{(0)})^\top,i\in\mathcal{J}_0\}$ denote the target dataset with $\mathcal{J}_0=[n_0]$, 
where $n_0$ is the target sample size. $\boldsymbol{X}_i^{(0)}\in\mathbb{R}^p$ and $Y_i^{(0)}\in\mathbb{R}$ represent the covariate vector and the corresponding response for the subject $i$ in the target domain, respectively.

For a given quantile level \(\tau\in(0,1)\), we assume that the $\tau$-th quantile $Q_{\tau}(Y_i^{(0)}\mid \boldsymbol{X}_i^{(0)})$ of $Y_i^{(0)}$ given $\boldsymbol{X}_i^{(0)}$ satisfies $Q_{\tau}(Y_i^{(0)}\mid \boldsymbol{X}_i^{(0)})=(\boldsymbol{X}_i^{(0)})^\top\boldsymbol{\beta}_{\tau}^{(0)}$. For notational simplicity, we write $\boldsymbol{\beta}^{(0)}:=\boldsymbol{\beta}_{\tau}^{(0)}$ hereafter. Equivalently, the model can be written as
\begin{equation}
    Y_i^{(0)}=(\boldsymbol{X}_i^{(0)})^\top \boldsymbol{\beta}^{(0)}+\varepsilon_{i}^{(0)},
    \label{quantile model}
\end{equation}
where the error  $\varepsilon_i^{(0)}$ is independent of $\boldsymbol{X}_i^{(0)}$ and satisfies $\mathbb{P}(\varepsilon_i^{(0)}\leq 0)=\tau$.

Our primary objective is to estimate the target coefficient $\boldsymbol{\beta}^{(0)}$. 
We assume that $\boldsymbol{\beta}^{(0)}$ is sparse, i.e. $\|\boldsymbol{\beta}^{(0)}\|_0=s\ll p,$ which reflects that only a small number of covariates have nonzero effects on the conditional $\tau$-th quantile of the target response. In high-dimensional settings, where $p$ may be much larger than the target sample size $n_0$, such a sparsity assumption is essential for making the estimation feasible.

\subsection{Target-only $\ell_1$-penalized quantile regression }
\label{Subsec:target-lasso-QR}

Given the sparsity condition on \(\boldsymbol{\beta}^{(0)}\), a natural Target-only estimator for model \eqref{quantile model} is the \(\ell_1\)-penalized quantile regression estimator:
\begin{align}
    \hat{\boldsymbol{\beta}}_{\mathrm{T}}^{(0)}
\in
\arg \min_{\boldsymbol{\beta}\in\mathbb{R}^p}
\big\{
\sum_{i=1}^{n_0}
\rho_\tau(Y_i^{(0)}-(\boldsymbol{X}_i^{(0)})^\top\boldsymbol{\beta})/n_0
+
\lambda _{\mathrm{T}}\|\boldsymbol{\beta}\|_1
\big\},
\label{hat_beta_T}
\end{align}
where \(\rho_\tau(u)=u\{\tau-\mathbf{1}(u\leq 0)\}\) is the check loss function and
\(\lambda_{\mathrm{T}}>0\) is a tuning parameter.

Under suitable regularity conditions, the estimator in \eqref{hat_beta_T} satisfies $ \|\hat{\boldsymbol{\beta}}_{\mathrm{T}}^{(0)} - \boldsymbol{\beta}^{(0)} \|_2 = O_p( \sqrt{ s\log p/n_0 } )  $ \citep{Belloni2011}.
This rate indicates that when the target sample size $n_0$ is small relative to $s\log p$, the target data alone may be insufficient for efficient estimation of $\boldsymbol{\beta}_0$. This limitation motivates the use of related source datasets to improve the estimation of the target quantile regression coefficient.

We next introduce the auxiliary source datasets considered throughout the transfer learning procedures. For $k=1,\ldots,K$, let $\mathcal{D}^{(k)}=\{\boldsymbol{Z}_i^{(k)}, i\in \mathcal{J}_k\}_{i=1}^{n_k}=\{((\boldsymbol{X}_i^{(k)})^\top,Y_i^{(k)})^\top, i\in \mathcal{J}_k\}_{i=1}^{n_k}$ denote the $k$-th source dataset with $\mathcal{J}_k=[n_k]$, and assume that the source datasets are independent and satisfy $n_k\gg n_0$. For each source domain $k$, we consider the linear quantile regression model
$Q_\tau(Y_i^{(k)}\mid \boldsymbol{X}_i^{(k)})=(\boldsymbol{X}_i^{(k)})^\top \boldsymbol{\beta}_{\tau}^{(k)}$. For simplicity, we write $\boldsymbol{\beta}^{(k)}:=\boldsymbol{\beta}_{\tau}^{(k)}$ hereafter. Equivalently,  $Y_i^{(k)}=(\boldsymbol{X}_i^{(k)})^\top\boldsymbol{\beta}^{(k)}+\varepsilon_i^{(k)}$, 
where the error term $\varepsilon_i^{(k)}$ is independent of $\boldsymbol{X}_i^{(k)}$ and satisfies $\mathbb{P}(\varepsilon_i^{(k)}\leq 0)=\tau$.

\subsection{Source-level transfer learning for quantile regression}
\label{Subsec:source-level}

We first review the Trans-SQR, a representative source-level transfer learning method for
high-dimensional quantile regression proposed by \citet{zhang2022transfer}. This method employs an all-in-or-all-out source selection rule under which a source dataset is included in its entirety if it satisfies a specified transferability criterion and is excluded otherwise. For the $k$-th source domain, its discrepancy from the target domain is measured by $\boldsymbol{\delta}^{(k)}=\boldsymbol{\beta}^{(k)}-\boldsymbol{\beta}^{(0)}$. Given a threshold $h>0$, the set of transferable sources is defined as $\mathcal{A}_{h}=\{1\leq k\leq K:\|\boldsymbol{\delta}^{(k)}\|_1\leq h\}$ with $n_{\mathcal{A}_h}:=\sum_{k\in\mathcal A_h}n_k$. Thus, all observations from sources in $\mathcal{A}_h$ are included in the estimation procedure, whereas observations from sources outside $\mathcal{A}_h$ are discarded. The oracle version of Trans-SQR assumes that the transferable source set $\mathcal{A}_h$ is known. It first pools the target data with the source datasets indexed by $\mathcal{A}_h$ to estimate a pooled regression coefficient and then uses only the target data to estimate a correction vector. The resulting estimation procedure of oracle Trans-SQR is summarized in Algorithm \ref{alg:Trans-SQR} of Section \ref{supp:alg} in Supplementary Material.

Under suitable regularity conditions, the oracle Trans-SQR estimator $\hat{\boldsymbol{\beta}}_{\mathrm{S}}$ satisfies the convergence bound $\|\hat{\boldsymbol{\beta}}_{\mathrm{S}}-\boldsymbol{\beta}^{(0)}\|_2
=O_p(({\log p}/{n_0})^{1/4}\{\sqrt{s}({\log p}/{(n_0+n_{\mathcal{A}_h})})^{1/4}+\sqrt{h}\}
)$.
Although this procedure can exploit information from transferable source domains, its selection rule operates at the source level and therefore treats all observations from the same source uniformly. 
To mitigate the information loss induced by  source-level selection, we next develop a sample-level transfer learning  based on importance weighting.


\subsection{Sample-level transfer learning for quantile regression}
\label{Subsec:method}

We propose a method based on importance weights called sample-level transfer learning (SL-TL). Instead of accepting or rejecting an entire source domain, SL-TL assesses source sample individually. Consequently, samples from any source domain that are compatible with the target model may contribute to the estimation, including those from domains that would be excluded under the source-level transferability criterion. 
Let 
$\boldsymbol{Z}=
((\boldsymbol{X})^\top,Y)^\top$ and denote by $p_0(\cdot)$ and $p_k(\cdot)$ the joint probability density functions of $\boldsymbol{Z}$ in the target domain and the $k$-th source domain, respectively. Then, for any $\boldsymbol{\beta}\in\mathbb{R}^p$, we have $\mathbb{E}_{\boldsymbol{Z}{\sim}p_0}[\rho_\tau(Y-\boldsymbol{X}^\top \boldsymbol{\beta})
]=\mathbb{E}_{\boldsymbol{Z}{\sim}p_k}[\omega^{(k)}(\boldsymbol{Z})\rho_\tau(Y-\boldsymbol{X}^\top \boldsymbol{\beta})]$, where $\omega^{(k)}(\boldsymbol{Z})=p_0(\boldsymbol{Z})/p_k(\boldsymbol{Z})$. Consequently, it can be easily derived that 
\begin{equation}
\boldsymbol{\beta}^{(0)}\in\arg\min_{\boldsymbol{\beta}\in\mathbb{R}^p}\mathbb{E}_{\boldsymbol{Z}{\sim}p_k}
[\omega^{(k)}(\boldsymbol{Z})\rho_\tau(Y-(\boldsymbol{X})^\top \boldsymbol{\beta})],~ k=1,\dots,K. 
\label{hat_beta_0}
\end{equation}


The importance weight $\omega^{(k)}(\cdot)$ in \eqref{hat_beta_0} involves the joint density ratio of the \((p+1)\)-dimensional vector \(\boldsymbol{Z}\).
When \(p\) is large, directly estimating the ratio is generally infeasible due to the curse of dimensionality. 
Therefore, rather than estimating the high-dimensional joint density ratio, we exploit the quantile regression structure to construct a sample-level importance weight through one-dimensional densities,  inspired by \citep{zhao2026residual}. 


\begin{theorem}
\label{prop:RIW-QR}
Fix a quantile level \(\tau\in(0,1)\). For each source domain \(k=1,\ldots,K\), suppose that \(\varepsilon^{(k)}\) admits a probability density function \(p_{\varepsilon^{(k)}}(\cdot)\) in the $k$-th source domain. Let $g_\tau(\cdot)$ be a univariate probability density function satisfying $\int_{-\infty}^{0} g_\tau(u)\,du=\tau$ and $\int_{\mathbb{R}}\left|u\right|g_{\tau}(u)du<\infty$. Then the equation \eqref{hat_beta_0} holds with $\omega^{(k)}(\boldsymbol{Z}_i^{(k)})=g_\tau(Y_i^{(k)}-(\boldsymbol{ X}_i^{(k)})^\top\boldsymbol{\beta}^{(0)})/p_{\varepsilon^{(k)}}(Y_i^{(k)}-(\boldsymbol{X}_i^{(k)})^\top\boldsymbol{\beta}^{(k)})$.
Equivalently,
\begin{equation}
    \omega^{(k)}(\boldsymbol{Z}^{(k)}_i)=\frac{g_\tau(\varepsilon_i^{(k)}+\eta^{(k)}_i
)}{p_{\varepsilon^{(k)}}(\varepsilon_i^{(k)})}=\frac{g_\tau(r_i^{(k)}
)}{p_{\varepsilon^{(k)}}(\varepsilon_i^{(k)})},
\label{omega}
\end{equation}
where $\eta^{(k)}_i:=(\boldsymbol{X}_i^{(k)})^\top(\boldsymbol{\beta}^{(k)}-\boldsymbol{\beta}^{(0)})=(\boldsymbol{X}_i^{(k)})^\top\boldsymbol{\delta}^{(k)}$ and  $r_i^{(k)}:=\varepsilon_i^{(k)}+\eta_i^{(k)}$.
\end{theorem}

Theorem \ref{prop:RIW-QR} shows that the proposed importance weight can be expressed in terms of one-dimensional densities, thereby avoiding the direct estimation of a high-dimensional joint density ratio and alleviating the curse of dimensionality. Specifically, $\varepsilon_i^{(k)}$ is the error term under the $k$-th source model, whereas
$\eta_i^{(k)}$ represents the observation specific shift between the source and target
models. Thus, the proposed weight reflects the compatibility of each individual source observation with the target quantile model. 

Although the weight in \eqref{omega} involves only one-dimensional densities, it may still become unstable when its denominator is close to zero. 
To address this issue, we introduce a selection rule that ensures both $\varepsilon_i^{(k)}$ and $r_i^{(k)}$ in \eqref{omega} are uniformly bounded, from which uniformly bounded importance weights are derived. We first specify $g_\tau(\cdot)$ as
\begin{equation}\label{g_tau}
    g_\tau(u)=I\{-\tau A\leq u\leq(1-\tau)A\}/A,~A>0.
\end{equation}
For each source domain $k$, we define the sample selection subsets as
\begin{equation}
\mathcal{I}_k=\{i\in[n_k]:-\tau A\leq r_i^{(k)}\leq (1-\tau)A,\ 
|\eta_i^{(k)}|\leq M_k\},
\label{selected sample subset}
\end{equation}
where \(A>0\) and \(M_k>0\) are tuning parameters. Let $n_{\mathcal{I}_k}=|\mathcal{I}_k|$ denote the number of observations selected from the source domain $k$. We further define the pooled index set $\mathcal{I}=\{(k,i):k=1,\ldots,K,i\in\mathcal{I}_k\}$ with the total sample size $n_{\mathcal{I}}=\sum_{k=1}^Kn_{\mathcal{I}_k}$. 

Moreover, for any $i\in\mathcal{I}_k$, we can obtain $|\varepsilon_i^{(k)}|\leq|r_i^{(k)}|+|\eta_i^{(k)}|\leq A+M_k$.
Hence, both  $\varepsilon_i^{(k)}$ and  $r_i^{(k)}$  involved in \eqref{omega} are
restricted to bounded regions. This selection rule prevents excessively large importance weights under appropriate conditions as established in Theorem \ref{pkboundedwkbounded} in Section \ref{sec: Theoretical Analysis}. Since the selection is performed observation by observation, informative observations may still be retained even when their source domain is not globally transferable. In this way, SL-TL can recover useful information that may be discarded by Trans-SQR in Section \ref{Subsec:source-level}.

We next consider the estimation of $\boldsymbol{\beta}^{(0)}$ under two scenarios concerning the selected sample subsets $\{\mathcal{I}_k\}_{k=1}^{K}$ in \eqref{selected sample subset} and the weights $\{\omega^{(k)}(\cdot)\}_{k=1}^{K}$ in \eqref{omega}: (i) when they are assumed to be known, and (ii) when they are unknown and need to be estimated from the data. We discuss these two scenarios separately below.


\begin{itemize}
    \item  
    \textbf{Scenario I: when the selected subsets and weights are known.}
\end{itemize}
In this scenario, we consider an oracle setting where the subsets \(\mathcal{I}_k\) in \eqref{selected sample subset} and the weights \(\omega^{(k)}(\cdot)\) in \eqref{omega} are known for each $k=1,2,\dots,K$.
To unify the notation for the target and source samples,
let $\mathcal{I}_0=[n_0]$ and define
\begin{equation}
    \omega^{(0)}_i=1,i\in[n_0],~~
    \omega^{(k)}_i=\omega^{(k)}(\boldsymbol{Z}^{(k)}_i)=g_\tau(r_i^{(k)})/p_{\varepsilon^{(k)}}(\varepsilon_i^{(k)}),i\in[n_k].
\label{omegafinalgeneral}
\end{equation} 
We further define the total number of all observations by $N=n_0+\sum_{k=1}^{K}n_k$.
The oracle SL-TL estimator is defined as
\begin{align}
\hat{\boldsymbol{\beta}}_{\mathrm{ora}}^{(0)}\in
\arg\min_{\boldsymbol{\beta}\in\mathbb{R}^p}
\big\{
\sum_{k=0}^{K}
\sum_{i\in\mathcal{I}_k}
\omega_i^{(k)}
\rho_\tau(
Y_i^{(k)}
-
(\boldsymbol{X}_i^{(k)})^\top\boldsymbol{\beta}
)/N
+
\lambda_{\mathrm{ora}}
\|\boldsymbol{\beta}\|_1
\big\},
\label{eq:oracle-SL-TL}
\end{align}
where
\(\lambda_{\mathrm{ora}}>0\) is a  tuning parameter.

Thus, unlike source-level transfer methods such as Trans-SQR defined in Algorithm \ref{alg:Trans-SQR}, the oracle SL-TL estimator in \eqref{eq:oracle-SL-TL} transfers source information at the sample level by assigning an individual importance weight to each selected source observation.


\begin{itemize}
    \item  
    \textbf{Scenario II: when the selected subsets and weights are unknown.}
\end{itemize}

In practice, neither the oracle selected sample subsets $\{\mathcal{I}_k\}_{k=1}^{K}$ nor the importance weights $\{\omega^{(k)}(\cdot)\}_{k=1}^{K}$ are directly available, since their construction depends on the unknown target and source coefficients, as well as the unknown source error densities. 
Estimating these initial regression coefficients and error densities using the same observations may lead to overfitting and induce complicated dependence between the estimated weights and the weighted empirical loss. To separate these estimations, we adopt a sample splitting procedure. 
For each $k=0,1,\dots,K$, $\mathcal{D}^{(k)}=\{((\boldsymbol{X}_i^{(k)})^\top,Y_i^{(k)})^\top,i\in\mathcal{J}_{k}\}$, we randomly partition $[n_k]$ into three mutually disjoint index subsets $\mathcal{J}_{k,1}$, $\mathcal{J}_{k,2}$ and $\mathcal{J}_{k,3}$ with respective sizes $n_{k,1}$, $n_{k,2}$ and $n_{k,3}$, where $n_{k,1}=n_{k,2}=n_{k,3}=n_k/3$. The corresponding data subsets are indicated by $\mathcal{D}_r^{(k)}=\{((\boldsymbol{X}_i^{(k)})^\top,Y_i^{(k)})^\top,i\in\mathcal{J}_{k,r}\},~r=1,2,3$.

For the target domain $k=0$, two subsets are pooled to construct the initial target quantile regression estimator, while the remaining subset is reserved for the weighted estimation. For each source domain $k=1,2,\dots,K$, the three subsets are used, respectively, to construct the initial source quantile regression estimator, estimate the source error density, and estimate the importance weights used to select transferable source observations. 
We estimate the residual density function with estimated residual $\tilde\varepsilon_i^{(k)}$ in the $k$-th source domain by the B-spline projection estimator \citep{Gine2008AdaptiveEO} $\hat p_{\varepsilon^{(k)}}(\cdot)$  defined in \eqref{Bsplinehat-pk} of Section \ref{supp:Density estimation} in Supplementary Material. 
\begin{algorithm}[!htbp]
\caption{Cross-fitted SL-TL }
\label{alg:SL-TL}
\begin{algorithmic}[1] 

\REQUIRE Target dataset \(\mathcal D^{(0)}\), source datasets
\(\{\mathcal D^{(k)}\}_{k=1}^{K}\), quantile level \(\tau\),
resolution level \(\{j_k\}_{k=1}^{K}\), screening parameters
\(A,\{M_k\}_{k=1}^{K}\), bandwidth parameters $b_{k}$ and regularization parameters $\lambda_{k,\mathrm{ini}}$ and $\lambda$.

\STATE \textbf{Sample splitting:} 
For each $k=0,1,\dots,K$, uniformly at random partition $[n_k]$ into three mutually disjoint index subsets $\mathcal{J}_{k,1}$, $\mathcal{J}_{k,2}$ and $\mathcal{J}_{k,3}$ with $n_{k,r}=|\mathcal{J}_{k,r}|,~r=1,2,3$.

\STATE \textbf{Initial estimation:} Set $\mathcal{O}_{0,1}=\mathcal{J}_{0,1}\cup\mathcal{J}_{0,2}$ and $\mathcal{O}_{k,1}=\mathcal{J}_{k,1},~k=1,2,\ldots,K$. 
For $k=0,\ldots,K$, compute the initial quantile regression estimators based on $\mathcal{O}_{k,1}$ by
\begin{align}
    \tilde{\boldsymbol{\beta}}^{(k)}
\in
\arg\min_{\boldsymbol{\beta}\in\mathbb{R}^{p}}
\big\{&
\sum_{i\in\mathcal{O}_{k,1}}
\int_{-\infty}^{\infty}\rho_{\tau}(u)K((u+(\boldsymbol X_i^{(k)})^{\top}\boldsymbol\beta-Y_i^{(k)})/b_{k})\nonumber\\
&\quad /(|\mathcal{O}_{k,1}|b_{k})\,du
+
\lambda_{k,\mathrm{ini}}\|\boldsymbol{\beta}\|_1
\big\}.
\label{eqinitialbetak}
\end{align}

\STATE \textbf{Density estimation:} 
For each $k=1,2,\dots,K$ and $i\in\mathcal{J}_{k,2}$, compute $\tilde{\varepsilon}_i^{(k)}=Y_i^{(k)}-(\boldsymbol{X}_i^{(k)})^\top\tilde{\boldsymbol{\beta}}^{(k)}$. Then construct the B-spline projection estimator $\hat p_{\varepsilon^{(k)}}(\cdot)$ in \eqref{Bsplinehat-pk} by replacing $\mathcal{J}_{k}$ and $n_k$ with $\mathcal{J}_{k,2}$ and $n_{k,2}$, respectively.
\STATE \textbf{Weighting and screening:} For each $k=1,2,\dots,K$ and $i\in\mathcal{J}_{k,3}$, compute $\tilde{\varepsilon}_i^{(k)}=Y_i^{(k)}-(\boldsymbol{X}_i^{(k)})^\top\tilde{\boldsymbol{\beta}}^{(k)}$ and $\tilde\eta_i^{(k)}=(\boldsymbol{X}_i^{(k)})^\top(\tilde{\boldsymbol{\beta}}^{(k)}-\tilde{\boldsymbol{\beta}}^{(0)})$. Estimate the importance weights by
\begin{equation*}
    \hat\omega_i^{(k)}=g_{\tau}(\tilde\varepsilon_i^{(k)}+\tilde\eta_i^{(k)})/\hat p_{\varepsilon^{(k)}}(\tilde\varepsilon_i^{(k)}).
    \label{hatomegaSL-TL}
\end{equation*}
Define the estimated sample selection subset as
\begin{equation}
    \hat{\mathcal I}_{k,3}=\{i\in\mathcal{J}_{k,3}:-\tau A \leq \tilde{\varepsilon}_i^{(k)}+\tilde\eta_i^{(k)}\leq (1-\tau)A,~|\tilde\eta_i^{(k)}|\leq M_k\}.
    \label{hatIk3SL-TL}
\end{equation}

\STATE \textbf{Weighted estimation:} Set $\hat{\mathcal I}_{0,3}=\mathcal{J}_{0,3}$ and $\hat\omega_i^{(0)}=1,~i\in\mathcal{J}_{0,3}$. Compute
\begin{equation*}
\hat{\boldsymbol{\beta}}^{(0)}_1
=
\arg\min_{\boldsymbol{\beta}\in\mathbb{R}^p}
\big\{
\big[
\sum_{k=0}^{K}
\sum_{i\in\hat{\mathcal I}_{k,3}}
\hat\omega_i^{(k)}
\rho_\tau(
Y_i^{(k)}
-
(\boldsymbol{X}_i^{(k)})^\top\boldsymbol{\beta}
)/\sum_{k=0}^{K}n_{k,3}
\big]
+
\lambda
\|\boldsymbol{\beta}\|_1
\big\}.
\label{eq:RIW-QR}
\end{equation*}
\STATE \textbf{Cross-fitting:} For the target domain, cyclically use
$(\mathcal{J}_{0,1}\cup\mathcal{J}_{0,2},\mathcal{J}_{0,3})$,
$(\mathcal{J}_{0,2}\cup\mathcal{J}_{0,3},\mathcal{J}_{0,1})$, and
$(\mathcal{J}_{0,3}\cup\mathcal{J}_{0,1},\mathcal{J}_{0,2})$
for initial and weighted estimation, respectively. Meanwhile, for each source domain $k$, cyclically rotate the roles of $\mathcal{J}_{k,1},~\mathcal{J}_{k,2},~\mathcal{J}_{k,3}$. Repeat Steps 2-5 to get estimators $\hat{\boldsymbol{\beta}}^{(0)}_2$ and $\hat{\boldsymbol{\beta}}^{(0)}_3$.


\ENSURE The cross-fitted estimator
$\hat{\boldsymbol\beta}_{\mathrm{SL\text{-}TL}}^{(0)}=(\hat{\boldsymbol{\beta}}^{(0)}_1+\hat{\boldsymbol{\beta}}^{(0)}_2+\hat{\boldsymbol{\beta}}^{(0)}_3)/3$.
\algorithmeqnumber{betaSL-TLhat0}
\end{algorithmic}
\end{algorithm}

Although sample splitting reduces overfitting and simplifies the analysis of the estimated weights, a single split uses only one subset for the final coefficient estimation and may be sensitive to the particular random partition. To improve data utilization and reduce split-induced variability, we employ cross-fitting. Specifically, for the target domain, two subsets are used for initial estimation and the remaining subset for weighted estimation, with their roles cyclically rotated across the three folds. For each source domain, the roles of the three subsets are cyclically rotated among initial estimation, density estimation, and weighting and screening. The three resulting estimators are then averaged to obtain the final cross-fitted estimator. In this way, every observation contributes to the final estimation stage in one of the three rotations.
The resulting estimation procedure of SL-TL is summarized in Algorithm \ref{alg:SL-TL}.
\section{Theoretical Analysis} \label{sec: Theoretical Analysis}
In this section, we investigate the convergence rates of the proposed estimators $\hat{\boldsymbol\beta}_{\mathrm{ora}}^{(0)}$ in \eqref{eq:oracle-SL-TL} and $\hat{\boldsymbol\beta}_{\mathrm{SL\text{-}TL}}^{(0)}$ in \eqref{betaSL-TLhat0}. The theoretical results are presented in Sections \ref{sub3.1} and \ref{sub3.2}, respectively. Throughout this section, we define $\rho_{\mathcal{I}}:=(n_0+\mathbb{E}(n_{\mathcal{I}}))/\sum_{k=0}^{K}n_k$, $\Psi_n:=\max_{0\le k\le K}\max_{i\in\mathcal J_{k}}\|\bm X_i^{(k)}\|_\infty$, and $\pi_{M_k}:=\mathbb{E}(n_{\mathcal{I}_k})/\mathbb{E}(n_{\mathcal{I}})$, $h_k=\|\boldsymbol{\beta}^{(k)}-\boldsymbol{\beta}^{(0)}\|_1$ for each $k=1,2,\ldots,K$ and let $F_{\varepsilon^{(k)}|\boldsymbol{X}^{(k)}}(t\mid \boldsymbol{X}_i^{(k)})=\mathbb{P}{(\varepsilon_i^{(k)}\leq t|\boldsymbol{X}_i^{(k)})}$ denote the conditional distribution function of $\varepsilon_i^{(k)}$.


\subsection{When the selected subsets and weights are known}
\label{sub3.1}
In this subsection, we study the theoretical properties of the oracle SL-TL estimator $\hat{\boldsymbol{\beta}}_{\mathrm{ora}}^{(0)}$ in \eqref{eq:oracle-SL-TL}. 
We first introduce some technical assumptions for the subsequent analysis. We then establish the boundedness of the oracle importance weights, characterize the effective sample size contributed by each source domain, and derive the convergence rate of the oracle SL-TL estimator.  

\begin{assumption}
    \label{pkbound}
    For each $k=1,2,\ldots,K$, the probability error density function $p_{\varepsilon^{(k)}}(t)$ is strictly positive for all finite $t$ and has a bounded first derivative.
\end{assumption}

\begin{assumption}
    \label{subGaussian}
For each $k=0,1,\ldots,K$, the errors $\{\varepsilon_i^{(k)}\}_{i=1}^{n_k}$ are independent and identically distributed within the $k$-th domain. Moreover, there exist constants $\sigma_k>0$ such that $\mathbb{E}[\exp\{u(\varepsilon_i^{(k)}-\mathbb{E}(\varepsilon_i^{(k)}))\}]\leq\exp[u^2(\sigma_{k})^2/2]$ for all $u\in\mathbb{R}$.
\end{assumption}

\begin{assumption}
    \label{kappaF}
    For each $k=0,1,\ldots,K$, there exist constants $\xi>0$ and $\kappa>0$, independent of $k$, such that for all $\boldsymbol{X}_i^{(k)}$ and any $|t|\leq\xi$, there holds $|F_{\varepsilon^{(k)}|\boldsymbol{X}^{(k)}}(t|\boldsymbol{X}_i^{(k)})-F_{\varepsilon^{(k)}|\boldsymbol{X}^{(k)}}(0|\boldsymbol{X}_i^{(k)})|\geq\kappa|t|$.
\end{assumption}

\begin{assumption}
    \label{Gaussian maximal inequalit}
    Let $H_0=\mathrm{supp}(\boldsymbol{\beta}^{(0)})$ denote the support set of $\boldsymbol{\beta}^{(0)}$ and define the cone $\mathcal{E}(H_0,3)=\{\boldsymbol{v}\in\mathbb{R}^{p}:\|\boldsymbol{v}_{H_{0}^{c}}\|_{1}\leq3\|\boldsymbol{v}_{H_{0}}\|_{1}\}$. There exists a constant $\phi>0$ such that
        \begin{equation*}
        \inf_{\boldsymbol{v}\in\mathcal{E}(H_0,3)\backslash\{\mathbf{0}_{p}\}}\boldsymbol{v}^\top\Big[(\mathbf{X}^{(0)})^\top\mathbf{X}^{(0)}+\sum_{k=1}^K(\mathbf{X}^{(k)})^\top\mathbf{I}_{\mathcal{I}_k}\mathbf{X}^{(k)}\Big]\boldsymbol{v}\big/((n_0+n_\mathcal{I})\left\|\boldsymbol{v}\right\|^2)\geq\phi^2,
    \end{equation*}
    where $\mathbf{I}_{\mathcal{I}_k}$ is the diagonal matrix 
    with the $i$-th diagonal element being $\mathds{1}(i\in\mathcal{I}_k)$.
\end{assumption}
Assumption \ref{pkbound} controls the smoothness and lower bound of error densities, while Assumption \ref{subGaussian} controls the error tails. Assumptions \ref{kappaF} and \ref{Gaussian maximal inequalit} provide local curvature conditions ensuring the identifiability and stability of the parameter estimation. 
To impose a smoothness condition on $p_{\varepsilon^{(k)}}(\cdot)$, we introduce the following Hölder class.
\begin{definition}[H\"older class]
\label{def:holder-class}
Let $t=q+\alpha$, where $q\in\mathbb{N}_0$ and $\alpha\in(0,1]$. For a constant $B>0$, the H\"older class $\mathcal H^t(\mathbb R,B)$ on $\mathbb{R}$ is defined as 
\begin{equation} \label{holder}
\mathcal{H}^t(\mathbb R,B) := \big\{ f: \mathbb{R} \to \mathbb{R} \,\big|\, \max_{0\leq \ell\leq q}
\|f^{(\ell)}\|_\infty
\leq B,
\sup_{\substack{u,v\in\mathbb{R}\\u\neq v}}
\frac{
|f^{(q)}(u)-f^{(q)}(v)|
}{
|u-v|^\alpha
}
\leq B \big\}.
\end{equation}
\end{definition}
Throughout the paper, we assume that the error density function
$p_{\varepsilon^{(k)}}(\cdot)$ belongs to $\mathcal H^t(\mathbb R,B)$ for some $t>0$
and $B>0$. 
We then establish the boundedness of the oracle importance weights and characterize the effective sample size in Theorems \ref{pkboundedwkbounded} and \ref{effectivesamplesize}, respectively.

\begin{theorem}
    \label{pkboundedwkbounded}
    Suppose that Assumption \ref{pkbound} holds and $A>0,~M_k>0,~k=1,2,\ldots,K$ and $K$ are fixed. Then, for each $k=1,2,\dots,K$ and for all $i \in \mathcal{I}_{k}$ in \eqref{selected sample subset}, the error density function $p_{\varepsilon^{(k)}}(\varepsilon^{(k)}_i)$ is strictly bounded away from 0. Moreover, there exist positive constants $c_\omega$ and $C_\omega$, independent of $i$ and $k$, such that $0<c_\omega\leq \omega^{(k)}(\boldsymbol{Z}_i^{(k)})\leq C_\omega<\infty$.
\end{theorem}
Theorem~\ref{pkboundedwkbounded} shows that, as a consequence of the sample selection procedure, the oracle importance weights are uniformly bounded away from both zero and infinity over the selected source observations.
In particular, the upper bound prevents any selected source observation from having an excessively large contribution to the weighted quantile loss,
whereas the lower bound ensures that every selected source observation makes a non-negligible contribution to the weighted loss rather than being effectively discarded.
This uniform boundedness facilitates the control of the weighted empirical loss in the subsequent convergence analysis.

\begin{theorem}
    \label{effectivesamplesize}
    Suppose that for each $k=1,2,\dots,K$, the covariates $\boldsymbol{X}^{(k)}$ are independently drawn from a multivariate normal distribution $N_p(\boldsymbol{0},\Sigma^{(k)})$, where the eigenvalues of $\Sigma^{(k)}$ satisfy $\zeta_l\leq\lambda_{\min}(\Sigma^{(k)})\leq\lambda_{\max}(\Sigma^{(k)})\leq\zeta_u$ for some positive constants $0<\zeta_l\leq\zeta_u<\infty$ independent of $k$. Suppose that Assumptions \ref{pkbound} and \ref{subGaussian} hold and that the screening parameters $A>0$ and $M_k>0$ satisfy $\min\{\tau,1-\tau\}A-M_k\geq c_0$ for some $c_0>0$. Let $d_k=\|\boldsymbol{\delta}^{(k)}\|_2$. Then,
    \begin{equation*}
        \mathbb{E}(n_{\mathcal{I}_k})\asymp n_k\{(M_{k} / d_k)\wedge 1\},
    \end{equation*}
    where $n_k\{(M_{k}/d_k)\wedge1\}:=n_k$ when $d_k=0$. In particular, when $M_k\geq d_k$, we have $\mathbb{E}(n_{\mathcal{I}_k})\asymp n_k$.
\end{theorem}
Theorem~\ref{effectivesamplesize} shows that the effective sample size
$\mathbb{E}(n_{\mathcal I_k})$ contributed by the source $k$ is determined
by the relative magnitudes of the selection parameter $M_k$ and the
source discrepancy $d_k$. If $M_k\geq d_k$, the effective sample size
from source $k$ is of the same order as its original sample size
$n_k$. In contrast, when $M_k<d_k$, the effective sample size is of the order $n_kM_k/d_k$.
Therefore, for a fixed $M_k$, a source domain contributes fewer effective observations as its discrepancy from the target domain increases.
We next study the convergence rate of the oracle SL-TL estimator $\hat{\boldsymbol{\beta}}_{\mathrm{ora}}^{(0)}$ in \eqref{eq:oracle-SL-TL}.
\begin{theorem}
    \label{oracletheorem}
    Suppose that Assumptions \ref{pkbound}-\ref{Gaussian maximal inequalit} are satisfied and $K$, $\{\boldsymbol{X}_i^{(k)}\}_{i=1}^{n_k}$ are fixed. Assume that $\mathcal{I}_{k}$ in \eqref{selected sample subset} and $\omega_{i}^{(k)}$ in \eqref{omega} are known. If the regularization parameter satisfies $\lambda_{\mathrm{ora}}\asymp\rho_{\mathcal{I}}\sqrt{\log p/(n_0+\mathbb{E}(n_{\mathcal{I}}))}$, and $\Psi_ns\sqrt{\log p/(n_0+\mathbb{E}(n_{\mathcal{I}}))}=o(1)$, provided that $\Psi_n=O(1)$. 
    Suppose $n_0+\mathbb{E}(n_{\mathcal{I}})\to\infty$, then the oracle estimator satisfies
    \begin{center}
        $\parallel\hat{\boldsymbol{\beta}}_{\mathrm{ora}}^{(0)}-\boldsymbol{\beta}^{(0)}\parallel^2_2=O_p(s\log p/(n_0+\mathbb{E}(n_\mathcal{I}))).$
    \end{center}
\end{theorem}
Theorem~\ref{oracletheorem} shows that the squared $\ell_2$-error of the oracle SL-TL estimator scales with the total effective sample size $n_0+\mathbb E(n_{\mathcal I})$, rather than with the target sample size $n_0$ alone. 
We next compare its squared $\ell_2$-convergence rate with those of the Target-only estimator $\hat{\boldsymbol\beta}_{\mathrm T}^{(0)}$ and the oracle Trans-SQR estimator $\hat{\boldsymbol\beta}_{\mathrm S}$. Consider a single source domain setting with sample size $n_1$ satisfying $\|\boldsymbol{\delta}^{(1)}\|_1\leq h$ so that it is included by the oracle Trans-SQR. 

For brevity, let $N_{\mathrm{SL}}=n_0+\mathbb E(n_{\mathcal I_1})$ and $N_{\mathrm S}=n_0+n_1$. Note that the squared $\ell_2$-convergence rates of the Target-only, oracle Trans-SQR, and oracle SL-TL estimators are $r_{\mathrm{T}}:=s\log p/{n_0}$, $r_{\mathrm{S}}\asymp s\log p/\sqrt{n_0+n_1}+h\sqrt{\log p/n_0}$, and $r_{\mathrm{SL}}:=s\log p/{N_{\mathrm{SL}}}$, respectively. 
To study how the rate comparison changes as the source-target discrepancy level $h$ increases, we define two benchmark thresholds $h_a=s\sqrt{\log p/N_{\mathrm{S}}}$ and $h_b=s\sqrt{\log p/n_0}$. Under fixed $M_1$ and $h_b=o_p(1)$, the rate comparison is summarized in Table \ref{tab:rate-comparison-single-source}.  
\begin{table}[!htbp]
\centering
\caption{Comparison of the squared $\ell_2$-convergence rates of the Target-only, oracle Trans-SQR, and oracle SL-TL estimators in the single source setting.}
\label{tab:rate-comparison-single-source}\par
\renewcommand{\arraystretch}{1.6}
\resizebox{\linewidth}{!}{\begin{tabular}{c r r c c} 
\toprule
 & \multicolumn{1}{c}{Target-only} & \multicolumn{1}{c}{Oracle Trans-SQR} & \multicolumn{1}{c}{Oracle SL-TL} & \multicolumn{1}{c}{$N_{\mathrm{SL}}$} \\
\midrule
$h\lesssim h_a$ & $s\log p/n_0$ & $s\log p/\sqrt{n_0N_\mathrm{S}}$ & $s\log p/N_\mathrm{S}$ & $N_\mathrm{SL}\asymp N_\mathrm{S}$ \\
$h_a\ll h\ll h_b$   & $s\log p/n_0$ & $h\sqrt{\log p/n_0}$ & $s\log p/N_\mathrm{S}$ & $N_\mathrm{SL}\asymp N_\mathrm{S}$\\
$h\gtrsim h_b$ & $s\log p/n_0$ & $h\sqrt{\log p/n_0}$ & $s\log p/(n_0+n_1\{(M_1/d_1)\wedge1\})$ & $N_\mathrm{SL}\asymp n_0+n_1\{(M_1/d_1)\wedge1\}$\\
\bottomrule
\end{tabular}}
\end{table}

As shown in Table \ref{tab:rate-comparison-single-source}, since $h_b=o_p(1)$ and $d_1=\|\boldsymbol{\delta}^{(1)}\|_2\leq\|\boldsymbol{\delta}^{(1)}\|_1\leq h$, we have $N_\mathrm{SL}\asymp N_\mathrm{S}$ when $h\lesssim h_a$ or $h_a\ll h\ll h_b$. 
The oracle SL-TL estimator attains the smallest rate across all regimes.  
When $h\gtrsim h_b$, oracle Trans-SQR is no better than Target-only, while oracle SL-TL attains $s\log p/(n_0+n_1\{(M_1/d_1)\wedge1\})$ and remains no worse than Target-only since $N_{\mathrm{SL}}\gtrsim n_0$.

\subsection{When the selected subsets and weights are unknown}
\label{sub3.2}
We next investigate the theoretical properties of the cross-fitted SL-TL estimator $\hat{\boldsymbol{\beta}}_{\mathrm{SL}\text{-}\mathrm{TL}}^{(0)}$ in \eqref{betaSL-TLhat0}. 
We first establish Lemmas \ref{lemma:selection-consistency} and \ref{lemma hat p_k}, which provide the error bounds arising from the estimation of the 
selected sample subsets and densities, respectively.
Based on these results, we then derive the convergence rate of the proposed estimator $\hat{\boldsymbol\beta}_{\mathrm{SL\text{-}TL}}^{(0)}$. We begin with an assumption on the initial estimators for the subsequent theoretical analysis.

\begin{assumption}
    \label{initial beta}
Suppose that \(K\) is fixed. For each $k=0,1,\cdots,K$, the initial estimator $\tilde{\boldsymbol{\beta}}^{(k)}$ in \eqref{eqinitialbetak} satisfies $\|\tilde{\boldsymbol{\beta}}^{(k)}-\boldsymbol{\beta}^{(k)}\|_1
=O_p(\gamma_k)$ with $\gamma_k=o(1)$.
\end{assumption}
Assumption \ref{initial beta} requires the initial estimators to be consistent in the $\ell_1$ norm. 
In Algorithm \ref{alg:SL-TL}, $\tilde {\boldsymbol{\beta}}^{(k)}$ is constructed using the $\ell_1$-penalized SQR estimator proposed by \citet{tan2022high}. Under standard sparsity and regularity conditions, its $\ell_1$-estimation error typically satisfies $\|\tilde{\boldsymbol{\beta}}^{(k)}-\boldsymbol{\beta}^{(k)}\|_1
=O_p(s_k\sqrt{\log p/n_k})$ where $s_k=\|\boldsymbol{\beta}^{(k)}\|_0$. Therefore, Assumption \ref{initial beta} holds whenever $s_k\sqrt{\log p/n_k}=o(1)$.


Let $\mathcal X=\{\bm X_i^{(k)}:i\in\mathcal J_k,\;k=0,\ldots,K\}$ denote the collection of design vectors, and set $L_\tau=-\tau A$ and $U_\tau=(1-\tau)A$. For $k=1,\ldots,K$ and $j=1,2,3$, Algorithm \ref{alg:SL-TL} yields $\hat{\mathcal I}_{k,j} = \{ i\in\mathcal J_{k,j}: L_{\tau} \le\tilde\varepsilon_i^{(k)} + \tilde\eta_i^{(k)} \le U_{\tau},~| \tilde\eta_i^{(k)}|  \le M_k \}$.
Since $r_i^{(k)}=\varepsilon_i^{(k)}+\eta_i^{(k)}$, we define the inner and outer perturbed oracle subsets by $ \mathcal{I}_{k,j}^{-}=\{i\in\mathcal{J}_{k,j}:L_\tau+a_n\le r_i^{(k)}\le U_\tau-a_n,~|\eta_i^{(k)}|\le M_k-a_n\}$ and $\mathcal{I}_{k,j}^{+}=\{i\in\mathcal{J}_{k,j}:L_\tau-a_n\le r_i^{(k)}\le U_\tau+a_n,~|\eta_i^{(k)}|\le M_k+a_n\}$, respectively, 
where $a_n=o(1)$ is defined in Lemma \ref{lem:S1} of Section \ref{supp:lemma S1} in Supplementary Material.

\begin{lemma}
\label{lemma:selection-consistency}
Suppose that the assumptions of Lemma \ref{lem:S1} and \ref{lem:S2} in Supplementary Material hold. 
Define
$\pi_k
=
\sum_{i\in\mathcal J_k}
\mathbb P
\{
L_\tau
\le r_i^{(k)}
\le U_\tau,\;
|\eta_i^{(k)}|
\le M_k
\}/n_k$ and suppose that $\min_{1\le k\le K}\pi_k
\ge c_\pi>0.$ For each \(k=1,\ldots,K\) and \(j=1,2,3\), we have $
    \mathbb{P}(\mathcal I_{k,j}^{-} \subseteq \hat{\mathcal I}_{k,j} \subseteq \mathcal I_{k,j}^{+})\xrightarrow{}1 
    $ 
and $\mathbb{E}|\mathcal{I}_{k,j}^{-}|
\asymp
\mathbb{E}|\hat{\mathcal{I}}_{k,j}|
\asymp
\mathbb{E}|\mathcal{I}_{k,j}|
\asymp
\mathbb{E}|\mathcal{I}_{k,j}^{+}|.$ 
If, in addition, $n_{k,j}\pi_k\to\infty$, then
$|\hat{\mathcal{I}}_{k,j}|/({n_{k,j}\pi_k})
\overset{p}{\longrightarrow}1.$
\end{lemma}
Lemma \ref{lemma:selection-consistency} shows that the estimated sample selection subset $\hat{\mathcal I}_{k,j}$ provides a good approximation to the oracle subset $\mathcal{I}_{k,j}$. Specifically, with probability tending to one, it is sandwiched between two slightly perturbed oracle subsets, and 
its expected cardinality is asymptotically equivalent to that of the
oracle subset $\mathcal I_{k,j}$. Therefore, estimating the sample selection subset has an asymptotically negligible effect on the effective source sample size.
Having controlled the discrepancy arising from the estimated sample selection subset, we now examine the additional error induced by density estimation. The following lemma establishes a uniform convergence rate for the B-spline projection estimator $\hat p_{\varepsilon^{(k)}}(\cdot)$ in \eqref{Bsplinehat-pk}.
\begin{assumption}
For each $j=1,2,3$, the restricted eigenvalue condition in Assumption \ref{Gaussian maximal inequalit} holds for the $j$-th fold, with the target design matrix and $n_0$ replaced by their counterparts, $\mathcal I_k$ replaced by $\mathcal I_{k,j}^{-}$ for all $k$, and $n_{\mathcal I}$ replaced by $\sum_{k=1}^{K}|\mathcal I_{k,j}^{-}|$.
\label{assumptionIjk-}
\end{assumption}

Assumption \ref{assumptionIjk-} ensures sufficient restricted curvature required to control the $\ell_2$ estimation error of the cross-fitted estimator.

\begin{lemma}
\label{lemma hat p_k}
Suppose that Assumptions \ref{pkbound}, \ref{subGaussian} and \ref{initial beta} are satisfied. For each $k=1,2,\ldots,K$, assume that the error density function
$p_{\varepsilon^{(k)}}(\cdot)$ of
$\varepsilon^{(k)}$ belongs to
$\mathcal H^t(\mathbb R,B)$ defined in \eqref{holder}, for some $0<t\leq r$ and $B>0$,
where $r$ is the order of the B-spline basis. Recall that $\Psi_n=\max_{1\le k\le K}\max_{i\in\mathcal D^{(k)}}\|\bm X_i^{(k)}\|_\infty$. Suppose that $r\ge 2$, $j_k2^{j_k}/n_k=o(1)$, and $2^{2j_k}\Psi_n\gamma_k=o(1)$. Then, the density estimator $\hat p_{\varepsilon^{(k)}}(\cdot)$ in \eqref{Bsplinehat-pk} holds that
\begin{equation*}
    \sup_{u\in\mathbb R}
|
\hat p_{\varepsilon^{(k)}}(u)
-
p_{\varepsilon^{(k)}}(u)
|
=
O_p(
2^{2j_k}\Psi_n\gamma_k
+
2^{-j_kt}
+
(j_k2^{j_k}/{n_k}
)^{1/2}
).
\end{equation*}
Moreover, choosing $2^{j_k}
\asymp
(
n_k/{\log n_k}
)^{1/(2t+1)}$, gives
\begin{equation*}
 \sup_{u\in\mathbb R}
|
\hat p_{\varepsilon^{(k)}}(u)
-
p_{\varepsilon^{(k)}}(u)
|
=
O_p(
2^{2j_k}\Psi_n\gamma_k
+
(
\log n_k/{n_k}
)^{t/(2t+1)}
)
=O_p(q_k)
=o_p(1),
\end{equation*}
where $q_k=u_k+v_k$ with $u_k=2^{2j_k}\Psi_n\gamma_k$ and $v_k=(
\log n_k/{n_k}
)^{t/(2t+1)}$.
\end{lemma}
The error bound has two components under the above choice of the resolution parameter. Recalling the definition of the estimator $\hat p_{\varepsilon^{(k)}}(\cdot)$ in \eqref{Bsplinehat-pk}, the term $u_k$ is the plug-in error arising from replacing the unobserved true errors $\varepsilon_i^{(k)}$ with their estimated residuals $\tilde\varepsilon_i^{(k)}$. This term depends on both the initial coefficient estimation error and the resolution parameter of the B-spline basis. The term $v_k$ is the intrinsic uniform error of the B-spline projection density estimator based on the true errors.

\begin{remark}
The uniform error in Lemma \ref{lemma hat p_k} contains both a
plug-in error and an intrinsic error. Let
$a_k=\log n_k/{n_k}.$
Under the choice
$2^{j_k}\asymp a_k^{-1/(2t+1)}$,
the B-spline projection estimator satisfies
$q_k
=
\Psi_n\gamma_k a_k^{-2/(2t+1)}
+
a_k^{t/(2t+1)}
.$
For comparison, the second-order kernel density estimator employed in \citet{zhao2026residual} has the uniform error bound $O_p(q_k^{\mathrm{K}})$, where 
$q_k^{\mathrm{K}}
=
\Psi_n\gamma_k/{w_k^2}
+
w_k^2
+
(\log n_k/{(n_kw_k)})^{1/2}
.$
Under the conventional bandwidth choice
$w_k\asymp a_k^{1/5},$
this becomes
$q_k^{\mathrm{K}}
=
\Psi_n\gamma_k a_k^{-2/5}
+
a_k^{2/5}
.$
For $t>2$, both the two components of $q_k$ are of strictly smaller order than their counterparts in $q_k^{\mathrm{K}}$.
Therefore, under these tuning choices, the B-spline projection estimator achieves a
faster convergence rate when $t>2$ and $r\ge t$, and has the same order when $t=2$.
\end{remark}

The density estimation error characterized above directly affects the accuracy of the estimated importance weights and the performance of the resulting transfer estimator. Building on the above results, we next establish the convergence rate of the proposed estimator.

\begin{theorem}
    \label{estimatebetabound}
    Suppose that Assumptions \ref{pkbound}\text{-}\ref{assumptionIjk-} are satisfied, the conditions of Lemmas \ref{lemma:selection-consistency}-\ref{lemma hat p_k} are satisfied, and $\{\boldsymbol{X}_i^{(k)}\}_{i=1}^{n_k}$ are fixed. If the regularization parameter satisfies $\lambda=2(\lambda_1+\lambda_2)$ with $\lambda_1\asymp\Psi_n\rho_{\mathcal{I}}\sqrt{\log p/(n_{0}+\mathbb{E}(n_{\mathcal{I}}))}$ and $\lambda_2\asymp\Psi_n\rho_{\mathcal{I}}\sum_{k=1}^{K}\pi_{M_k}[(\Psi_n\gamma_0)\wedge(M_k+\Psi_n(\gamma_k+h_k))+q_k]$.
    Provided that $\Psi_ns(\lambda_1+\lambda_2)/\rho_{\mathcal{I}}=o_p(1)$, 
    as $\min_{0\leq k\leq K}n_k\to\infty$, then the cross-fitted SL-TL estimator $\hat{\boldsymbol\beta}_{\mathrm{SL\text{-}TL}}^{(0)}$ in \eqref{betaSL-TLhat0} satisfies
\begin{align*}
    \|\hat{\boldsymbol\beta}_{\mathrm{SL\text{-}TL}}^{(0)}-\boldsymbol{\beta}^{(0)}\|^2_2=O_p\big\{&\frac{s\mathrm{log}p}{n_0+\mathbb{E}(n_\mathcal{I})}+s\Psi_n^2(\sum_{k=1}^{K}\pi_{M_k}[(\Psi_n\gamma_0)\wedge(M_k+\Psi_n(\gamma_k+h_k))+q_k])^2\big\}.
\end{align*}
Moreover, assume the following conditions: (i) For $k=1,2,\ldots,K$, $n_k\gg n_0$, such that $\gamma_k\ll\gamma_0$ and $q_k\ll\gamma_0$, (ii) $\Psi_n=O_p(1)$, and (iii) $M_k,k=1,2,\ldots,K$ are fixed. It has the following
\begin{equation*}
    \|\hat{\boldsymbol\beta}_{\mathrm{SL\text{-}TL}}^{(0)}-\boldsymbol{\beta}^{(0)}\|^2_2=O_p\big\{\frac{s\mathrm{log}p}{n_0+\mathbb{E}(n_\mathcal{I})}+s(\sum_{k=1}^{K}\pi_{M_k}[\gamma_0\wedge(M_k+h_k)])^2\big\}.
\end{equation*}
\end{theorem}
Theorem \ref{estimatebetabound} decomposes the squared $\ell_2$-estimation error of the estimator $\hat{\boldsymbol\beta}_{\mathrm{SL\text{-}TL}}^{(0)}$ into an oracle term 
and an additional error due to estimating the quantities required for feasible transfer.
The oracle term $s\log p/(n_0+\mathbb{E}(n_{\mathcal{I}}))$ can be substantially smaller than the Target-only rate $s\log p/n_0$ when sufficiently many informative source observations are retained. The additional term depends on the source contribution $\pi_{M_k}$, the source-target discrepancy $h_k$, the initial estimation error $\gamma_k$ and the density estimation error $q_k$.
Under the additional conditions of Theorem \ref{estimatebetabound}, the effects of $\gamma_k$ and $q_k$ are asymptotically negligible, so the additional term reduces to $s(\sum_{k=1}^{K}\pi_{M_k}[\gamma_0\wedge(M_k+h_k)])^2$.

\section{Numerical Studies}
\label{sec: Numerical Studies}
\subsection{Simulations}



In this section, we conduct Monte Carlo simulations to evaluate the performance of SL-TL for high-dimensional quantile regression and compare it with Trans-SQR \citep{zhang2022transfer}, a source-level transfer learning method. For both approaches, we consider both an oracle setting and a feasible setting. In the oracle setting, the transferable source set is assumed known for Trans-SQR, whereas the selected source subsets and importance weights are known for SL-TL. In the feasible setting, these quantities are estimated from the data. We denote the oracle and feasible procedures by O-Trans-SQR, O-SL-TL, F-Trans-SQR, and F-SL-TL, respectively. The Target-only $\ell_1$-penalized estimator in \eqref{hat_beta_T} is included as a benchmark. The same data generating mechanism is used throughout the simulations.

We consider one target domain ($k=0$) and $K$ source domains ($k=1,\ldots,K$), where $K\in\{1,10,20\}$. The target and source sample sizes are set to $n_0=300$ and $n_k=450$, respectively. We set the covariate dimension to $p=500$ and the sparsity level to $s=16$. The results are based on 100 Monte Carlo replications for each combination of $K$ and $\tau$.

For the target domain and each source domain, the response is generated from the linear quantile regression model
\begin{equation*}
    Y_i^{(k)}=(\boldsymbol{X}_i^{(k)})^\top \boldsymbol{\beta}^{(k)}+\varepsilon_i^{(k)},
    \quad k=0,1,\ldots,K,
\end{equation*}
where $\mathbb{P}(\varepsilon_i^{(k)}\leq 0\mid \boldsymbol{X}_i^{(k)})=\tau$. The target covariates are generated as $\boldsymbol{X}_i^{(0)}\operatorname*{\sim}N_p(\boldsymbol{0}_p,\Sigma)$, where $\Sigma_{jl}=0.7^{|j-l|},~1\leq j,l\leq p$. For each source domain $k$, we generate $\mu_k\operatorname*{\sim}U(0.3,0.5)$ and set $\boldsymbol{X}_i^{(k)}\operatorname*{\sim}N_p(\boldsymbol{0}_p,\Sigma^{(k)})$, where $\Sigma^{(k)}=\Sigma+\boldsymbol v^{(k)}{(\boldsymbol v^{(k)})}^\top/p$ and $\boldsymbol v^{(k)}\sim N_p(\boldsymbol{0}_p,\mu_k^2\boldsymbol{I}_p)$.

The target coefficient vector is sparse and defined as $\boldsymbol{\beta}^{(0)}=(\boldsymbol{1}^\top_s, \boldsymbol{0}^\top_{p-s})^\top$. For each source domain $k$, recall that the coefficient discrepancy vector is defined as $\boldsymbol{\delta}^{(k)}=\boldsymbol{\beta}^{(k)}-\boldsymbol{\beta}^{(0)}$. We randomly select two index sets, $T_{k,1}\subseteq[s]$ and $T_{k,2}\subseteq\{s+1,\ldots,p\}$ with $|T_{k,1}|=|T_{k,2}|=6$ and let $T_k=T_{k,1}\cup T_{k,2}$. We further generate independent Rademacher random variables satisfying $P(\sigma_j^{(k)}=1)=P(\sigma_j^{(k)}=-1)=1/2$. The entries of the coefficient discrepancy vector are defined as
\begin{equation*}
    {\delta}_j^{(k)}=\Delta_k\sigma_j^{(k)}I\{j\in T_k\}/|T_k|,\quad j=1,2,\ldots,p,
\end{equation*}
 where $\Delta_k$ are independently drawn from $U(0.1,4)$ and the source coefficient vector is then given by $\boldsymbol{\beta}^{(k)}=\boldsymbol{\beta}^{(0)}+\boldsymbol{\delta}^{(k)}$. By construction, $\boldsymbol{\delta}^{(k)}$ has exactly 12 nonzero entries and satisfies $\|\boldsymbol{\delta}^{(k)}\|_1=\Delta_k$. To ensure that zero is the conditional $\tau$-th quantile of the error distribution, we generate $U_i^{(k)}~\overset{\mathrm{iid}}{\operatorname*{\operatorname*{\sim}}}~N(0,1)$ and set $\varepsilon_i^{(k)}=U_i^{(k)}-\Phi^{-1}(\tau),~k=0,1,\ldots,K$, where $\Phi(\cdot)$ denotes the standard normal distribution function. Since $U_i^{(k)}$ is independent of $\boldsymbol{X}_i^{(k)}$, then $P(\varepsilon_i^{(k)}\leq0\mid \boldsymbol{X}_i^{(k)})=P(U_i^{(k)}\leq\Phi^{-1}(\tau))=\tau$. Therefore, $Q_\tau(Y_i^{(k)}\mid \boldsymbol{X}_i^{(k)})=(\boldsymbol{X}_i^{(k)})^\top\boldsymbol{\beta}^{(k)}$.

For a given value of $K$, the domain specific quantities $\{\mu_k,\Delta_k\}_{k=1}^K$ are generated once and kept fixed throughout the 100 Monte Carlo replications, while all other random quantities are regenerated in each replication. 
All regularization parameters are selected by five-fold cross-validation. For estimators based on convolution-smoothed quantile loss, the smoothing bandwidth is set to $\max\{0.05,\sqrt{\tau(1-\tau)}\left(\log p/n\right)^{1/4}\}$, where $n$ denotes the actual sample size used in the corresponding fit. 
For both O-SL-TL and F-SL-TL, the screening parameters in \eqref{selected sample subset} and \eqref{hatIk3SL-TL} are selected by grid search at each quantile level $\tau$, with $A\in\{1.1,1.2,\ldots,3.4\}$ and $M_k=M\in\{0.4,0.5,\ldots,0.7\}$ for all $k$, subject to the condition $\min (\tau,1-\tau)A>M$ in Section \ref{sup:S2.2} of Supplementary Material. 
The admissible pair with the best average performance over the Monte Carlo replications is selected.
We next report the numerical results for the oracle and feasible settings separately.

\subsubsection{Oracle setting}

In the oracle setting, we compare Target-only, O-Trans-SQR, and O-SL-TL, where the corresponding transferable information is assumed to be known. The data generating mechanism and the screening parameters are specified as described above. We next specify the implementation details for the transfer learning methods.  

We evaluate the three estimators at the quantile levels $\tau\in\{0.2,0.5,0.7\}$ with the number of source domains set to $K\in\{1,10,20\}$.
When $K=1$, we investigate the impact of the discrepancy between the source and target domains by considering three different levels of the source-target coefficient discrepancy $\Delta_1=\|\boldsymbol \delta^{(1)}\|_1$ over $\{0.5,2.5,7.0\}$. These settings represent increasing degrees of heterogeneity between the source and target domains. When $K=10$ or $20$, we consider the transferability thresholds $h\in\{0.5,1.0,2.5\}$ in $\mathcal{A}_{h}=\{1\leq k\leq K:\|\boldsymbol{\delta}^{(k)}\|_1\leq h\}$ for O-Trans-SQR. Under these settings, O-Trans-SQR selects 1, 4, and 7 transferable source domains when $K=10$, and selects 3, 5, and 11 transferable source domains when $K=20$, respectively. Since $h$ is specific to the source-level screening rule of O-Trans-SQR, the results for Target-only and O-SL-TL do not vary with $h$. The numerical results are reported in Table \ref{tab:oracle_l2_error}, where each entry is the mean $\ell_2$ estimation error over 100 Monte Carlo replications, with the corresponding standard deviation given in parentheses. 
\begin{table}[!htbp]
\centering
\caption{Mean $\ell_2$ estimation errors (standard deviations) of different methods under the oracle setting with $n_0=300$, $n_k=450$, and $p=500$.}
\label{tab:oracle_l2_error}
\small
\setlength{\tabcolsep}{5pt}
\begin{tabular}{cccrrr}
\toprule
\multirow{2}{*}{$K$}
& \multirow{2}{*}{Setting}
& \multirow{1}{*}{Method}
& \multicolumn{1}{c}{$\tau=0.2$} & \multicolumn{1}{c}{$\tau=0.5$} & \multicolumn{1}{c}{$\tau=0.7$} \\
\cmidrule(lr){3-6}

& & Target-only
& 0.665 (0.112)
& 0.568 (0.104)
& 0.593 (0.113) \\

\midrule

\multirow{6}{*}{$1$}

& \multirow{2}{*}{$\Delta_1=0.5$}
& O-Trans-SQR
& 0.363 (0.074)
& 0.337 (0.055)
& 0.357 (0.064) \\

& & O-SL-TL
& 0.307 (0.057)
& 0.283 (0.034)
& 0.293 (0.046) \\

\cmidrule(lr){3-6}

& \multirow{2}{*}{$\Delta_1=2.5$}
& O-Trans-SQR
& 0.491 (0.081)
& 0.470 (0.470)
& 0.483 (0.080) \\

& & O-SL-TL
& 0.442 (0.068)
& 0.372 (0.051)
& 0.411 (0.063) \\

\cmidrule(lr){3-6}

& \multirow{2}{*}{$\Delta_1=7.0$}
& O-Trans-SQR
& 0.929 (0.135)
& 0.867 (0.123)
& 0.895 (0.127) \\

& & O-SL-TL
& 0.563 (0.085)
& 0.494 (0.078)
& 0.538 (0.084) \\

\midrule

\multirow{4}{*}{$10$}

& $h=0.5$ & \multirow{3}{*}{O-Trans-SQR}
& 0.373 (0.069)
& 0.339 (0.057)
& 0.349 (0.063) \\

& $h=1.0$ & 
& 0.240 (0.043)
& 0.203 (0.045)
& 0.222 (0.048) \\

& $h=2.5$ &
& 0.196 (0.035)
& 0.178 (0.036)
& 0.182 (0.043) \\

\cmidrule(lr){3-6}

& -- & O-SL-TL
& 0.179 (0.080)
& 0.125 (0.056)
& 0.131 (0.057) \\

\midrule

\multirow{4}{*}{$20$}

& $h=0.5$ & \multirow{3}{*}{O-Trans-SQR}
& 0.256 (0.043)
& 0.229 (0.049)
& 0.237 (0.047) \\

& $h=1.0$ &
& 0.211 (0.044)
& 0.194 (0.043)
& 0.205 (0.045) \\

& $h=2.5$ &
& 0.170 (0.040)
& 0.159 (0.044)
& 0.168 (0.043) \\

\cmidrule(lr){3-6}

& -- & O-SL-TL
& 0.154 (0.056)
& 0.105 (0.061)
& 0.130 (0.068) \\

\bottomrule
\end{tabular}
\end{table}

As shown in Table \ref{tab:oracle_l2_error}, O-SL-TL consistently achieves the smallest mean $\ell_2$ estimation error across all considered combinations of $\tau,K$ and parameter settings.
This demonstrates the advantage of sample-level transfer over source-level transfer, as O-SL-TL is able to retain informative observations from source domains while reducing the impact of heterogeneous information. In particular, when $K=1$ and $\Delta_1=7.0$, O-Trans-SQR performs even worse than Target-only, indicating the occurrence of negative transfer caused by directly incorporating a highly heterogeneous source domain. This result is consistent with the theoretical comparison in Table \ref{tab:rate-comparison-single-source}. 
In contrast, O-SL-TL avoids negative transfer in this challenging setting by selectively utilizing transferable source samples.

\subsubsection{Feasible setting}
In the feasible setting, we compare the performance of Target-only, F-Trans-SQR, and F-SL-TL, with the transferable information estimated from the data. Target-only and F-SL-TL are defined in \eqref{hat_beta_T} and \eqref{betaSL-TLhat0}, respectively, while F-Trans-SQR follows Algorithm 2 of \citet{zhang2022transfer}.
It splits the target sample into training and validation subsets, selects sources according to their reduction in the target validation quantile loss, and applies the procedure in Algorithm \ref{alg:Trans-SQR} of Supplementary Material to the selected sources.
For F-SL-TL, the implementation details of the B-spline projection density estimator are provided in Section \ref{supp:Details} of Supplementary Material.  

We evaluate the three methods at the quantile levels \(\tau\in\{0.2,0.5,0.7\}\), with the number of source domains set to \(K\in\{10,20\}\). 
The numerical results are summarized in Table \ref{tab:feasible_l2_error}, where each entry reports the mean $\ell_2$ estimation error over 100 Monte Carlo replications, with the corresponding standard deviation given in parentheses.
As shown in Table \ref{tab:feasible_l2_error}, F-SL-TL consistently achieves the smallest mean $\ell_2$ estimation error across all considered values of $K$ and $\tau$, while F-Trans-SQR also consistently outperforms Target-only.
As $K$ increases from $10$ to $20$, the mean estimation errors of both transfer learning methods generally decrease, indicating gains from additional source information. 
Nevertheless, F-SL-TL maintains the best overall performance, highlighting the advantage of identifying and exploiting transferable information at the sample level rather than treating each source domain as an indivisible unit.

\begin{table}[htbp]
\centering
\caption{Mean $\ell_2$ estimation errors (standard deviations) of different methods under the feasible setting with $n_0=300$, $n_k=450$, and $p=500$.}
\label{tab:feasible_l2_error}
\small
\setlength{\tabcolsep}{7pt}
\begin{tabular}{ccrrr}
\toprule
\multirow{1}{*}{$K$}
& \multirow{1}{*}{Method}
&  \multicolumn{1}{c}{$\tau=0.2$}
&  \multicolumn{1}{c}{$\tau=0.5$}
&  \multicolumn{1}{c}{$\tau=0.7$} \\
\midrule

-- & Target-only
& 0.665 (0.112)
& 0.568 (0.104)
& 0.593 (0.113) \\

\midrule

\multirow{2}{*}{$10$}
& F-Trans-SQR
& 0.419 (0.059)
& 0.425 (0.055)
& 0.433 (0.054) \\

& F-SL-TL
& 0.402 (0.071)
& 0.366 (0.121)
& 0.371 (0.074) \\

\midrule

\multirow{2}{*}{$20$}
& F-Trans-SQR
& 0.406 (0.064)
& 0.421 (0.055)
& 0.426 (0.054) \\

& F-SL-TL
& 0.389 (0.070)
& 0.345 (0.129)
& 0.370 (0.128) \\

\bottomrule
\end{tabular}
\end{table}

\subsection{Real Data Analysis}
To evaluate the practical performance of the proposed sample-level transfer learning method, we analyze GTEx v11 tissue-level gene expression data (\citealp{GTExConsortium2020}; \url{https://gtexportal.org/}), measured in TPM, treating each tissue as a domain. We compare F-SL-TL method with the Target-only estimator and the F-Trans-SQR method to examine whether sample-level transfer from heterogeneous source tissues improves out-of-sample quantile prediction. 
The response variable is the expression level of JAM2. JAM2 is involved in cell adhesion and lymphocyte transendothelial migration \citep{JOHNSONLEGER20022479}, and rare loss-of-function variants in JAM2 have been associated with primary familial brain calcification \citep{Cen2020,Schottlaender2020}, making it a biologically relevant target for cross-tissue transfer learning. Candidate predictor genes are obtained from MODULE\_137 and its additional enriched genes \citep{Segal2004}. After removing duplicate gene symbols and matching them to GTEx v11 annotations, we retain genes shared across the selected tissues and use their expression levels as predictors, excluding JAM2.


To construct tissue domains, we first retain tissues with more than 180 samples with available JAM2 expression. All brain tissues are included due to their biological relevance, together with the seven largest non-brain tissues, yielding 20 domains (13 brain and 7 non-brain). 
Among the retained brain tissues, nine with more than 250 JAM2 samples are considered in turn as target domains, with the remaining 19 tissues serving as source domains (8,068–8,114 samples per experiment). Target tissues and their sample sizes are reported in Table \ref{tab:target_brain_tissues} of Section \ref{supp:Tables and Figures} in Supplementary Material. 
Each experiment evaluates cross-domain quantile prediction using one target tissue and multiple heterogeneous source tissues with F-SL-TL.










Before model fitting, gene expression measurements are transformed using
$\log_2(\mathrm{TPM}+1)$. Predictors and JAM2 are standardized within each tissue domain, and predictors with zero variance in any selected domain are removed, leaving 1,270 common predictors. Prediction performance is evaluated using five-fold cross-validation, where four folds are used for model fitting and tuning and the remaining fold for out-of-sample prediction. The prediction performance is measured at $\tau\in\{0.2,0.5,0.7\}$.

For a method $\mathcal{M}$, the mean quantile prediction error is defined as $\operatorname{MQE}_{\mathcal{M}}(\tau)=\sum_{i=1}^{n_0}\rho_\tau(\\Y_i^{(0)}-\hat Q_{\mathcal{M},-v(i)}(X_i^{(0)}))/n_0$, where $v(i)$ denotes the validation fold containing observation $i$. The relative prediction error is $\operatorname{RPE}^{\operatorname{MQE}}_{\mathcal{M}}(\tau)=\operatorname{MQE}_{\mathcal{M}}(\tau)/{\operatorname{MQE}_{\text{Target-only}}(\tau)}$, where values below one indicate improvement over the Target-only estimator. We additionally report the mean squared prediction error 
and its relative version, $\mathrm{RPE}^{\mathrm{MSE}}_{\mathcal{M}}(\tau)
=\mathrm{MSE}_{\mathcal{M}}(\tau)/{\mathrm{MSE}_{\text{Target-only}}(\tau)}$, with the same interpretation.
\begin{figure}[!t]
    \centering
    \includegraphics[width=0.8\textwidth]{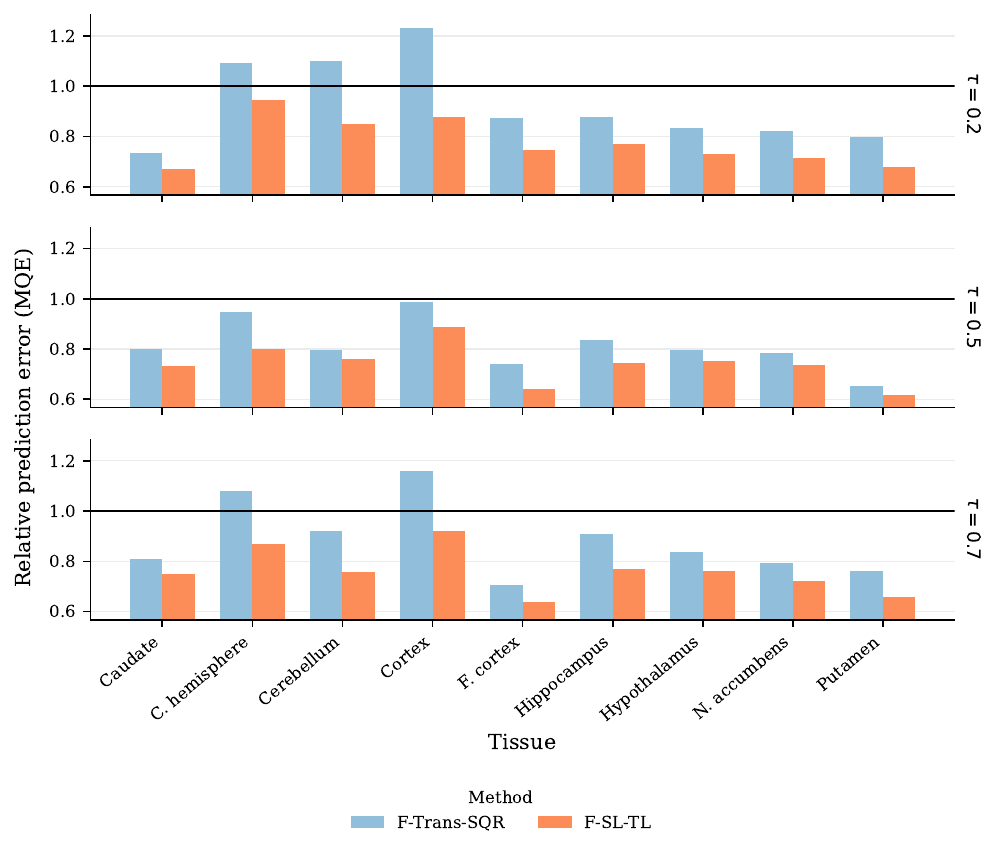}
    \caption{Relative quantile prediction errors of F-Trans-SQR and F-SL-TL relative to that of Target-only across 9 different target tissues at quantile levels 0.2, 0.5 and 0.7. The horizontal line represents an RPE of one.}
\label{fig:gtex_rpe}
\end{figure}

Figure \ref{fig:gtex_rpe} reports the relative quantile prediction errors across the nine target brain tissues. F-SL-TL consistently achieves relative errors below one and uniformly outperforms F-Trans-SQR across all target tissues and quantile levels, while F-Trans-SQR exceeds one in several tissues, indicating negative transfer. 
The corresponding MSE-based results in Figure \ref{fig:gtex_mse} of Section \ref{supp:Tables and Figures} in Supplementary Material exhibit a similar pattern, with F-SL-TL generally
yielding lower relative errors than F-Trans-SQR.

\section*{Supplementary Materials}
The Supplementary Material contains all technical proofs, additional results, and methodological details.
	\par
\section*{Acknowledgements}
The research was supported by the Public Computing Cloud, Renmin University of China.
\par

\bibliographystyle{apalike_copy}
\bibliography{./Beyond_Source_Level_Transfer__Sample_Level_Learning_for_High_Dimensional_Quantile_Regression}

\begin{thebibliography}{}

\bibitem[Belloni and Chernozhukov, 2011]{Belloni2011}
Belloni, A. and Chernozhukov, V. (2011).
\newblock $\ell_1$-penalized quantile regression in high-dimensional sparse models.
\newblock {\em The Annals of Statistics}, 39(1), 82--130.

\bibitem[Bickel et~al., 2009]{bickel2009}
Bickel, P.~J., Ritov, Y., and Tsybakov, A.~B. (2009).
\newblock Simultaneous analysis of lasso and {Dantzig} selector.
\newblock {\em The Annals of Statistics}, 37(4), 1705--1732.

\bibitem[Cand{\`e}s and Tao, 2007]{candes2007dantzig}
Cand{\`e}s, E. and Tao, T. (2007).
\newblock The {Dantzig} selector: Statistical estimation when $p$ is much larger than $n$.
\newblock {\em The Annals of Statistics}, 35(6), 2313--2351.

\bibitem[Cen et~al., 2020]{Cen2020}
Cen, Z., Chen, Y., Chen, S., et~al. (2020).
\newblock Biallelic loss-of-function mutations in {JAM2} cause primary familial brain calcification.
\newblock {\em Brain}, 143(2), 491--502.

\bibitem[Engle and Manganelli, 2004]{engle2004caviar}
Engle, R.~F. and Manganelli, S. (2004).
\newblock {CAViaR}: Conditional autoregressive value at risk by regression quantiles.
\newblock {\em Journal of Business \& Economic Statistics}, 22(4), 367--381.

\bibitem[Fan and Lv, 2008]{fan2008}
Fan, J. and Lv, J. (2008).
\newblock Sure independence screening for ultrahigh dimensional feature space.
\newblock {\em Journal of the Royal Statistical Society: Series B (Statistical Methodology)}, 70(5), 849--911.

\bibitem[Feng et~al., 2024]{FengJMLRQuantile}
Feng, X., He, X., Jiao, Y., Kang, L., and Wang, C. (2024).
\newblock Deep nonparametric quantile regression under covariate shift.
\newblock {\em Journal of Machine Learning Research}, 25(385), 1--50.

\bibitem[Gin{\'e} and Nickl, 2010]{Gine2008AdaptiveEO}
Gin{\'e}, E. and Nickl, R. (2010).
\newblock Adaptive estimation of a distribution function and its density in sup-norm loss by wavelet and spline projections.
\newblock {\em Bernoulli}, 16(4), 1137--1163.

\bibitem[{GTEx Consortium}, 2020]{GTExConsortium2020}
{GTEx Consortium} (2020).
\newblock The {GTEx Consortium} atlas of genetic regulatory effects across human tissues.
\newblock {\em Science}, 369(6509), 1318--1330.

\bibitem[Johnson-L{\'e}ger et~al., 2002]{JOHNSONLEGER20022479}
Johnson-L{\'e}ger, C.~A., Aurrand-Lions, M., Beltraminelli, N., Fasel, N., and Imhof, B.~A. (2002).
\newblock Junctional adhesion molecule-2 ({JAM-2}) promotes lymphocyte transendothelial migration.
\newblock {\em Blood}, 100(7), 2479--2486.

\bibitem[Knight, 1998]{knight1998limiting}
Knight, K. (1998).
\newblock Limiting distributions for {$L_1$} regression estimators under general conditions.
\newblock {\em The Annals of Statistics}, 26(2), 755--770.

\bibitem[Li et~al., 2022]{Lisai2022}
Li, S., Cai, T.~T., and Li, H. (2022).
\newblock Transfer learning for high-dimensional linear regression: Prediction, estimation and minimax optimality.
\newblock {\em Journal of the Royal Statistical Society: Series B (Statistical Methodology)}, 84(1), 149--173.

\bibitem[Meinshausen and Yu, 2009]{meinshausen2009}
Meinshausen, N. and Yu, B. (2009).
\newblock Lasso-type recovery of sparse representations for high-dimensional data.
\newblock {\em The Annals of Statistics}, 37(1), 246--270.

\bibitem[Novak et~al., 2017]{novak2017change}
Novak, N.~L., Geronimus, A.~T., and Martinez-Cardoso, A.~M. (2017).
\newblock Change in birth outcomes among infants born to latina mothers after a major immigration raid.
\newblock {\em International Journal of Epidemiology}, 46(3), 839--849.

\bibitem[Schottlaender et~al., 2020]{Schottlaender2020}
Schottlaender, L.~V., Abeti, R., Jaunmuktane, Z., et~al. (2020).
\newblock Bi-allelic {JAM2} variants lead to early-onset recessive primary familial brain calcification.
\newblock {\em The American Journal of Human Genetics}, 106(3), 412--421.

\bibitem[Segal et~al., 2004]{Segal2004}
Segal, E., Friedman, N., Koller, D., and Regev, A. (2004).
\newblock A module map showing conditional activity of expression modules in cancer.
\newblock {\em Nature Genetics}, 36(10), 1090--1098.

\bibitem[Tan et~al., 2022]{tan2022high}
Tan, K.~M., Wang, L., and Zhou, W.-X. (2022).
\newblock High-dimensional quantile regression: Convolution smoothing and concave regularization.
\newblock {\em Journal of the Royal Statistical Society: Series B (Statistical Methodology)}, 84(1), 205--233.

\bibitem[Tibshirani, 1996]{LASSO1996}
Tibshirani, R. (1996).
\newblock Regression shrinkage and selection via the lasso.
\newblock {\em Journal of the Royal Statistical Society: Series B (Methodological)}, 58(1), 267--288.

\bibitem[van~de Geer, 2008]{vande2008}
van~de Geer, S.~A. (2008).
\newblock High-dimensional generalized linear models and the lasso.
\newblock {\em The Annals of Statistics}, 36(2), 614--645.

\bibitem[Wainwright, 2019]{wainwright2019high}
Wainwright, M.~J. (2019).
\newblock {\em High-dimensional statistics: A non-asymptotic viewpoint}, Vol. 48, Cambridge university press.

\bibitem[Wang et~al., 2026]{wang2026deep}
Wang, Y., Xing, X., and Xu, W. (2026).
\newblock Deep regression for repeated measurements under covariate shift.
\newblock arXiv preprint arXiv:2605.24854.

\bibitem[Zhang et~al., 2023]{zhangnegative2026}
Zhang, W., Deng, L., Zhang, L., and Wu, D. (2023).
\newblock A survey on negative transfer.
\newblock {\em IEEE/CAA Journal of Automatica Sinica}, 10(2), 305--329.

\bibitem[Zhang and Zhu, 2025]{zhang2022transfer}
Zhang, Y. and Zhu, Z. (2025).
\newblock Transfer learning for high-dimensional quantile regression via convolution smoothing.
\newblock {\em Statistica Sinica}, 35(2), 939--958.

\bibitem[Zhao et~al., 2026]{zhao2026residual}
Zhao, J., Zheng, S., and Leng, C. (2026).
\newblock Residual importance weighted transfer learning for high-dimensional linear regression.
\newblock {\em Journal of the American Statistical Association}, 0, 1--27.

\end{thebibliography}

\newpage
\renewcommand{\baselinestretch}{2}


\makeatletter
\setcounter{authors}{0}
\setcounter{affil}{0}
\def\AB@authlist{}
\def\AB@affillist{}
\makeatother
\spacingset{1.9}
\title{\bfseries Supplementary Materials for ``Beyond Source-Level Transfer: Sample-Level Learning for High-Dimensional Quantile Regression"}
  \author[a]{Xiangyu Xing} 
  \author[a]{Yingxuan Wang}
  \author[a] {Wangli Xu$^*$}
  \affil[a]{Center for Applied Statistics and School of Statistics, Renmin University of China, Beijing, 100872, China}
\maketitle
\vspace{-0.8cm}
This supplementary material is organized as follows. Section \ref{supp:Appendix Theorem1} provides the proof of Theorem \ref{prop:RIW-QR}. Section \ref{supp:Appendix A} presents the proofs of Theorems \ref{pkboundedwkbounded}--\ref{oracletheorem} for the oracle setting. Section \ref{supp:Density estimation} introduces the B-spline projection density estimator, and Section \ref{supp:Appendix B} gives the proofs of Lemmas \ref{lemma:selection-consistency}--\ref{lemma hat p_k} and Theorem \ref{estimatebetabound} for the feasible setting, together with auxiliary results including Lemmas \ref{lem:S1}--\ref{lem:S2} and Proposition \ref{prop:gaussian-boundary}. Section \ref{supp:Tables and Figures} provides additional results from the real data analysis. Section \ref{supp:alg} presents the detailed oracle Trans-SQR algorithm. Finally, Section \ref{supp:Details} describes the simulation settings for the B-spline projection density estimator used in F-SL-TL.
\par

\setcounter{section}{0}
\setcounter{equation}{0}
\renewcommand{\thesection}{S\arabic{section}}
\renewcommand{\thesubsection}{S\arabic{section}.\arabic{subsection}}
\numberwithin{equation}{section}
\renewcommand{\theequation}{S\arabic{section}.\arabic{equation}}
\setcounter{table}{0}
\setcounter{figure}{0}
\setcounter{algorithm}{0}
\renewcommand{\thetable}{S\arabic{table}}
\renewcommand{\thefigure}{S\arabic{figure}}
\renewcommand{\thealgorithm}{S\arabic{algorithm}}


\section{Proof of Theorem \ref{prop:RIW-QR}} 
\label{supp:Appendix Theorem1}
In this section, we provide the proof of Theorem \ref{prop:RIW-QR}, which establishes the explicit form of the one-dimensional sample-level importance weights. 
\begin{proof}
For $k=1,\ldots,K$, recalling the definition of $\omega^{(k)}(\cdot)$, we have
\begin{align*}
&\arg\min_{\boldsymbol{\beta}\in\mathbb{R}^p}\mathbb{E}_{\boldsymbol{Z}^{(k)}\sim p_k}\left\{\omega^{(k)}(\boldsymbol{Z}^{(k)})
\rho_\tau(Y^{(k)}-(\boldsymbol{X}^{(k)})^\top\boldsymbol{\beta})\right\}\\
=&\arg\min_{\boldsymbol{\beta}\in\mathbb{R}^p}\mathbb{E}_{\boldsymbol{X}^{(k)}}\left[\mathbb{E}_{Y^{(k)}|\boldsymbol{X}^{(k)}}\left\{\omega^{(k)}(\boldsymbol{Z}^{(k)})\rho_\tau(Y^{(k)}-(\boldsymbol{X}^{(k)})^\top\boldsymbol{\beta})\right\}\right]\\
=&\arg\min_{\boldsymbol{\beta}\in\mathbb{R}^p}\mathbb{E}_{\boldsymbol{X}^{(k)}}\big\{\int\rho_\tau(Y^{(k)}-(\boldsymbol{X}^{(k)})^\top\boldsymbol{\beta})\frac{g_\tau(Y^{(k)}-(\boldsymbol{X}^{(k)})^\top\boldsymbol{\beta}^{(0)})}{p_{\varepsilon^{(k)}}(Y^{(k)}-(\boldsymbol{X}^{(k)})^\top\boldsymbol{\beta}^{(k)})}p_k(Y^{(k)}|\boldsymbol{X}^{(k)})dY^{(k)}\big\}\\
=&\arg\min_{\boldsymbol{\beta}\in\mathbb{R}^p}\mathbb{E}_{\boldsymbol{X}^{(k)}}\big\{\int\rho_\tau(Y^{(k)}-(\boldsymbol{X}^{(k)})^\top\boldsymbol{\beta})g_\tau(Y^{(k)}-(\boldsymbol{X}^{(k)})^\top\boldsymbol{\beta}^{(0)})dY^{(k)}\big\}\\
=&\arg\min_{\boldsymbol{\beta}\in\mathbb{R}^p}\mathbb{E}_{\boldsymbol{X}^{(k)}}\{\mathbb{E}_{U\sim g_{\tau}}[\rho_\tau(U-(\boldsymbol{X}^{(k)})^\top(\boldsymbol{\beta}-\boldsymbol{\beta}^{(0)}))]\},
\end{align*}
where $U=Y^{(k)}-(\boldsymbol{X}^{(k)})^\top\boldsymbol{\beta}^{(0)}$. For any fixed $\boldsymbol{X}^{(k)}$, let $d=(\boldsymbol{X}^{(k)})^\top(\boldsymbol{\beta}-\boldsymbol{\beta}^{(0)})$ and define $H(d)=\int_{\mathbb{R}}\rho_{\tau}(u-d)g_{\tau}(u)du$. Then $H^{\prime}(d)=\int_{-\infty}^dg_\tau(u)du-\tau$. By assumption $\int_{-\infty}^0g_\tau(u)du=\tau$, it therefore follows that $H^{\prime}(0)=0$. Moreover, since $H(d)$ is convex in $d$, the first-order optimality condition implies that $0\in\arg\min_{d\in\mathbb{R}}H(d)$. Consequently, for each fixed $\boldsymbol{X}^{(k)}$, the conditional risk is minimized whenever $(\boldsymbol{X}^{(k)})^\top(\boldsymbol{\beta}-\boldsymbol{\beta}^{(0)})=0$. In particular, $\boldsymbol{\beta}=\boldsymbol{\beta}^{(0)}$ satisfies this condition and therefore attains the minimum. This completes the proof.
\end{proof}

\section{Proofs of Theorems \ref{pkboundedwkbounded}--\ref{oracletheorem}}
\label{supp:Appendix A}
In this section, we provide the proofs of Theorems \ref{pkboundedwkbounded}--\ref{oracletheorem} for the oracle SL-TL estimator, where the transferable information is assumed to be known. Specifically, Theorem \ref{pkboundedwkbounded} establishes the boundedness of the oracle importance weights over the selected source observations. 
Theorem \ref{effectivesamplesize} characterizes the effective sample size contributed by each source domain through the discrepancy between the source and target domains and the sample-splitting parameter. 
Theorem \ref{oracletheorem} derives the $\ell_2$-error bound for the oracle SL-TL estimator.

\subsection{Proof of Theorem \ref{pkboundedwkbounded}}
\begin{proof}
For each $k=1,2,\dots,K$ and for all $i \in \mathcal{I}_{k}$, it is easy to show that
\begin{equation*}
|\varepsilon_i^{(k)}|\leq|\varepsilon_i^{(k)}+\eta_i^{(k)}|+|\eta_i^{(k)}|\leq A+M_k.
\end{equation*}
Under Assumption \ref{pkbound}, $p_{\varepsilon^{(k)}}(\cdot)$ is continuous and strictly positive on the compact interval $[-A-M_k,A+M_k]$. Hence, the extreme value theorem yields
\begin{equation*}
    0<\underline{p}_k:=\inf_{|u|\leq A+M_k}p_{\varepsilon^{(k)}}(u)\leq \overline{p}_k:=\sup_{|u|\leq A+M_k}p_{\varepsilon^{(k)}}(u)<\infty.
\end{equation*}
For every $i\in\mathcal{I}_k$, by the definition of $g_\tau(\cdot)$, we have $g_\tau(r_i^{(k)})=1/A$. Therefore,
\begin{equation*}
    1/(A\overline{p}_k)\leq\omega^{(k)}(\boldsymbol{Z}_i^{(k)})=1/(Ap_{\varepsilon^{(k)}}(\varepsilon_i^{(k)}))\leq1/(A\underline{p}_k).
\end{equation*}
Since $K$ is finite, define
\begin{equation*}
    c_{\omega}:=\min_{1\leq k\leq K}1/(A\overline{p}_{k})>0,\quad C_{\omega}:=\max_{1\leq k\leq K}1/(A\underline{p}_{k})<\infty.
\end{equation*}
Thus, $0<c_\omega\leq \omega^{(k)}(\boldsymbol{Z}_i^{(k)})\leq C_\omega<\infty$ uniformly over $k=1,2,\ldots,K$ and $i\in\mathcal{I}_k$.
\end{proof}

\subsection{Proof of Theorem \ref{effectivesamplesize}}
\label{sup:S2.2}
\begin{proof}
The proof follows an argument similar to that used in the proof of Proposition 3 in \citet{zhao2026residual}. By the definition of $\mathcal{I}_k$, we have
\begin{align*}
    \mathbb{E}(n_{\mathcal{I}_k})&=n_k\mathbb{E}[I\{i\in\mathcal{I}_k\}]=n_k\mathbb{P}(-\tau A\leq\varepsilon_i^{(k)}+\eta_i^{(k)}\leq(1-\tau)A,|\eta_i^{(k)}|\leq M_k)\\
    &=n_k\mathbb{P}(-\tau A\leq\varepsilon_i^{(k)}+\eta_i^{(k)}\leq(1-\tau)A\mid|\eta_i^{(k)}|\leq M_k)\mathbb{P}(|\eta_i^{(k)}|\leq M_k)\\
    &\geq n_k\mathbb{P}(-\tau A+M_k\leq\varepsilon_i^{(k)}\leq(1-\tau)A-M_k\mid|\eta_i^{(k)}|\leq M_k)\mathbb{P}(|\eta_i^{(k)}|\leq M_k)\\
    &=n_k\mathbb{P}(-\tau A+M_k\leq\varepsilon_i^{(k)}\leq(1-\tau)A-M_k)\mathbb{P}(|\eta_i^{(k)}|\leq M_k)\\
    &\geq n_k\mathbb{P}(|\varepsilon_i^{(k)}|\leq \min\{\tau,1-\tau\}A-M_k)\mathbb{P}(|\eta_i^{(k)}|\leq M_k).
\end{align*}
On the other hand, by the definition of $\mathcal{I}_k$
\begin{equation*}
    \mathbb{P}(-\tau A\leq\varepsilon_i^{(k)}+\eta_i^{(k)}
    \leq(1-\tau)A,|\eta_i^{(k)}|\leq M_k)\leq\mathbb{P}(|\eta_i^{(k)}|\leq M_k).
\end{equation*}
By Assumption \ref{pkbound} and the condition $\min\{\tau, 1-\tau\} A-M_k \geq c_0$ for some $c_0>0$, we have 
\begin{equation*}
\mathbb{P}(|\varepsilon_i^{(k)}|\leq\min\{\tau,1-\tau\}A-M_k)\geq \mathbb{P}(|\varepsilon_i^{(k)}|\leq c_0)\geq2c_0\inf_{|u|\leq c_0}p_{\varepsilon^{(k)}}(u)>0.
\end{equation*}
Therefore, for fixed $K$, $\mathbb{P}(|\varepsilon_i^{(k)}|\leq\min\{\tau,1-\tau\}A-M_k)\asymp1$ uniformly over $k=1,2,\ldots,K$. 

Next, we consider the probability $\mathbb{P}(|\eta_i^{(k)}|\leq M_k)$. Since $\boldsymbol{X}_i^{(k)}\sim N_p(\boldsymbol{0},\Sigma^{(k)})$, we have $\eta_i^{(k)}=(\boldsymbol{X}_i^{(k)})^\top\boldsymbol{\delta}^{(k)}\sim N(0,(\boldsymbol{\delta}^{(k)})^\top\Sigma^{(k)}\boldsymbol{\delta}^{(k)})$. We define $\zeta_k:=(\boldsymbol{\delta}^{(k)})^\top\Sigma^{(k)}\boldsymbol{\delta}^{(k)}/\|\boldsymbol{\delta}^{(k)}\|_2^2$ when $d_k=\|\boldsymbol{\delta}^{(k)}\|_2>0$. By the eigenvalue bounds on $\Sigma^{(k)}$, we have $\zeta_l\leq\zeta_k\leq\zeta_u$. Let $U$ and $V$ be two independent random variables following $N(0,\zeta_kd_k^2)$. Then we can obtain
\begin{align*}
    [\mathbb{P}(|\eta_i^{(k)}|\leq M_k)]^2&=\mathbb{P}(|U|\leq M_k,\mathrm{~}|V|\leq M_k)\\
    &=\int_{-M_k}^{M_k}\int_{-M_k}^{M_k}\exp\{-(u^2+v^2)/(2\zeta_k d_k^2)\}/(2\pi\zeta_k d_k^2)dudv.
\end{align*}
Since $\{u^2+v^2\leq M_k^2\}\subset[-M_k,M_k]^2\subset\{u^2+v^2\leq 2M_k^2\}$, it leads to
\begin{equation*}
    1-\exp[-M_k^2/(2\zeta_k d_k^2)]\leq[\mathbb{P}(|\eta_i^{(k)}|\leq M_k)]^2\leq1-\exp[-M_k^2/(\zeta_k d_k^2)].
\end{equation*}
Combining the results, we have
\begin{equation*}
    \mathbb{E}(n_{\mathcal{I}_k})\asymp n_k\mathbb{P}(|\eta_i^{(k)}|\leq M_k)\asymp n_k\sqrt{1-\exp(-M_k^2/d_k^2)}.
\end{equation*}
When $0<M_k/d_k<1$, we have $\sqrt{1-\exp(-M_k^2/d_k^2)}\asymp M_k/d_k$. When $M_k/d_k\geq1$, we have $\sqrt{1-\exp(-1)}\leq\sqrt{1-\exp(-M_k^2/d_k^2)}\leq 1$, so we can summarize that
\begin{equation*}
    \mathbb{E}(n_{\mathcal{I}_k})\asymp n_k\{M_{k}/d_k\wedge1\},~d_k>0.
\end{equation*}
If $d_k=\|\boldsymbol{\delta}^{(k)}\|_2=0$, then $\boldsymbol{\delta}^{(k)}=\boldsymbol{0}_p$ and hence $\eta_i^{(k)}=0$ almost surely. Since $M_k>0$, $\mathbb{P}(|\eta_i^{(k)}|\leq M_k)=1$. Therefore, $\mathbb{E}(n_{\mathcal{I}_k})\asymp n_k$ when $d_k=0$.
This completes the proof.
\end{proof}

\subsection{Proof of Theorem \ref{oracletheorem}}
\begin{proof}
Throughout the proof, all probabilities and expectations are understood conditionally on the fixed design $\{\boldsymbol{X}_i^{(k)}\}_{i=1}^{n_k}$. 
For notational simplicity, we present the proof for $K=1$. The argument for a fixed $K>1$ follows identically by summing the corresponding source terms over $k=1,2,\ldots,K$, with $n_{\mathcal{I}_1}$ replaced by $n_\mathcal{I}=\sum_{k=1}^Kn_{\mathcal{I}_k}$.
We denote
\begin{equation*}
    L(\boldsymbol{\beta}):=\frac{1}{n_{0}+n_{1}}[\sum_{i=1}^{n_{0}}\rho_{\tau}(Y_{i}^{(0)}-(\boldsymbol{X}_{i}^{(0)})^\top \boldsymbol{\beta}+\sum_{i=1}^{n_{1}}\omega_{i}^{(1)}\mathds{1}\{i\in \mathcal{I}_{1}\}\rho_{\tau}(Y_{i}^{(1)}-(\boldsymbol{X}_{i}^{(1)})^\top \boldsymbol{\beta})].
\end{equation*}
By the definition of $\hat{\boldsymbol{\beta}}_{\mathrm{ora}}^{(0)}$, we know that $L(\hat{\boldsymbol{\beta}}_{\mathrm{ora}}^{(0)})+\lambda\|\hat{\boldsymbol{\beta}}_{ora}^{(0)}\|_{1}\leq L(\boldsymbol{\beta}^{(0)})+\lambda\|\boldsymbol{\beta}^{(0)}\|_{1}$. 
Furthermore, applying Knight’s identity \citep{knight1998limiting} with $\psi_{\tau}(x)=\tau-\mathds{1}\{x\leq0\}$, $\Delta(x,z)=\mathds{1}\{x\leq z\}-\mathds{1}\{x\leq 0\}$ and $r_{i}^{(k)}=\varepsilon_{i}^{(k)}+\eta_{i}^{(k)}$, we obtain
\begin{align}
    &\frac{1}{n_{0}+n_{1}}\big[\sum_{i=1}^{n_{0}}\int_{0}^{(\boldsymbol{X}_{i}^{(0)})^{T}\hat{\boldsymbol{\delta}}^{(0)}}\Delta(\varepsilon_{i}^{(0)},z)dz+\sum_{i=1}^{n_{1}}\int_{0}^{(\boldsymbol{X}_{i}^{(1)})^{T}\hat{\boldsymbol{\delta}}^{(0)}}\omega_{i}^{(1)}\mathds{1}\{i\in \mathcal{I}_{1}\}\Delta(r_{i}^{(1)},z)dz\big]\nonumber\\
    \leq&\lambda\|\boldsymbol{\beta}^{(0)}\|_{1}-\lambda\|\hat{\boldsymbol{\beta}}_{\mathrm{ora}}^{(0)}\|_{1}
    +\frac{1}{n_{0}+n_{1}}\big[\sum_{i=1}^{n_{0}}\boldsymbol{X}_{i}^{(0)}\psi_{\tau}(\varepsilon_{i}^{(0)})+\sum_{i=1}^{n_{1}}\omega_{i}^{(1)}\mathds{1}\{i\in \mathcal{I}_{1}\}\boldsymbol{X}_{i}^{(1)}\psi_{\tau}(r_{i}^{(1)})\big]^{\top} \hat{\boldsymbol{\delta}}^{(0)},
    \label{I+II}
\end{align}
where $\hat{\boldsymbol{\delta}}^{(0)}=\hat{\boldsymbol{\beta}}_{\mathrm{ora}}^{(0)}-\boldsymbol{\beta}^{(0)}$. We define
\begin{align*}
&(I)=\frac{1}{n_{0}+n_{1}}\big[\sum_{i=1}^{n_{0}}\boldsymbol{X}_{i}^{(0)}\psi_{\tau}(\varepsilon_{i}^{(0)})+\sum_{i=1}^{n_{1}}\omega_{i}^{(1)}\mathds{1}\{i\in \mathcal{I}_{1}\}\boldsymbol{X}_{i}^{(1)}\psi_{\tau}(r_{i}^{(1)})\big],\\ 
&(II)=\frac{1}{n_{0}+n_{1}}\big[\sum_{i=1}^{n_{0}}\int_{0}^{(\boldsymbol{X}_{i}^{(0)})^{T}\hat{\delta}^{(0)}}\Delta(\varepsilon_{i}^{(0)},z)dz
+\sum_{i=1}^{n_{1}}\int_{0}^{(\boldsymbol{X}_{i}^{(1)})^{T}\hat{\delta}^{(0)}}\omega_{i}^{(1)}\mathds{1}\{i\in \mathcal{I}_{1}\}\Delta(r_{i}^{(1)},z)dz\big].
\end{align*}

\textbf{Step 1}. We first control $(I)$. We have $\mathbb{E}[X_{ij}^{(0)}\psi_{\tau}(\varepsilon_{i}^{(0)})\mid\boldsymbol{X}_{i}^{(0)}]=X_{ij}^{(0)}\mathbb{E}[\tau-\mathds{1}\{\varepsilon_{i}^{(0)}\leq0\}\mid \boldsymbol{X}_{i}^{(0)}]=0$. Using the definition of the weight $\omega_{i}^{(1)}$, we obtain
\begin{align*}
    &\mathbb{E}[\omega_{i}^{(1)}\mathds{1}\{i\in\mathcal{I}_{1}\}X_{ij}^{(1)}\psi_{\tau}(r_{i}^{(1)})\mid\boldsymbol{X}_{i}^{(1)}]  =\int_{i\in\mathcal{I}_{1}}X_{ij}^{(1)}\frac{g_\tau(r_{i}^{(1)})}{p_{\varepsilon^{(1)}}(\varepsilon_{i}^{(1)})}\psi_{\tau}(r_{i}^{(1)})p_{\varepsilon^{(1)}}(\varepsilon_{i}^{(1)})d\varepsilon_{i}^{(1)}\\
    =&X_{ij}^{(1)}\mathds{1}\{|\eta_i^{(1)}\leq M_1|\}\int_{-\tau A}^{(1-\tau)A}g_\tau(r_i^{(1)})\psi_\tau(r_i^{(1)})dr_i^{(1)}
    =0. 
\end{align*}
Let $U_{ij}^{(0)}=X_{ij}^{(0)}\psi_{\tau}(\varepsilon_{i}^{(0)})$, $U_{ij}^{(1)}=\omega_{i}^{(1)}X_{ij}^{(1)}I\{i\in\mathcal{I}_{1}\}\psi_{\tau}(r_{i}^{(1)})$. It follows that $\mathbb{E}(U_{ij}^{(0)}\mid \boldsymbol{X}_{i}^{(0)})=\mathbb{E}(U_{ij}^{(1)}\mid \boldsymbol{X}_{i}^{(1)})=0$, $|U_{ij}^{(0)}|\leq\Psi_{n}$ and $|U_{ij}^{(1)}|\leq C_{\omega}\Psi_{n}$. 
By a straightforward derivation, we obtain
\begin{align*}
    (n_0+n_1)^2V_j&:=\sum_{i=1}^{n_0}Var(U_{ij}^{(0)})+\sum_{i=1}^{n_1}Var(U_{ij}^{(1)})\\
    &=\sum_{i=1}^{n_0}\mathbb{E}[(U_{ij}^{(0)})^2]+\sum_{i=1}^{n_1}\mathbb{E}[(U_{ij}^{(1)})^2]\\
    &=\sum_{i=1}^{n_0}\mathbb{E}((X_{ij}^{(0)}\psi_{\tau}(\varepsilon_{i}^{(0)}))^2)+\sum_{i=1}^{n_1}\mathbb{E}((\omega_{i}^{(1)}X_{ij}^{(1)}\psi_{\tau}(r_{i}^{(1)}))^2I\{i\in\mathcal{I}_{1}\})\\
    &\leq n_0\Psi_{n}^2+\sum_{i=1}^{n_1}\Psi_{n}^2C_{\omega}^2P(i\in\mathcal{I}_{1})\\
    &\leq \Psi_{n}^2(1+C_{\omega}^2)(n_0+\mathbb{E}(n_{\mathcal{I}_1})).
\end{align*}
Using Bernstein's inequality \citep{wainwright2019high}, it follows that
\begin{align*}
    &\mathbb P(\frac{1}{n_0+n_1}|\sum_{i=1}^{n_0}U_{ij}^{(0)}+\sum_{i=1}^{n_1}U_{ij}^{(1)}|>t)
    \\\leq&2\exp\big\{-\frac{t^2}{2(V_j+(C_{\omega}+1)\Psi_{n}t/(n_0+n_1))}\big\}\\
    \leq&2\exp\big\{-\frac{(n_0+n_1)^2t^2}{2(\Psi_{n}^2(1+C_{\omega}^2)(n_0+\mathbb{E}(n_{\mathcal{I}_1}))+(C_{\omega}+1)(n_0+n_1)\Psi_{n}t)}\big\},
\end{align*}
 and let $t=C_t\sqrt{(n_{0}+E(n_{\mathcal{I}_{1}}))\log p}/(n_{0}+n_{1})$
 \begin{align*}
     &\mathbb P(\frac{1}{n_0+n_1}|\sum_{i=1}^{n_0}U_{ij}^{(0)}+\sum_{i=1}^{n_1}U_{ij}^{(1)}|>\frac{C_t\sqrt{(n_{0}+E(n_{\mathcal{I}_{1}}))\log p}}{n_{0}+n_{1}})\\
     \leq&2\exp\big\{-\frac{C_t^2\log p}{2[\Psi_{n}^2(1+C_{\omega}^2)+(1+C_{\omega})\Psi_{n}C_{t}\sqrt{\log p/(n_0+\mathbb{E}(n_{\mathcal{I}_1}))}]}\big\}:=2p^{-C_0},
 \end{align*}
 where $C_0:=C_t^2/2[\Psi_{n}^2(1+C_{\omega}^2)+(1+C_{\omega})\Psi_{n}C_{t}\sqrt{\log p/(n_0+\mathbb{E}(n_{\mathcal{I}_1}))}]$. Then
 \begin{align*}
     &\mathbb P\big(\|(I)\|_{\infty}>\frac{C_t\sqrt{(n_{0}+E(n_{\mathcal{I}_{1}}))\log p}}{n_{0}+n_{1}}\big)\\
     \leq&\sum_{j=1}^pP(\frac{1}{n_0+n_1}|\sum_{i=1}^{n_0}U_{ij}^{(0)}+\sum_{i=1}^{n_1}U_{ij}^{(1)}|>\frac{C_t\sqrt{(n_{0}+E(n_{\mathcal{I}_{1}}))\log p}}{n_{0}+n_{1}})\leq2p^{1-C_0}.
 \end{align*}
Therefore, by choosing the fixed constant $C_t$ sufficiently large, there exists a constant $c>0$ such that $C_0>1+c$. 
Thus, we obtain $\|(I)\|_{\infty}=O_p(\rho_{\mathcal{I}_1}\sqrt{\log p/(n_0+\mathbb{E}(n_{\mathcal{I}_1}))})$. Since $\lambda_{\mathrm{ora}}\asymp\rho_{\mathcal{I}}\sqrt{\log p/(n_0+\mathbb{E}(n_{\mathcal{I}}))}$, we have $\|(I)\|_{\infty}\lesssim\lambda_{\mathrm{ora}}/2$ with probability tending to one.

\textbf{Step 2}. Next, we control $(II)$. Decompose $(II)=\mathbb{E}[(II)\mid \boldsymbol{X}]+(II)-\mathbb{E}[(II)\mid \boldsymbol{X}]$. Let $H_0=\mathrm{supp}(\boldsymbol{\beta}^{(0)})$ and $|H_0|=s$. On the event $\|(I)\|_{\infty}\lesssim\lambda_{\mathrm{ora}}/2$ combining \eqref{I+II} with H\"{o}lder's inequality gives
\begin{align*}
    (II)&\leq\lambda_{\mathrm{ora}}\|\boldsymbol{\beta}^{(0)}\|_1-\lambda_{\mathrm{ora}}\|\hat{\boldsymbol{\beta}}_{\mathrm{ora}}^{(0)}\|_1+(I)^\top\hat{\boldsymbol{\delta}}^{(0)}\nonumber\\
    &\leq\lambda_{\mathrm{ora}}\|\boldsymbol{\beta}^{(0)}_{H_{0}}\|_1-\lambda_{\mathrm{ora}}\|\boldsymbol{\beta}^{(0)}_{H_{0}}+\hat{\boldsymbol{\delta}}_{H_{0}}^{(0)}\|_1-\lambda_{\mathrm{ora}}\|\hat{\boldsymbol{\delta}}_{H_{0}^{c}}^{(0)}\|_1+
    \frac{\lambda_{\mathrm{ora}}}{2}\|\hat{\boldsymbol{\delta}}^{(0)}\|_1\nonumber\\
    &\leq\lambda_{\mathrm{ora}}\|\boldsymbol{\beta}^{(0)}_{H_{0}}\|_1-\lambda_{\mathrm{ora}}(\|\boldsymbol{\beta}^{(0)}_{H_{0}}\|_1-\|\hat{\boldsymbol{\delta}}_{H_{0}}^{(0)}\|_1)-\lambda_{\mathrm{ora}}\|\hat{\boldsymbol{\delta}}_{H_{0}^{c}}^{(0)}\|_1+
    \frac{\lambda_{\mathrm{ora}}}{2}(\|\hat{\boldsymbol{\delta}}_{H_0}^{(0)}\|_1+\|\hat{\boldsymbol{\delta}}_{H_0^c}^{(0)}\|_1)\nonumber\\
    &=\frac{3\lambda_{\mathrm{ora}}}{2}\|\hat{\boldsymbol{\delta}}_{H_0}^{(0)}\|_1-\frac{\lambda_{\mathrm{ora}}}{2}\|\hat{\boldsymbol{\delta}}_{H_0^c}^{(0)}\|_1.
\end{align*}
Since the integral remainder term in Knight's identity is nonnegative, we have $(II)\geq0$. Therefore, the preceding inequality implies $\|\hat{\boldsymbol{\delta}}_{H_0^c}^{(0)}\|_1\leq3\|\hat{\boldsymbol{\delta}}_{H_0}^{(0)}\|_1$ and hence $\hat{\boldsymbol{\delta}}^{(0)}\in\mathcal{E}(H_0,3)$. Consequently, 
\begin{align}
    (II)&\leq\frac{3\lambda_{\mathrm{ora}}}{2}\|\hat{\boldsymbol{\delta}}_{H_0}^{(0)}\|_1-\frac{\lambda_{\mathrm{ora}}}{2}\|\hat{\boldsymbol{\delta}}_{H_0^c}^{(0)}\|_1\nonumber\\
    &\leq\frac{3\lambda_{\mathrm{ora}}}{2}\|\hat{\boldsymbol{\delta}}_{H_0}^{(0)}\|_1\nonumber\\
    &\leq\frac{3\lambda_{\mathrm{ora}}}{2}\sqrt{s}\|\hat{\boldsymbol{\delta}}^{(0)}\|_2.
    \label{IIupbound}
\end{align}
For the target unit, conditional on $\boldsymbol{X}_i^{(0)}$, we have
\begin{align}
    \mathbb{E}[\int_{0}^{(\boldsymbol{X}_{i}^{(0)})^\top\hat{\boldsymbol{\delta}}^{(0)}}\Delta(\varepsilon_{i}^{(0)},z)dz\mid\boldsymbol{X}_{i}^{(0)}]&=\int_{0}^{(\boldsymbol{X}_{i}^{(0)})^\top\hat{\boldsymbol{\delta}}^{(0)}}\mathbb{E}[\Delta(\varepsilon_{i}^{(0)},z)\mid\boldsymbol{X}_{i}^{(0)}]dz\nonumber\\
&=\int_{0}^{(\boldsymbol{X}_{i}^{(0)})^\top\hat{\boldsymbol{\delta}}^{(0)}}F_{\varepsilon^{(0)}\mid\boldsymbol{X}}(z)-F_{\varepsilon^{(0)}\mid\boldsymbol{X}}(0)dz.
\label{T4-II-target}
\end{align}
Under Assumption \ref{kappaF}, whenever $|(\boldsymbol{X}_{i}^{(0)})^\top\hat{\boldsymbol{\delta}}^{(0)}|\leq \xi$, \eqref{T4-II-target} yields
\begin{equation}
    \mathbb{E}[\int_{0}^{(\boldsymbol{X}_{i}^{(0)})^\top\hat{\boldsymbol{\delta}}^{(0)}}\Delta(\varepsilon_{i}^{(0)},z)dz\mid\boldsymbol{X}_{i}^{(0)}]\geq\kappa((\boldsymbol{X}_{i}^{(0)})^\top\hat{\boldsymbol{\delta}}^{(0)})^2/2.
    \label{T4-II-target-1}
\end{equation}
For the source unit, when $|(\boldsymbol{X}_i^{(1)})^\top\hat{\boldsymbol{\delta}}^{(0)}|\leq A\min\{\tau,1-\tau\}$, we similarly obtain that
\begin{align}
    &\mathbb{E}[\int_{0}^{(\boldsymbol{X}_{i}^{(1)})^\top\hat{\boldsymbol{\delta}}^{(0)}}\omega_i^{(1)}\mathds{1}\{i\in\mathcal{I}_1\}\Delta(r_{i}^{(1)},z)dz\mid\boldsymbol{X}_{i}^{(1)}]\nonumber\\
=&\int_0^{(\boldsymbol{X}_i^{(1)})^\top\hat{\boldsymbol{\delta}}^{(0)}}\mathbb{E}[\omega_i^{(1)}\mathds{1}\{i\in\mathcal{I}_1\}\Delta(r_i^{(1)},z)\mid\boldsymbol{X}_i^{(1)}]dz\nonumber\\
=&\mathds{1}\{|\eta_i^{(1)}|\leq M_1\}\int_0^{(\boldsymbol{X}_i^{(1)})^\top\hat{\boldsymbol{\delta}}^{(0)}}z/Adz\nonumber\\
=&\frac{\mathds{1}\{|\eta_{i}^{(1)}|\leq M_{1}\}}{2A}((\boldsymbol{X}_{i}^{(1)})^{\top}\hat{\boldsymbol{\delta}}^{(0)})^{2}.
\label{T4-II-source}
\end{align}
Recall that, throughout the proof, all probabilities and expectations are
understood conditionally on the fixed $\boldsymbol{X}$. Since $\mathbb{E}(n_{\mathcal{I}_1})=\sum_{i=1}^{n_1}\mathbb{P}(i\in\mathcal{I}_1)$
and $\mathds{1}\{i\in\mathcal{I}_1\},~i=1,2,\ldots,n_1$ are conditionally independent, we have
\begin{equation*}
    \operatorname{Var}(n_{\mathcal{I}_1})=\sum_{i=1}^{n_1}\mathbb{P}(i\in\mathcal{I}_1)(1-\mathbb{P}(i\in\mathcal{I}_1))\leq\sum_{i=1}^{n_1}\mathbb{P}(i\in\mathcal{I}_1)=\mathbb{E}(n_{\mathcal{I}_1}).
\end{equation*}
Since
\begin{align*}
    \mathbb{P}(|n_{\mathcal{I}_1}-\mathbb{E}(n_{\mathcal{I}_1})|>M\sqrt{\operatorname{Var}(n_{\mathcal{I}_1})})\leq\frac{\operatorname{Var}(n_{\mathcal{I}_1}-\mathbb{E}(n_{\mathcal{I}_1}))}{M^2\operatorname{Var}(n_{\mathcal{I}_1})}=1/M^2,~M>0,
\end{align*}
we have $n_{\mathcal{I}_1}-\mathbb{E}(n_{\mathcal{I}_1})=O_p(\sqrt{\operatorname{Var}(n_{\mathcal{I}_1})})$. 
\begin{align*}
    \frac{n_0+n_{\mathcal{I}_1}}{n_0+\mathbb{E}(n_{\mathcal{I}_1})}-1=O_p\big(\frac{\sqrt{\mathbb{E}(n_{\mathcal{I}_1})}}{n_0+\mathbb{E}(n_{\mathcal{I}_1})}\big)=o_p(1).
\end{align*}
Since $n_0+\mathbb{E}(n_{\mathcal{I}_1})\to\infty$, we obtain that
\begin{equation}
    \frac{n_0+n_{\mathcal{I}_1}}{n_0+\mathbb{E}(n_{\mathcal{I}_1})}\xrightarrow{P}1.
    \label{n1=En1}
\end{equation}

For any $\boldsymbol{\delta}\in\mathbb{R}^p$, we denote $(II)(\boldsymbol{\delta})=\frac{1}{n_0+n_1}\big\{\sum_{i=1}^{n_0}B_i^{(0)}(\boldsymbol{\delta})+\sum_{i=1}^{n_1}B_i^{(1)}(\boldsymbol{\delta})\big\}$ and $\mathcal{R}(\boldsymbol{\delta})=(II)(\boldsymbol{\delta})-\mathbb{E}[(II)(\boldsymbol{\delta})\mid\boldsymbol{X}]$, where $B_i^{(0)}(\boldsymbol{\delta})=\int_0^{(X_i^{(0)})^\top\boldsymbol{\delta}}\Delta(\varepsilon_i^{(0)},z)dz$ and $B_i^{(1)}(\boldsymbol{\delta})=\omega_i^{(1)}\mathds{1}\{i\in\mathcal{I}_1\}\int_0^{(\boldsymbol{X}_i^{(1)})^\top\boldsymbol{\delta}}\Delta(r_i^{(1)},z)dz$. 

We first establish the local curvature bound for a generic perturbation. For any  $\boldsymbol{\delta}\in\mathcal{E}(H_0,3)$ satisfying $\Psi_n\|\boldsymbol{\delta}\|_1\leq \min\{\xi,A\min(\tau,1-\tau)\}$, we have, for every target and source observation, 
\begin{center}
$    |(\boldsymbol{X}_i^{(k)})^\top\boldsymbol{\delta}|\leq
    \|\boldsymbol{X}_i^{(k)}\|_\infty\|\boldsymbol{\delta}\|_1\leq
    \Psi_n\|\boldsymbol{\delta}\|_1.$
\end{center}
Therefore, the local quadratic lower bounds in \eqref{T4-II-target-1} and \eqref{T4-II-source} are valid for $\boldsymbol{\delta}$. Combining \eqref{T4-II-target-1}, \eqref{T4-II-source}, \eqref{n1=En1}, and Assumption \ref{Gaussian maximal inequalit}, we obtain, with probability tending to one,
\begin{align}
    \mathbb{E}[(II)(\boldsymbol{\delta})\mid\boldsymbol{X}]&\geq\frac{1}{2(n_0+n_1)}[\kappa\sum_{i=1}^{n_0}((\boldsymbol{X}_{i}^{(0)})^\top\boldsymbol{\delta})^2+\sum_{i=1}^{n_1}\mathds{1}\{|\eta_{i}^{(1)}|\leq M_{1}\}((\boldsymbol{X}_{i}^{(1)})^\top{\boldsymbol{\delta}})^2/A]\nonumber\\
    &\geq\frac{\min(\kappa,A^{-1})}{2(n_0+n_1)}[\sum_{i=1}^{n_0}((\boldsymbol{X}_{i}^{(0)})^\top{\boldsymbol{\delta}})^2+\sum_{i\in\mathcal{I}_1}((\boldsymbol{X}_{i}^{(1)})^\top{\boldsymbol{\delta}})^2]\nonumber\\
    &\geq\frac{\min(\kappa,A^{-1})}{2(n_0+n_1)}{\boldsymbol{\delta}}^\top[(\boldsymbol{X}^{(0)})^\top\boldsymbol{X}^{(0)}+(\boldsymbol{X}^{(1)})^\top \mathbf{I}_{\mathcal{I}_{1}}\boldsymbol{X}^{(1)}]{\boldsymbol{\delta}}\nonumber\\
    &\geq\frac{\min(\kappa,A^{-1})}{2(n_0+n_1)}\|{\boldsymbol{\delta}}\|^2_2\phi^2(n_0+n_{\mathcal{I}_{1}})\nonumber\\
    &\geq\frac{1}{2}\rho_{\mathcal{I}_{1}}\phi^2_{min}\|{\boldsymbol{\delta}}\|^2_2,
    \label{IIdownbound}
\end{align}
where $\phi^2_{min}:=\min(\kappa,A^{-1})\phi^2$ and $\rho_{\mathcal{I}_{1}}:=(n_0+\mathbb{E}(n_{\mathcal{I}_{1}}))/((n_0+n_1)$. 
It remains to control $(II)-\mathbb{E}[(II)\mid \boldsymbol{X}]$. 
Let $\{\sigma_i^{(0)}\}_{i=1}^{n_0}$ and
$\{\sigma_i^{(1)}\}_{i=1}^{n_1}$ be two mutually independent
sequences of Rademacher random variables, independent of the data.
By Theorem \ref{pkboundedwkbounded}, $\omega_i^{(1)}\leq C_{\omega}$ on the selected observations. For any deterministic constant $R>0$, by the conditional symmetrization and the Rademacher contraction inequality yield
\begin{align*}
    &\mathbb{E}(\sup_{\|\boldsymbol{\delta}\|_1\le R}|\mathcal{R}(\boldsymbol{\delta})|\mid\boldsymbol{X})
    \leq\frac{2}{n_0+n_1}\mathbb{E}\mathbb{E}_{\sigma}\sup_{\|\boldsymbol{\delta}\|_1\le R}|\sum_{i=1}^{n_0}\sigma_{i}^{(0)}B_i^{(0)}(\boldsymbol{\delta})+\sum_{i=1}^{n_1}\sigma_{i}^{(1)}B_i^{(1)}(\boldsymbol{\delta})|\\
    \leq&\frac{2C}{n_0+n_1}\mathbb{E}\mathbb{E}_{\sigma}\sup_{\|\boldsymbol{\delta}\|_1\le R}|\sum_{i=1}^{n_0}\sigma_{i}^{(0)}(\boldsymbol{X}_i^{(0)})^\top\boldsymbol{\delta}+C_{\omega}\sum_{i=1}^{n_1}\sigma_{i}^{(1)}\mathds{1}\{i\in\mathcal{I}_1\}(\boldsymbol{X}_i^{(1)})^\top\boldsymbol{\delta}|\\
    \leq&\frac{2CR}{n_0+n_1}\mathbb{E}\mathbb{E}_{\sigma}\big\|\sum_{i=1}^{n_0}\sigma_i^{(0)}\boldsymbol{X}_i^{(0)}+C_\omega\sum_{i=1}^{n_1}\sigma_i^{(1)}\mathds{1}\{i\in\mathcal{I}_1\}\boldsymbol{X}_i^{(1)}\big\|_\infty.
\end{align*}
By the conditional Rademacher maximal inequality and the definition of $\Psi_n$,
\begin{align*}
    \mathbb{E}(\sup_{\|\boldsymbol{\delta}\|_1\le R}|\mathcal{R}(\boldsymbol{\delta})|\mid\boldsymbol{X})&\leq\frac{2C\Psi_nR\sqrt{\log(2p)}}{n_0+n_1}(\sqrt{n_0}+C_{\omega}\mathbb{E}[\sqrt{n_{\mathcal{I}_1}}|\boldsymbol{X}])\\
    &\leq\frac{2C\Psi_nR\sqrt{(n_0+\mathbb{E}(n_{\mathcal{I}_1}))\log(2p)}}{n_0+n_1}.
\end{align*}
Therefore, provided that $\Psi_n=O(1)$ and $\lambda_{\mathrm{ora}}\asymp\rho_{\mathcal{I}_1}\sqrt{\log p/(n_0+\mathbb{E}(n_{\mathcal{I}_1}))}$, by Markov's inequality, we can obtain
\begin{align}
    \sup_{\|\boldsymbol{\delta}\|_1\leq R}|\mathcal{R}(\boldsymbol{\delta})|=\sup_{\|\boldsymbol{\delta}\|_1\leq R}|(II)(\boldsymbol{\delta})-\mathbb{E}[(II)(\boldsymbol{\delta})\mid\boldsymbol{X}]|=O_p(\lambda_{\mathrm{ora}} R).
    \label{T4supRdelta}
\end{align}
We next apply this bound through a localization argument. 
We next show that $\hat{\boldsymbol{\delta}}^{(0)}$ lies in the local region required in \eqref{IIdownbound}. Define $c_n=s\lambda_{\mathrm{ora}}/\rho_{\mathcal{I}_1}$. We first prove $\|\hat{\boldsymbol{\delta}}^{(0)}\|_1=O_p(c_n)$.

To this end, fix an arbitrary $\epsilon>0$. By \eqref{T4supRdelta}, there exists a constant $C_{\epsilon}>0$ such that, for any deterministic radius $R>0$, $\mathbb{P}(\sup_{\|\boldsymbol{\delta}\|_1\leq R}|\mathcal{R}(\boldsymbol{\delta})|\leq C_\epsilon\lambda_\mathrm{ora}R)\geq1-\epsilon$ for all sufficiently large $n$. Let $R_n=D_{\epsilon}c_n=D_{\epsilon}s\lambda_{\mathrm{ora}}/\rho_{\mathcal{I}_1}$, where $D_{\epsilon}>0$ is a sufficiently large constant to be specified below. Since $\lambda_\mathrm{ora}/\rho_{\mathcal{I}_1}\asymp\sqrt{\log p/(n_0+\mathbb{E}(n_{\mathcal{I}_1})})$, the condition in Theorem \ref{oracletheorem} gives
\begin{equation*}
    \Psi_nR_n\asymp D_\epsilon\Psi_ns\sqrt{\frac{\log p}{n_0+\mathbb{E}(n_{\mathcal{I}_1})}}=o(1).
\end{equation*}
Therefore, for all sufficiently large $n$, 
\begin{equation}
    \Psi_nR_n\leq \min\{\xi,A\min(\tau,1-\tau)\}.
    \label{T4phinRnleq}
\end{equation}
Suppose, to the contrary, that $\|\hat{\boldsymbol{\delta}}^{(0)}\|_1>R_n$. Define $t:=R_n/\|\hat{\boldsymbol{\delta}}^{(0)}\|_1\in(0,1)$ and $\tilde{\boldsymbol{\delta}}:=t\hat{\boldsymbol{\delta}}^{(0)}$. Then, $\|\tilde{\boldsymbol{\delta}}\|_1=R_n$. Recall that, on the event $\|(I)\|_{\infty}\leq\lambda_{\mathrm{ora}}/2$, we have already shown that $\hat{\boldsymbol{\delta}}^{(0)}\in\mathcal{E}(H_0,3)$. Since the cone $\mathcal{E}(H_0,3)$ is homogeneous, $\tilde{\boldsymbol{\delta}}\in\mathcal{E}(H_0,3)$. Let $Q(\boldsymbol{\beta})=L(\boldsymbol{\beta})+\lambda_{\mathrm{ora}}\|\boldsymbol{\beta}\|_1$ denote the penalized objective function. By the definition of $\hat{\boldsymbol{\beta}}_{\mathrm{ora}}^{(0)}$, $Q(\hat{\boldsymbol{\beta}}_{\mathrm{ora}}^{(0)})\leq Q(\boldsymbol{\beta}^{(0)})$. Moreover, $\boldsymbol{\beta}^{(0)}+\tilde{\boldsymbol{\delta}}=(1-t)\boldsymbol{\beta}^{(0)}+t\hat{\boldsymbol{\beta}}_\mathrm{ora}^{(0)}$. By convexity of $Q(\cdot)$, we have
\begin{equation*}
    Q(\boldsymbol{\beta}^{(0)}+\tilde{\boldsymbol{\delta}})\leq(1-t)Q(\boldsymbol{\beta}^{(0)})+tQ(\hat{\boldsymbol{\beta}}_\mathrm{ora}^{(0)})\leq Q(\boldsymbol{\beta}^{(0)}).
\end{equation*}
Thus, the same basic-inequality argument used in deriving \eqref{IIupbound} gives
\begin{equation}
    (II)(\tilde{\boldsymbol{\delta}})\leq\frac{3\lambda_\mathrm{ora}}{2}\sqrt{s}\|\tilde{\boldsymbol{\delta}}\|_2.
    \label{T4IIdeltatildeupbound}
\end{equation}
On the other hand, by \eqref{T4phinRnleq}, $\Psi_n\|\tilde{\boldsymbol{\delta}}\|_1=\Psi_nR_n\leq \min\{\xi,A\min(\tau,1-\tau)\}$. Hence the local curvature bound \eqref{IIdownbound} is valid for $\tilde{\boldsymbol{\delta}}$. Applying \eqref{IIdownbound} and \eqref{T4supRdelta} with the deterministic radius $R_n$, we obtain, with probability at least $1-\epsilon-o(1)$,
\begin{equation}
    (II)(\tilde{\boldsymbol{\delta}})\geq\frac{1}{2}\rho_{\mathcal{I}_1}\phi_{\min}^2\|\tilde{\boldsymbol{\delta}}\|_2^2-2C_{\epsilon}\lambda_{\mathrm{ora}}\|\tilde{\boldsymbol{\delta}}\|_1.
    \label{T4IItildedeltabound}
\end{equation}
Furthermore, since $\tilde{\boldsymbol{\delta}}\in\mathcal{E}(H_0,3)$, 
\begin{align*}
\|\tilde{\boldsymbol{\delta}}\|_1 =\|\tilde{\boldsymbol{\delta}}_{H_0}\|_1+\|\tilde{\boldsymbol{\delta}}_{H_0^c}\|_1\leq4\|\tilde{\boldsymbol{\delta}}_{H_0}\|_1 \leq4\sqrt{s}\|\tilde{\boldsymbol{\delta}}\|_2.
\end{align*}
Since $\|\tilde{\boldsymbol{\delta}}\|_1=R_n$, it follows that 
\begin{equation}
    \|\tilde{\boldsymbol{\delta}}\|_2\geq\frac{R_n}{4\sqrt{s}}.
    \label{T4tildedeltalowbound}
\end{equation}
Combining \eqref{T4IIdeltatildeupbound}, \eqref{T4IItildedeltabound} and \eqref{T4tildedeltalowbound} gives
\begin{align}
    \frac{1}{2}\rho_{\mathcal{I}_1}\phi_{\min}^2\|\tilde{\boldsymbol{\delta}}\|_2^2&\leq\frac{3\lambda_{\mathrm{ora}}}{2}\sqrt{s}\|\tilde{\boldsymbol{\delta}}\|_2+2C_{\epsilon}\lambda_{\mathrm{ora}}\|\tilde{\boldsymbol{\delta}}\|_1\nonumber\nonumber\\
    &=\frac{3\lambda_{\mathrm{ora}}}{2}\sqrt{s}\|\tilde{\boldsymbol{\delta}}\|_2+2C_{\epsilon}\lambda_{\mathrm{ora}}R_n\nonumber\nonumber\\
    &\leq\frac{6s\lambda_{\mathrm{ora}}}{R_n}\|\tilde{\boldsymbol{\delta}}\|_2^2+\frac{32C_{\epsilon}s\lambda_{\mathrm{ora}}}{R_n}\|\tilde{\boldsymbol{\delta}}\|_2^2\nonumber\\
    &=\frac{6+32C_{\epsilon}}{D_{\epsilon}}\rho_{\mathcal{I}_1}\|\tilde{\boldsymbol{\delta}}\|_2^2.
    \label{T4-IIupperbound}
\end{align}
Choosing $D_{\epsilon}$ sufficiently large such that $(6+32C_{\epsilon})/D_{\epsilon}<\phi_{\mathrm{min}}^2/2$ leads to a contradiction. Hence,
\begin{equation*}
    \mathbb{P}(\|\hat{\boldsymbol{\delta}}^{(0)}\|_1>R_n)\leq\epsilon+o(1).
\end{equation*}
Since $\epsilon>0$ is arbitrary, we conclude that 
\begin{equation}
    \|\hat{\boldsymbol{\delta}}^{(0)}\|_1=O_p(s\lambda_\mathrm{ora}/\rho_{\mathcal{I}_1}).
    \label{hatdeltaOp}
\end{equation}
We now verify that $\hat{\boldsymbol{\delta}}^{(0)}$ lies in the local region required for \eqref{IIdownbound}. By \eqref{hatdeltaOp},
\begin{align}
\Psi_n\|\hat{\boldsymbol{\delta}}^{(0)}\|_1 & =O_p(\Psi_ns\lambda_\mathrm{ora}/\rho_{\mathcal{I}_1}) \nonumber\\
 & =O_p(\Psi_ns\sqrt{\log p/(n_0+\mathbb{E}(n_{\mathcal{I}_1})}))=o_p(1), 
 \label{T4psinhatdelta}
\end{align}
where the last equality follows from the condition in Theorem \ref{oracletheorem}. Since $\min\{\xi,A\min(\tau,1-\tau)\}>0$ is a fixed constant, \eqref{T4psinhatdelta} implies that 
\begin{equation*}
    \mathbb{P}(\Psi_n\|\hat{\boldsymbol{\delta}}^{(0)}\|_1\leq\min\{\xi,A\min(\tau,1-\tau)\})\to1.
\end{equation*}
Therefore, with probability tending to one, the local curvature bound in \eqref{IIdownbound} is valid for $\hat{\boldsymbol{\delta}}^{(0)}$. 
Next, we control
$\mathcal{R}(\hat{\boldsymbol{\delta}}^{(0)})$.
By \eqref{hatdeltaOp}, for any $\epsilon>0$, there exists a sufficiently large constant
$M_\epsilon>0$ such that
\begin{equation*}
    \mathbb{P}(\|\hat{\boldsymbol{\delta}}^{(0)}\|_1\leq M_\epsilon s\lambda_{\rm ora}/\rho_{\mathcal{I}_1})\geq 1-\epsilon
\end{equation*}
for all sufficiently large $n$. Applying \eqref{T4supRdelta} with the deterministic radius $R=M_\epsilon s\lambda_{\rm ora}/\rho_{\mathcal{I}_1}$ gives
\begin{equation}
|\mathcal{R}(\hat{\boldsymbol{\delta}}^{(0)})|=O_p(s\lambda_{\rm ora}^2/\rho_{\mathcal{I}_1}).
\label{T4RhatdeltaOp}
\end{equation}
Combining \eqref{IIdownbound} and \eqref{T4RhatdeltaOp}, we obtain
\begin{align}
(II)(\hat{\boldsymbol{\delta}}^{(0)})& =\mathbb{E}[
(II)(\hat{\boldsymbol{\delta}}^{(0)})|\boldsymbol{X}]+\mathcal{R}(\hat{\boldsymbol{\delta}}^{(0)}) \nonumber\\
&\geq \frac{1}{2}\rho_{\mathcal{I}_1}\phi_{\min}^2\|\hat{\boldsymbol{\delta}}^{(0)}\|_2^2
-O_p(s\lambda_{\rm ora}^2/\rho_{\mathcal{I}_1}).
\label{T4IIhatdeltabound1}
\end{align}
On the other hand, by \eqref{IIupbound},
\begin{equation*}
    (II)(\hat{\boldsymbol{\delta}}^{(0)})\leq\frac{3\lambda_{\rm ora}}{2}\sqrt{s}\,\|\hat{\boldsymbol{\delta}}^{(0)}\|_2.
\end{equation*}
Consequently,
\begin{equation}
\frac{1}{2}\rho_{\mathcal{I}_1}\phi_{\min}^2\|\hat{\boldsymbol{\delta}}^{(0)}\|_2^2
\leq\frac{3\lambda_{\rm ora}}{2}\sqrt{s}\,\|\hat{\boldsymbol{\delta}}^{(0)}\|_2+O_p(
s\lambda_{\rm ora}^2/\rho_{\mathcal{I}_1}).
\label{T4phideltabound2}
\end{equation}
By Young's inequality, for some constant $C>0$,
\begin{equation*}
    \frac{3\lambda_{\rm ora}}{2}\sqrt{s}\,\|\hat{\boldsymbol{\delta}}^{(0)}\|_2\leq\frac{1}{4}\rho_{\mathcal{I}_1}\phi_{\min}^2\|\hat{\boldsymbol{\delta}}^{(0)}\|_2^2+C\frac{s\lambda_{\rm ora}^2}{\rho_{\mathcal{I}_1}}.
\end{equation*}
Substituting this bound into \eqref{T4phideltabound2} and absorbing the quadratic term
on the right-hand side into the left-hand side yield
\begin{equation*}
    \|\hat{\boldsymbol{\delta}}^{(0)}\|_2^2=O_p(s\lambda_{\rm ora}^2/\rho_{\mathcal{I}_1}^2).
\end{equation*}
Finally, since $\lambda_{\rm ora}\asymp\rho_{\mathcal{I}_1}\sqrt{\log p/(n_0+\mathbb{E}(n_{\mathcal{I}_1})})$, we obtain
\begin{equation*}
    \|\hat{\boldsymbol{\beta}}_{\rm ora}^{(0)}-\boldsymbol{\beta}^{(0)}\|_2^2=O_p(\frac{s\log p}{n_0+\mathbb{E}(n_{\mathcal{I}_1})}).
\end{equation*}
This completes the proof.
\end{proof}

\section{Density Estimation}
\label{supp:Density estimation}
The error densities of each source domain are estimated using the B-spline projection approach of \citet{Gine2008AdaptiveEO}.
Following the approach, we estimate the density function with estimated residual $\tilde\varepsilon_i^{(k)}$ in the $k$-th source domain by the B-spline projection estimator defined as 
\begin{align}
\hat p_{\varepsilon^{(k)}}(u)
=
2^{j_k}
\sum_{i\in\mathcal J_{k}}
\sum_{\ell\in\mathbb Z}
\sum_{m\in\mathbb Z}
b_{\ell m}
N_{j_k,\ell,r}
(
\tilde\varepsilon_i^{(k)}
)
N_{j_k,m,r}(u)/{n_k},~ k=1,\ldots,K,
\label{Bsplinehat-pk}
\end{align}
where \(j_k\in\mathbb{N}_0\) is a resolution parameter which controls the complexity of the density estimator in each source domain, and $r\geq2$ is the order of the B-spline basis.
The order-$r$ cardinal B-spline is defined by $r\text{-}\mathrm{fold}$ convolution function 
$N_{0,r}(u)
=
\underbrace{
\mathds{1}_{[0,1)}
*\cdots*
 \mathds{1}_{[0,1)}
}_{r\text{ times}}(u)$  with $\mathds{1}_{[0,1)}(u)=\mathds{1}_{0 \le u < 1}$, which admits the closed form expression $\sum_{i=0}^{r}(-1)^{i}\binom{r}{i}(u-i)_{+}^{r-1}/(r-1)!$  and 
 can be appropriately transformed as $N_{j_k,\ell,r}(u)
:=N_{\ell,r}(2^{j_k}u)
=
N_{0,r}(2^{j_k}u-\ell),
\ell\in\mathbb Z.$
The coefficients $b_{\ell m}$ are the entries of the inverse of B-spline Gram matrix, ensuring that $K_{j_k}(u,v):=
2^{j_k}
\sum_{\ell\in\mathbb Z}
\sum_{m\in\mathbb Z}
b_{\ell m}
N_{j_k,\ell,r}(v)
N_{j_k,m,r}(u)
$ is the kernel of the orthogonal projection onto the
spline space generated by
$\{N_{j_k,\ell,r}(\cdot):\ell\in\mathbb Z\}$.

\section{Proofs of Lemmas 1--2 and Theorem \ref{estimatebetabound}} 
\label{supp:Appendix B}
In this section, we provide the proof of Theorem \ref{estimatebetabound} for the feasible SL-TL estimator. 
We adopt the fixed-design framework and treat
$\mathcal X=
\{\bm X_i^{(k)}: i\in\mathcal J_k,\ k=0,\ldots,K\}$
as deterministic, except for Lemma \ref{lemma hat p_k}. 
Unless otherwise stated, all probabilities, expectations, and stochastic-order statements are taken with respect to the errors and the random sample splitting.
For the cross-fitted procedure, we present the proof for a generic rotation and suppress the rotation index whenever no confusion arises. 
The results for other rotations follow by cyclic relabeling of the folds.

The proof relies on several auxiliary results. 
Lemma~\ref{lem:S1} controls the perturbation induced by the initial coefficient estimators, while Lemma~\ref{lem:S2} bounds the probability mass near the screening boundaries. 
Proposition~\ref{prop:gaussian-boundary} verifies the required boundary condition under Gaussian random designs. 
Lemma \ref{lemma:selection-consistency} establishes the consistency of the estimated selected subsets, and Lemma \ref{lemma hat p_k} provides the uniform convergence rate of the B-spline estimators for the source error densities. 
Combining these results with restricted curvature and empirical process arguments yields the $\ell_2$-error bound in Theorem \ref{estimatebetabound}.

\subsection{Proof of Lemma \ref{lem:S1}}
\label{supp:lemma S1}
\setcounter{lemma}{0}
\renewcommand{\thelemma}{S\arabic{lemma}}
\renewcommand{\theHlemma}{S.\arabic{lemma}}

\begin{lemma}
\label{lem:S1}
Suppose that Assumption~5 holds. Let
$\bar\gamma_n
=
\max_{0\le k\le K}\gamma_k.$
Suppose that 
$\Psi_n
:=
\max_{0\le k\le K}
\max_{i\in\mathcal J_k}
\|\bm X_i^{(k)}\|_\infty
\le C_X b_n$ for some deterministic sequence $b_n>0$ and some constant $C_X>0$, and that
$\rho_n
:=
b_n\bar\gamma_n
\to0.$
Define
$A_{n}
=
2\Psi_{n}
\max_{0\le k\le K}
\|
\tilde{\boldsymbol\beta}^{(k)}
-
\boldsymbol\beta^{(k)}
\|_1$, $a_n:=\sqrt{\rho_n}$
and $\mathcal E_n=\{A_n\le a_n\}.$
Then we have
\begin{center}
$a_n\to0
\quad\text{and}\quad
\mathbb P(\mathcal E_n)\to1.$
\end{center}
\end{lemma}

\begin{proof}
Since $K$ is fixed and Assumption~5 gives $\|\tilde{\boldsymbol\beta}^{(k)}-\boldsymbol\beta^{(k)}\|_1
=
O_p(\gamma_k),$ for $k=0,\ldots,K$, we have $\max_{0\le k\le K}
\|\tilde{\boldsymbol\beta}^{(k)}-\boldsymbol\beta^{(k)}\|_1
=
O_p(\bar\gamma_n).$
To see this, note that
\begin{center}
    $\max_{0\le k\le K}
\|\tilde{\boldsymbol{\beta}}^{(k)}-\boldsymbol{\beta}^{(k)}\|_1
/{
\overline\gamma_n
}
\le
\max_{0\le k\le K}
\|\tilde{\boldsymbol{\beta}}^{(k)}-\boldsymbol{\beta}^{(k)}\|_1
/{
\gamma_k
}$,
\end{center}
and the maximum of finitely many $O_p(1)$ random variables remains
$O_p(1)$.

Because the design is fixed and
$\Psi_n\le C_Xb_n$, it follows that
\[
A_n
=
2\Psi_n
\max_{0\le k\le K}
\|\tilde{\boldsymbol\beta}^{(k)}-\boldsymbol\beta^{(k)}\|_1
=
O_p(b_n\bar\gamma_n)
=
O_p(\rho_n).
\]
Therefore, we have $A_n/\rho_n=
O_p(1).$
Since $\rho_n\to0$, it holds that $a_n=\sqrt{\rho_n}\to0$
and $a_n/{\rho_n}=1/{\sqrt{\rho_n}}\to\infty.$
Consequently, we have
\begin{align*}
\mathbb P(\mathcal E_n^c)
&=
\mathbb P(A_n>a_n)
=
\mathbb P\left(
A_n/{\rho_n}
>
1/{\sqrt{\rho_n}}
\right)
\to0.
\end{align*}
Hence, it holds that
\begin{center}
    $\mathbb P(\mathcal E_n)
\to1.$
\end{center}
This completes the proof of Lemma \ref{lem:S1}.
\end{proof}

\subsection{Proof of Lemma \ref{lem:S2}}
\label{supp:lemma S2}

For the perturbation scale $a_n$ defined in
Lemma~\ref{lem:S1}, we impose the following local boundary regularity
condition on the fixed design:
\begin{equation}
\max_{1\le k\le K}
\frac{1}{n_k}
\sum_{i\in\mathcal J_k}
\mathds{1}
\big\{
\big|
|\eta_i^{(k)}|-M_k
\big|
\le a_n
\big\}
\le
C_\eta a_n
\label{boundary condition}
\end{equation}
for some constant $C_\eta>0$ independent of $n_k$. This condition only
prevents a non-negligible fraction of fixed design points from accumulating
within an $a_n$-neighborhood of the screening thresholds $\pm M_k$.

Let $L_\tau=-\tau A,$ $U_\tau=(1-\tau)A$ and define $R_k
=
[L_\tau,U_\tau]\times[-M_k,M_k]$, $R_k^{-}(t)
=
[L_\tau+t,U_\tau-t]
\times
[-M_k+t,M_k-t]$ and
$R_k^{+}(t)
=
[L_\tau-t,U_\tau+t]
\times
[-M_k-t,M_k+t].$

\begin{lemma}
\label{lem:S2}
Suppose that~\eqref{boundary condition} holds and that
$\sup_{u\in\mathbb R}
p_{\varepsilon^{(k)}}(u)
\le B,$ for $k=1,\ldots,K.$
Then, for all sufficiently large $n_k$,
\[
\max_{1\le k\le K}
\frac1{n_k}
\sum_{i\in\mathcal J_k}
\mathbb P
\big\{
(r_i^{(k)},\eta_i^{(k)})
\in
R_k^{+}(a_n)\setminus R_k^{-}(a_n)
\big\}
\le
C_Ba_n,
\]
where $C_B=4B+C_\eta.$
\end{lemma}

\begin{proof}
Since $a_n\to0$ and $M_k>0$, for all sufficiently large $n_k$,
$a_n
<
\min_{1\le k\le K}M_k$
and
$2a_n<U_\tau-L_\tau.$

For each $k$ and $i$, define
$D_{ki}(a_n)
=
\{
(r_i^{(k)},\eta_i^{(k)})
\in
R_k^{+}(a_n)\setminus R_k^{-}(a_n)
\}.$
If $D_{ki}(a_n)$ occurs, then at least one of the four boundaries
of $R_k$ is within distance $a_n$. Therefore, it holds
\begin{align}
D_{ki}(a_n)
\subseteq\;
\big\{
|r_i^{(k)}-L_\tau|\le a_n
\big\}
\cup
\big\{
|r_i^{(k)}-U_\tau|\le a_n
\big\}
\cup
\big\{
\big|
|\eta_i^{(k)}|-M_k
\big|
\le a_n
\big\}.
\label{eq:boundary-strip-inclusion}
\end{align}

Under the fixed-design framework, $\eta_i^{(k)}$ is deterministic and
$r_i^{(k)}
=
\varepsilon_i^{(k)}+\eta_i^{(k)}.$
Hence,
\begin{align*}
\mathbb P
(
|r_i^{(k)}-L_\tau|
\le a_n
)
&=
\int_{
L_\tau-a_n-\eta_i^{(k)}
}^{
L_\tau+a_n-\eta_i^{(k)}
}
p_{\varepsilon^{(k)}}(u)\,du
\le
2Ba_n.
\end{align*}
Similarly, we have
\begin{center}
$\mathbb P
(
|r_i^{(k)}-U_\tau|
\le a_n
)
\le
2Ba_n.$
\end{center}
By~\eqref{eq:boundary-strip-inclusion} and the union bound,
$\mathbb P\{D_{ki}(a_n)\}
\le
4Ba_n
+
\mathds{1}
\big\{
\big|
|\eta_i^{(k)}|-M_k
\big|
\le a_n
\big\}.$
Averaging over $i\in\mathcal J_k$ and applying \eqref{boundary condition}, we obtain
\[
\frac1{n_k}
\sum_{i\in\mathcal J_k}
\mathbb P\{D_{ki}(a_n)\}
\le
(4B+C_\eta)a_n=C_Ba_n.
\]
Taking the maximum over $k=1,\ldots,K$ completes the proof.
\end{proof}

The following proposition is not used as an additional probabilistic step inside the fixed-design proof. Instead, it shows that the deterministic
condition~\eqref{boundary condition} is satisfied by a Gaussian random design with probability tending to one.

\begin{proposition}
\label{prop:gaussian-boundary}
Suppose that, for each $k=1,\ldots,K$,
$\bm X_i^{(k)}
\stackrel{\mathrm{i.i.d.}}{\sim}
N_p(\bm 0,\Sigma^{(k)}),$
that $K$ is fixed, and
$M_{\min}
=
\min_{1\le k\le K}M_k>0.$
If
$\min_{1\le k\le K}n_ka_n\to\infty,$
then, for some constant $C_\eta>0$ independent of $n_k$,
\[
\mathbb P_X
\Big(
\max_{1\le k\le K}
\frac1{n_k}
\sum_{i\in\mathcal J_k}
\mathbf 1
\big\{
\big|
|\eta_i^{(k)}|-M_k
\big|
\le a_n
\big\}
\le C_\eta a_n
\Big)
\to1,
\]
where $\mathbb P_X$ denotes probability with respect to the random
covariates only.
\end{proposition}

\begin{proof}
Recall that
$\eta_i^{(k)}
=
(\bm X_i^{(k)})^\top
(\bm\beta^{(k)}-\bm\beta^{(0)})
=
(\bm X_i^{(k)})^\top\bm\delta^{(k)}.$
Therefore,
$\eta_i^{(k)}
\sim
N(0,v_k^2)$ and $
v_k^2
=
(\bm\delta^{(k)})^\top
\Sigma^{(k)}
\bm\delta^{(k)}.$
For $v>0$, denote the density function of $N(0,v^2)$ by
$\phi_v(x)
=
\exp(
-x^2/{(2v^2)}
)/\sqrt{2\pi}v.$
For every fixed $x\ne0$,
$\sup_{v>0}\phi_v(x)
=
e^{-1/2}/({\sqrt{2\pi}|x|}),$
because the maximum over $v$ is attained at $v=|x|$.

Since $a_n\to0$, for all sufficiently large $n_k$, we have
$a_n\le M_{\min}/2.$
Whenever
$\left||x|-M_k\right|\le a_n,$
we have
$|x|
\ge
M_k-a_n
\ge
M_{\min}/2.$
Thus, it holds that
\[
\sup_{v>0}
\sup_{
\left||x|-M_k\right|\le a_n
}
\phi_v(x)
\le
\frac{
2e^{-1/2}
}{
\sqrt{2\pi}M_{\min}
}.
\]
The set
$\left\{
x:
\left||x|-M_k\right|\le a_n
\right\}$
has total Lebesgue measure $4a_n$. Thus, if $v_k>0$, we have
\[
\mathbb P_{X}
\big\{
\big|
|\eta_i^{(k)}|-M_k
\big|
\le a_n
\big\}
\le
C_Ga_n,
\]
where
$C_G
=
8e^{-1/2}
/({
\sqrt{2\pi}M_{\min}
}).$
If $v_k=0$, then $\eta_i^{(k)}\equiv0$ and, since
$a_n<M_{\min}/2$ for all sufficiently large $n_k$, the event above is empty.
Hence the same bound holds in this case as well.

Define
$Z_{ki}
=
\mathds{1}
\big\{
\big|
|\eta_i^{(k)}|-M_k
\big|
\le a_n
\big\}.$
Then
$E_{X}(Z_{ki})\le C_Ga_n.$
By Bernstein's inequality, for some constant $c>0$,
\[
\mathbb P_{X}
\Big(
\frac1{n_k}
\sum_{i=1}^{n_k}
Z_{ki}
>
2C_Ga_n
\Big)
\le
\exp(-cn_ka_n).
\]
Since $n_ka_n\to\infty$ and $K$ is fixed,
\[
\mathbb P_{X}
\Big(
\max_{1\le k\le K}
\frac1{n_k}
\sum_{i=1}^{n_k}
Z_{ki}
\le
2C_Ga_n
\Big)
\to1.
\]
The result follows with $C_{\eta}=2C_G$.
\end{proof}

\subsection{Proof of Lemma \ref{lemma:selection-consistency}}
\label{supp:Lemma 1}

\begin{proof}
Fix a cross-fitting rotation $j$. We suppress the rotation index from the initial estimators.
We first establish the sandwich relation. 
By the definitions of
\(\tilde\varepsilon_i^{(k)}\),
\(\tilde\eta_i^{(k)}\) and \(r_i^{(k)}\), we have
\[
\begin{aligned}
\tilde\varepsilon_i^{(k)}
+
\tilde\eta_i^{(k)}
-
r_i^{(k)}
&=
Y_i^{(k)}
-
(\bm X_i^{(k)})^\top\tilde{\boldsymbol\beta}^{(k)}
+
(\bm X_i^{(k)})^\top
(
\tilde{\boldsymbol\beta}^{(k)}
-
\tilde{\boldsymbol\beta}^{(0)}
)
-
r_i^{(k)}
\\
&=
-
(\bm X_i^{(k)})^\top
(
\tilde{\boldsymbol\beta}^{(0)}
-
\boldsymbol\beta^{(0)}
).
\end{aligned}
\]
Therefore, on the event \(\mathcal E_n\) defined in Lemma
\ref{lem:S1},
\begin{align}
|
\tilde\varepsilon_i^{(k)}
+
\tilde\eta_i^{(k)}
-
r_i^{(k)}
|
&\le
\|\bm X_i^{(k)}\|_\infty
\|
\tilde{\boldsymbol\beta}^{(0)}
-
\boldsymbol\beta^{(0)}
\|_1 \nonumber
\\
&\le
\Psi_n
\|
\tilde{\boldsymbol\beta}^{(0)}
-
\boldsymbol\beta^{(0)}
\|_1\nonumber
\\
&\le
A_n/2
\le
a_n/2. \label{L1-1}
\end{align}
Similarly, we have
$\tilde\eta_i^{(k)}
-
\eta_i^{(k)}
=
(\bm X_i^{(k)})^\top
[
(
\tilde{\boldsymbol\beta}^{(k)}
-
\boldsymbol\beta^{(k)}
)
-
(
\tilde{\boldsymbol\beta}^{(0)}
-
\boldsymbol\beta^{(0)}
)
].$
Hence, on \(\mathcal E_n\), we have
\begin{align}
|
\tilde\eta_i^{(k)}
-
\eta_i^{(k)}
|
&\le
\|\bm X_i^{(k)}\|_\infty
[
\|
\tilde{\boldsymbol\beta}^{(k)}
-
\boldsymbol\beta^{(k)}
\|_1
+
\|
\tilde{\boldsymbol\beta}^{(0)}
-
\boldsymbol\beta^{(0)}
\|_1
] \nonumber
\\
&\le
2\Psi_n
\max_{0\le\ell\le K}
\|
\tilde{\boldsymbol\beta}^{(\ell)}
-
\boldsymbol\beta^{(\ell)}
\|_1 \nonumber
\\
&=
A_n
\le
a_n.\label{L1-2}
\end{align}

Take \(i\in\mathcal I_{k,j}^{-}\). By definition, it holds $L_\tau+a_n
\le
r_i^{(k)}
\le
U_\tau-a_n$, and 
$
|\eta_i^{(k)}|
\le
M_k-a_n.$
By \eqref{L1-1} and \eqref{L1-2}, on \(\mathcal E_n\), we have
\begin{center}
$\tilde\varepsilon_i^{(k)}
+
\tilde\eta_i^{(k)}
\ge
r_i^{(k)}-a_n/2
\ge
L_\tau
\quad
\text{and}
\quad
\tilde\varepsilon_i^{(k)}
+
\tilde\eta_i^{(k)}
\le
r_i^{(k)}+a_n/2
\le
U_\tau.$
\end{center}
Moreover, it holds that $|\tilde\eta_i^{(k)}|
\le
|\eta_i^{(k)}|
+
|\tilde\eta_i^{(k)}-\eta_i^{(k)}|
\le
M_k.$
Thus, we have $i\in\hat{\mathcal I}_{k,j},$
which shows that
\begin{equation}
\mathcal I_{k,j}^{-}
\subseteq
\hat{\mathcal I}_{k,j}.
\label{L1-left}
\end{equation}

Conversely, take \(i\in\hat{\mathcal I}_{k,j}\). By \eqref{L1-1} and \eqref{L1-2}, it follows that
\begin{center}
$L_\tau-a_n
\le
r_i^{(k)}
\le
U_\tau+a_n
\quad \text{and}
\quad
|\eta_i^{(k)}|
\le
|\tilde\eta_i^{(k)}|
+
|\tilde\eta_i^{(k)}-\eta_i^{(k)}|
\le
M_k+a_n.$
\end{center}
Therefore, it holds that $i\in\mathcal I_{k,j}^{+},$
and hence we have
\begin{equation}
\hat{\mathcal I}_{k,j}
\subseteq
\mathcal I_{k,j}^{+}.\label{L1-right}
\end{equation}
Combining \eqref{L1-left} and \eqref{L1-right}, we obtain $\mathcal I_{k,j}^{-}
\subseteq
\hat{\mathcal I}_{k,j}
\subseteq
\mathcal I_{k,j}^{+}$
on \(\mathcal E_n\). Since
\(\mathbb P(\mathcal E_n)\to1\), it follows that
\begin{equation}
\mathbb P(
\mathcal I_{k,j}^{-}
\subseteq
\hat{\mathcal I}_{k,j}
\subseteq
\mathcal I_{k,j}^{+}
)
\to1.
\label{L1-half}
\end{equation}

We next compare the average selection probabilities. 
For $\diamond\in\{-,+\}$, define
$\pi_k^\diamond(a_n)
:=
\sum_{i\in\mathcal J_k}
\mathbb P
\{
(r_i^{(k)},\eta_i^{(k)})
\in
R_k^\diamond(a_n)
\}/n_k.$
Since
$R_k^{-}(a_n)
\subseteq
R_k
\subseteq
R_k^{+}(a_n),$
Lemma~\ref{lem:S2} yields
\begin{equation}
 0
\le
\pi_k-\pi_k^{-}(a_n)
\le
C_Ba_n
\quad\text{and}\quad
0
\le
\pi_k^{+}(a_n)-\pi_k
\le
C_Ba_n.  \label{pai}
\end{equation}
Because $\pi_k\ge c_\pi>0$, by \eqref{pai}, we have
\begin{equation}
\left|
\frac{
\pi_k^\diamond(a_n)
}{
\pi_k
}
-1
\right|
\le
\frac{C_Ba_n}{c_\pi}
\to0,
\qquad
\diamond\in\{-,+\}.\label{asymp-pai}
\end{equation}

We now compare the expected cardinalities. Because
\(\mathcal J_{k,j}\) is uniformly sampled from \(\mathcal J_k\), the random splitting is independent of the errors, and $\mathbb P(i\in\mathcal J_{k,j})
=
n_{k,j}/{n_k}.$ 
Consequently, it holds that
\begin{align}
\mathbb E|\mathcal I_{k,j}|
&=
\sum_{i\in\mathcal J_k}
\mathbb P(i\in\mathcal J_{k,j})
\mathbb P\{
(r_i^{(k)},\eta_i^{(k)})
\in R_k
\}\nonumber
\\
&=
n_{k,j}
\sum_{i\in\mathcal J_k}
\mathbb P\{
(r_i^{(k)},\eta_i^{(k)})
\in R_k
\}/{n_k} \nonumber
\\
&=
n_{k,j}\pi_k.\label{EI}
\end{align}
Similarly, we have
\begin{equation}
\mathbb E|\mathcal I_{k,j}^{\diamond}|
=
n_{k,j}\pi_k^\diamond(a_n),
\qquad
\diamond\in\{-,+\}.
\label{EI+_}
\end{equation}
Combining \eqref{asymp-pai}--\eqref{EI+_}, we obtain
\begin{equation}
\frac{
\mathbb E|\mathcal I_{k,j}^{\diamond}|
}{
\mathbb E|\mathcal I_{k,j}|
}
=
\frac{
\pi_k^\diamond(a_n)
}{
\pi_k
}
\to1,
\qquad
\diamond\in\{-,+\}.
\label{frac-E}
\end{equation}

It remains to compare
\(\mathbb E|\hat{\mathcal I}_{k,j}|\) with
\(\mathbb E|\mathcal I_{k,j}|\). 
On the event $\mathcal E_n$, the sandwich relation gives
$|\mathcal I_{k,j}^-|
\le
|\hat{\mathcal I}_{k,j}|
\le
|\mathcal I_{k,j}^+|,$ while on $\mathcal E_n^c$, it holds $0\le|\mathcal I_{k,j}^{-}|,|\hat{\mathcal I}_{k,j}|,|\mathcal I_{k,j}^{+}|\le n_{k,j}$. Therefore, we have
\begin{center}
    $|\hat{\mathcal I}_{k,j}|=|\hat{\mathcal I}_{k,j}|\mathds{1}_{\mathcal{E}_n}+|\hat{\mathcal I}_{k,j}|\mathds{1}_{\mathcal{E}_n^{c}}
\le
|\mathcal I_{k,j}^{+}|\mathds{1}_{\mathcal{E}_n}+n_{k,j}\mathds{1}_{\mathcal{E}_n^{c}}
\le
|\mathcal I_{k,j}^{+}|+n_{k,j}\mathds{1}_{\mathcal{E}_n^{c}}$,
\end{center}
\begin{center}
    $|\hat{\mathcal I}_{k,j}|
    \geq
    |\hat{\mathcal I}_{k,j}|\mathds{1}_{\mathcal{E}_n}
    \geq
    |\mathcal I_{k,j}^{-}|\mathds{1}_{\mathcal{E}_n}
    =|\mathcal I_{k,j}^{-}|-|\mathcal I_{k,j}^{-}|\mathds{1}_{\mathcal{E}_n^{c}}
    \geq
    |\mathcal I_{k,j}^{-}|-n_{k,j}\mathds{1}_{\mathcal{E}_n^{c}}$.
\end{center}
Taking expectations yields
\begin{center}
$\mathbb E|\hat{\mathcal I}_{k,j}|
\le
\mathbb E|\mathcal I_{k,j}^{+}|
+n_{k,j}\mathbb P(\mathcal E_n^c), \quad
\mathbb E|\hat{\mathcal I}_{k,j}|
\ge
\mathbb E|\mathcal I_{k,j}^{-}|
-n_{k,j}\mathbb P(\mathcal E_n^c).$
\end{center}
Using \eqref{pai} and \eqref{EI+_}, we obtain
\begin{center}
$\mathbb E|\hat{\mathcal I}_{k,j}|
\le n_{k,j}\{\pi_{k}+C_Ba_n+\mathbb P(\mathcal E_n^c)\}, \quad
\mathbb E|\hat{\mathcal I}_{k,j}|
\ge n_{k,j}\{\pi_{k}-C_Ba_n-\mathbb P(\mathcal E_n^c)\}.$
\end{center}
Consequently, with \(a_n\to0\) and \(\mathbb P(\mathcal E_n^c)\to0\), it holds that
\begin{align}
\left|
\frac{
\mathbb E|\hat{\mathcal I}_{k,j}|
}{
n_{k,j}\pi_k
}
-1
\right|
\le
\frac{
C_Ba_n+\mathbb P(\mathcal E_n^c)
}{
\pi_k
}
\le
\frac{
C_Ba_n+\mathbb P(\mathcal E_n^c)
}{
c_\pi
}
\to0.
\label{eq:L1-estimated-expectation}
\end{align}
Since
$\mathbb E|\mathcal I_{k,j}|=n_{k,j}\pi_k$ in \eqref{EI},
equation \eqref{eq:L1-estimated-expectation} is equivalently written as
\[
\frac{
\mathbb E|\hat{\mathcal I}_{k,j}|
}{
\mathbb E|\mathcal I_{k,j}|
}
\to1.
\]
Combining this result with
\eqref{frac-E}, we obtain
\begin{equation}
\mathbb E|\mathcal I_{k,j}^-|
\asymp
\mathbb E|\hat{\mathcal I}_{k,j}|
\asymp
\mathbb E|\mathcal I_{k,j}|
\asymp
\mathbb E|\mathcal I_{k,j}^+|.
\label{eq:L1-all-expectations}
\end{equation}

Finally, we establish the convergence of the random cardinalities. Suppose that
$n_{k,j}\pi_k\to\infty.$
For $\diamond\in\{-,+\}$, conditional on the random partition
$\mathcal J_{k,j}$,
$|\mathcal I_{k,j}^{\diamond}|$ is a sum of independent Bernoulli
random variables. Therefore, it holds that
    $\operatorname{Var}
\left(
|\mathcal I_{k,j}^{\diamond}|
\,\middle|\,
\mathcal J_{k,j}
\right)
\le
E
\left(
|\mathcal I_{k,j}^{\diamond}|
\,\middle|\,
\mathcal J_{k,j}
\right).$
Taking expectations gives
\begin{equation}
\mathbb E\left[
\operatorname{Var}
\left(
|\mathcal I_{k,j}^{\diamond}|
\,\middle|\,
\mathcal J_{k,j}
\right)
\right]
\le
n_{k,j}\pi_k^\diamond(a_n).
\label{converg-EV}
\end{equation}

On the other hand, by the variance formula for uniform sampling
without replacement and all probabilities in the preceding sum belong to
\([0,1]\), it holds that
\begin{align*}
&\operatorname{Var}\left[
E
\left(
|\mathcal I_{k,j}^{\diamond}|
\,\middle|\,
\mathcal J_{k,j}
\right)
\right]
\\=&
n_{k,j}
\big(
1-\frac{n_{k,j}}{n_k}
\big)
\frac{1}{n_k-1}
\sum_{i\in\mathcal J_k}
\big[
\mathbb P
\{
(r_i^{(k)},\eta_i^{(k)})
\in
R_k^\diamond(a_n)
\}
-
\pi_k^\diamond(a_n)
\big]^2 \nonumber
\\
\le&
n_{k,j}
\big(
1-\frac{n_{k,j}}{n_k}
\big)\frac{1}{n_k-1}
\sum_{i\in\mathcal J_k}
\mathbb P
\{
(r_i^{(k)},\eta_i^{(k)})
\in
R_k^\diamond(a_n)
\}\nonumber
\\
=&
n_{k,j}\big(
1-\frac{n_{k,j}}{n_k}
\big)\frac{n_k}{n_k-1}\pi_k^\diamond(a_n)
\nonumber
\\
\le&
n_{k,j}\big(
1-\frac{n_{k,j}}{n_k}
\big)2\pi_k^\diamond(a_n)
\end{align*}
for all sufficiently large \(n_k\). Thus,
\begin{equation}
\operatorname{Var}\left[
E
\left(
|\mathcal I_{k,j}^{\diamond}|
\,\middle|\,
\mathcal J_{k,j}
\right)
\right]
\le
2n_{k,j}\pi_k^\diamond(a_n).
\label{converg-VE}
\end{equation}

Using the law of total variance and the variance formula for
sampling without replacement, \eqref{converg-EV} and \eqref{converg-VE} imply
\begin{equation}
\operatorname{Var}
\left(
|\mathcal I_{k,j}^{\diamond}|
\right)
\le
3n_{k,j}\pi_k^\diamond(a_n).
\label{VI}
\end{equation}
for all sufficiently large $n_k$.
Moreover, by \eqref{asymp-pai} and \(n_{k,j}\pi_k\to\infty\),
\begin{equation}
n_{k,j}\pi_k^\diamond(a_n)
=
n_{k,j}\pi_k
\frac{\pi_k^\diamond(a_n)}{\pi_k}
\to\infty.
\label{npai+_}
\end{equation}
Therefore, for every \(\epsilon>0\), Chebyshev's inequality gives
\[
\begin{aligned}
&
\mathbb P\left(
\left|
\frac{
|\mathcal I_{k,j}^{\diamond}|
}{
n_{k,j}\pi_k^\diamond(a_n)
}
-1
\right|
>\epsilon
\right)
\le
\frac{\operatorname{Var}
(
|\mathcal I_{k,j}^{\diamond}|
)}{
\epsilon^2
n_{k,j}^2
\{\pi_k^\diamond(a_n)\}^2
}
\le
\frac{
3
}{
\epsilon^2
n_{k,j}\pi_k^\diamond(a_n)
}
\to0.
\end{aligned}
\]
Hence, by \eqref{VI} and \eqref{npai+_}, we have
\begin{equation}
\frac{
|\mathcal I_{k,j}^{\diamond}|
}{
n_{k,j}\pi_k^\diamond(a_n)
}
\overset{p}{\to}1.
\label{I/npai}
\end{equation}
Combining \eqref{asymp-pai} and \eqref{I/npai}, we obtain
\begin{align}
\frac{
|\mathcal I_{k,j}^\diamond|
}{
n_{k,j}\pi_k
}
&=
\frac{
|\mathcal I_{k,j}^\diamond|
}{
n_{k,j}\pi_k^\diamond(a_n)
}
\frac{
\pi_k^\diamond(a_n)
}{
\pi_k
}
\overset{p}{\to}1,
\qquad
\diamond\in\{-,+\}.
\label{eq:L1-perturbed-cardinality}
\end{align}

It remains to transfer
\eqref{eq:L1-perturbed-cardinality} to the estimated set.
For every
$\epsilon>0$, on the event $\mathcal E_n$, the sandwich relation $|
\mathcal I_{k,j}^{-}
|
\le
|
\hat{\mathcal I}_{k,j}
|
\le
|
\mathcal I_{k,j}^{+}
|$
implies that
\[
\frac{
|\hat{\mathcal I}_{k,j}|
}{
n_{k,j}\pi_k
}
<1-\epsilon
~\Longrightarrow~
\frac{
|\mathcal I_{k,j}^-|
}{
n_{k,j}\pi_k
}
<1-\epsilon \quad \text{and}\quad
\frac{
|\hat{\mathcal I}_{k,j}|
}{
n_{k,j}\pi_k
}
>1+\epsilon
~\Longrightarrow~
\frac{
|\mathcal I_{k,j}^+|
}{
n_{k,j}\pi_k
}
>1+\epsilon.
\]
Consequently, it holds that
\begin{align*}
&
\left\{
\left|
\frac{
|\hat{\mathcal I}_{k,j}|
}{
n_{k,j}\pi_k
}
-1
\right|
>\epsilon
\right\}
\subseteq
\mathcal E_n^c
\cup
\left\{
\frac{
|\mathcal I_{k,j}^-|
}{
n_{k,j}\pi_k
}
<1-\epsilon
\right\}
\cup
\left\{
\frac{
|\mathcal I_{k,j}^+|
}{
n_{k,j}\pi_k
}
>1+\epsilon
\right\}.
\end{align*}
Taking probabilities and applying the union bound, we obtain
\begin{align*}
&\mathbb P\left(
\left|
\frac{
|\hat{\mathcal I}_{k,j}|
}{
n_{k,j}\pi_k
}
-1
\right|
>\epsilon
\right)
\le
\mathbb P(\mathcal E_n^c)
+
\mathbb P\left(
\frac{
|\mathcal I_{k,j}^-|
}{
n_{k,j}\pi_k
}
<1-\epsilon
\right)
+
\mathbb P\left(
\frac{
|\mathcal I_{k,j}^+|
}{
n_{k,j}\pi_k
}
>1+\epsilon
\right).
\end{align*}
By $\mathbb P(\mathcal E_n^c)\to0$ and
\eqref{eq:L1-perturbed-cardinality}, all three terms on the
right-hand side converge to zero. Therefore,
\begin{equation}
\frac{
|\hat{\mathcal I}_{k,j}|
}{
n_{k,j}\pi_k
}
\overset{p}{\to}1.
\label{eq:L1-estimated-cardinality}
\end{equation}

By \eqref{L1-half} and \eqref{eq:L1-all-expectations}, we have thus shown that, for every $k=1,\ldots,K$ and $j=1,2,3$,
\begin{center}
$\mathbb P(
\mathcal I_{k,j}^-
\subseteq
\hat{\mathcal I}_{k,j}
\subseteq
\mathcal I_{k,j}^+
)
\to1 \quad \text{and}\quad
\mathbb E|\mathcal I_{k,j}^-|
\asymp
\mathbb E|\hat{\mathcal I}_{k,j}|
\asymp
\mathbb E|\mathcal I_{k,j}|
\asymp
\mathbb E|\mathcal I_{k,j}^+|.$
\end{center}
Moreover, if $n_{k,j}\pi_k\to\infty$, then by \eqref{eq:L1-estimated-cardinality}, we have $|\hat{\mathcal I}_{k,j}|
/{
n_{k,j}\pi_k
}
\overset{p}{\to}1.$
Since \(K\) and the number of data splits are fixed, a finite union
bound shows that all the preceding conclusions hold simultaneously
over \(k=1,\ldots,K\) and \(j=1,2,3\). This completes the proof.
\end{proof}

\subsection{Proof of Lemma \ref{lemma hat p_k}}
For notational simplicity, we prove the result for a generic source sample $\mathcal J_k$ of size $n_k$. In each cross-fitting rotation, the same proof
applies after replacing $\mathcal J_k$ and $n_k$ by the corresponding
density estimation fold and its sample size. Since the fold sizes are of the same order as $n_k$, the resulting stochastic rates remain unchanged.
\begin{proof}
For each $k=1,\ldots,K$, define the B-spline
projection estimator based on the true errors by
\begin{align}
\tilde p_{\varepsilon^{(k)}}(u)
&=
\frac{1}{n_k}
\sum_{i\in\mathcal J_{k}}
K_{j_k}
(
u,\varepsilon_i^{(k)}
)
\nonumber\\
&=
\frac{2^{j_k}}{n_k}
\sum_{i\in\mathcal J_{k}}
\sum_{\ell\in\mathbb Z}
\sum_{m\in\mathbb Z}
b_{\ell m}
N_{j_k,\ell,r}
(
\varepsilon_i^{(k)}
)
N_{j_k,m,r}(u).
\end{align}
By the triangle inequality,
\begin{align}
\sup_{u\in\mathbb R}
|
\hat p_{\varepsilon^{(k)}}(u)
-
p_{\varepsilon^{(k)}}(u)
|
\leq {}&
\sup_{u\in\mathbb R}
|
\hat p_{\varepsilon^{(k)}}(u)
-
\tilde p_{\varepsilon^{(k)}}(u)
|
\nonumber\\
&+
\sup_{u\in\mathbb R}
|
\tilde p_{\varepsilon^{(k)}}(u)
-
p_{\varepsilon^{(k)}}(u)
|.
\label{eq:spline-decomposition}
\end{align}

Since $\tilde p_{\varepsilon^{(k)}}(\cdot)$ is the B-spline projection
estimator based on the i.i.d.\ true errors, the uniform convergence
result of \citet{Gine2008AdaptiveEO} gives
\begin{equation}
\label{eq:true-residual-spline}
\sup_{u\in\mathbb R}
|
\tilde p_{\varepsilon^{(k)}}(u)
-
p_{\varepsilon^{(k)}}(u)
|
=
O_p\big(
2^{-j_kt}
+
\sqrt{j_k\cdot2^{j_k}/{n_k}}
\big).
\end{equation}

It remains to control the error caused by replacing
$\varepsilon_i^{(k)}$ with $\tilde{\varepsilon}_i^{(k)}$. From the
definition of the two estimators, we obtain
\begin{align*}
&\sup_{u\in\mathbb R}
|
\hat p_{\varepsilon^{(k)}}(u)
-
\tilde p_{\varepsilon^{(k)}}(u)
|\\
\leq&
\frac{2^{j_k}}{n_k}
\sum_{i\in\mathcal J_{k}}
\sup_{u\in\mathbb R}
\sum_{\ell,m\in\mathbb Z}
|b_{\ell m}|
|
N_{j_k,\ell,r}
(\tilde{\varepsilon}_i^{(k)})
-
N_{j_k,\ell,r}
(\varepsilon_i^{(k)})
|
|N_{j_k,m,r}(u)|.
\end{align*}
For $r\geq2$, the cardinal B-spline $N_{0,r}$ is Lipschitz continuous.
Therefore,
\begin{center}
    $|
N_{j_k,\ell,r}(v_1)
-
N_{j_k,\ell,r}(v_2)
|
\leq
C2^{j_k}|v_1-v_2|.$
\end{center}
Moreover, the B-splines are compactly supported, and the entries
$b_{\ell m}$ of the inverse Gram matrix are uniformly bounded and
exponentially decreasing in $|\ell-m|$. Consequently,
\begin{equation}
\label{eq:spline-plugin-bound}
\sup_{u\in\mathbb R}
|
\hat p_{\varepsilon^{(k)}}(u)
-
\tilde p_{\varepsilon^{(k)}}(u)
|
\leq
C2^{2j_k}
\max_{i\in\mathcal J_{k}}
|
\tilde{\varepsilon}_i^{(k)}
-
\varepsilon_i^{(k)}
|.
\end{equation}
Recall that $\Psi_n=\max_{1\le k\le K}\max_{i\in\mathcal J_{k}}\|\bm X_i^{(k)}\|_\infty$, then by the definitions of the true and estimated errors, we have
\begin{align*}
\max_{i\in\mathcal J_{k}}
|
\tilde{\varepsilon}_i^{(k)}
-
\varepsilon_i^{(k)}
|
&=
\max_{i\in\mathcal J_{k}}
|
(\bm{X}_i^{(k)})^\top
(
\tilde{\bm{\beta}}^{(k)}
-
\bm{\beta}^{(k)}
)
|\\
&\leq
\Psi_{n}
\|
\tilde{\bm{\beta}}^{(k)}
-
\bm{\beta}^{(k)}
\|_1\\
&=
O_p(\Psi_{n}\gamma_k).
\end{align*}
It follows from \eqref{eq:spline-plugin-bound} that
\begin{equation}
\label{eq:spline-plugin-rate}
\sup_{u\in\mathbb R}
|
\hat p_{\varepsilon^{(k)}}(u)
-
\tilde p_{\varepsilon^{(k)}}(u)
|
=
O_p(
2^{2j_k}\Psi_{n}\gamma_k
).
\end{equation}
Combining \eqref{eq:spline-decomposition},
\eqref{eq:true-residual-spline}, and
\eqref{eq:spline-plugin-rate} yields
\[
\sup_{u\in\mathbb R}
|
\hat p_{\varepsilon^{(k)}}(u)
-
p_{\varepsilon^{(k)}}(u)
|
=
O_p\big(
2^{2j_k}\Psi_{n}\gamma_k
+
2^{-j_kt}
+
\sqrt{j_k2^{j_k}/{n_k}}
\big).
\]
Choosing $2^{j_k}
\asymp
(
n_k/{\log n_k}
)^{1/(2t+1)},$ $u_k=2^{2j_k}\Psi_{n}\gamma_k$ and $v_k=(
\log n_k/{n_k}
)^{t/(2t+1)}$,
we obtain that the bound becomes 
\begin{center}
    $\sup_{u\in\mathbb R}
|
\hat p_{\varepsilon^{(k)}}(u)
-
p_{\varepsilon^{(k)}}(u)
|
=
O_p(u_k+v_k)=o_p(1).$
\end{center}
This completes the proof.
\end{proof}

\subsection{Proof of Theorem \ref{estimatebetabound}}
\begin{proof}
Recall that $\hat{\boldsymbol\beta}_{\mathrm{SL\text{-}TL}}^{(0)}=(\hat{\boldsymbol{\beta}}^{(0)}_1+\hat{\boldsymbol{\beta}}^{(0)}_2+\hat{\boldsymbol{\beta}}^{(0)}_3)/3$. Therefore, it is sufficient to establish the convergence rate for $\hat{\boldsymbol{\beta}}^{(0)}_1$, since the same argument applies to $\hat{\boldsymbol{\beta}}^{(0)}_2$ and $\hat{\boldsymbol{\beta}}^{(0)}_3$. Recall the definitions of $\hat\omega_i^{(k)}$ and $\hat{\mathcal I}_{k,3}$. For the first rotation, let $N_3=\sum_{k=0}^{K}n_{k,3}$, $\hat {\boldsymbol{\delta}}^{(0)}_1=\hat{\boldsymbol{\beta}}_1^{(0)}-\boldsymbol{\beta}^{(0)}$, $\psi_{\tau}(x)=\tau-\mathds{1}\{x\leq0\}$, $\Delta(x,z)=\mathds{1}\{x\leq z\}-\mathds{1}\{x\leq 0\}$ and $r_{i}^{(k)}=\varepsilon_{i}^{(k)}+\eta_{i}^{(k)}$. By the optimality of $\hat{\boldsymbol{\beta}}_1^{(0)}$ and Knight's identity, following exactly the same argument as that used to obtain \eqref{I+II}, we obtain
\begin{equation}
    (II)\leq\lambda\|\boldsymbol{\beta}^{(0)}\|_{1}-\lambda\|\hat{\boldsymbol{\beta}}_1^{(0)}\|_{1}+(I)^\top\hat{\boldsymbol{\delta}}_1^{(0)},
\end{equation}
where 
\begin{align*}
&(I)=\frac{1}{N_3}\big[\sum_{i=1}^{n_{0,3}}\boldsymbol{X}_{i}^{(0)}\psi_{\tau}(\varepsilon_{i}^{(0)})+\sum_{k=1}^{K}\sum_{i=1}^{n_{k,3}}\hat{\omega}_{i}^{(k)}\mathds{1}\{i\in \hat{\mathcal{I}}_{k,3}\}\boldsymbol{X}_{i}^{(k)}\psi_{\tau}(r_{i}^{(k)})\big],
\end{align*}
\[
\begin{aligned}
(II)=&\frac{1}{N_3}\big[\sum_{i=1}^{n_{0,3}}\int_{0}^{(\boldsymbol{X}_{i}^{(0)})^{T}\hat{\boldsymbol{\delta}}_1^{(0)}}\Delta(\varepsilon_{i}^{(0)},z)dz
\\&+\sum_{k=1}^{K}\sum_{i=1}^{n_{k,3}}\int_{0}^{(\boldsymbol{X}_{i}^{(k)})^{T}\hat{\boldsymbol{\delta}}_1^{(0)}}\hat{\omega}_{i}^{(k)}\mathds{1}\{i\in \hat{\mathcal{I}}_{k,3}\}\Delta(r_{i}^{(k)},z)dz\big].
\end{aligned}
\]
We proceed in two steps, similar to the proof of Theorem \ref{oracletheorem}. All subsequent arguments are carried out on the high-probability event on which the estimated weights are well defined, positive and uniformly bounded.

\textbf{Step 1}. For each $k=1,2,\ldots,K$, define $\omega_i^{\circ,k}=1/(Ap_{\varepsilon^{(k)}}(\varepsilon_i^{(k)}))$. By Lemma \ref{lemma hat p_k}, Assumptions 1 and 5,
\begin{align}
    &|\hat{p}_{\varepsilon^{(k)}}(\tilde{\varepsilon}_i^{(k)})-p_{\varepsilon^{(k)}}(\varepsilon_i^{(k)})|\nonumber\\
    \leq&|\hat{p}_{\varepsilon^{(k)}}(\tilde{\varepsilon}_i^{(k)})-p_{\varepsilon^{(k)}}(\tilde{\varepsilon}_i^{(k)})|+|p_{\varepsilon^{(k)}}(\tilde{\varepsilon}_i^{(k)})-p_{\varepsilon^{(k)}}(\varepsilon_i^{(k)})|\nonumber\\
    \leq&O_p(2^{2j_k}\Psi_n\gamma_k+2^{-j_kt}+(j_k2^{j_k}/{n_{k,2}})^{1/2})+C|(\boldsymbol{X}_i^{(k)})^\top(\tilde{\boldsymbol{\beta}}^{(k)}-\boldsymbol{\beta}^{(k)})|\nonumber\\
    \leq&O_p(2^{2j_k}\Psi_n\gamma_k+2^{-j_kt}+(j_k2^{j_k}/{n_{k,2}})^{1/2}+\Psi_n \gamma_k).
    \label{T5-phat-p}
\end{align}
Moreover, by Lemma \ref{lemma:selection-consistency}, for every $i\in\hat{\mathcal{I}}_{k,3}$, $\varepsilon_i^{(k)}$ and $\tilde\varepsilon_i^{(k)}$ are restricted to bounded intervals with probability tending to one. Assumption \ref{pkbound} and \eqref{T5-phat-p} therefore imply  
\begin{align}
\max_{i\in\hat{\mathcal{I}}_{k,3}}|\hat{\omega}_i^{(k)}-\omega_i^{\circ,k}|=O_p(q_k),
\label{T5hatwOpq1}
\end{align}
where $q_k=2^{2j_k}\Psi_n\gamma_k+(\log n_k/{n_k})^{t/(2t+1)}$ absorbs $O_p(\Psi_n\gamma_k)$ under the conditions of Lemma \ref{lemma hat p_k}. There exist constants $0<\tilde c_{\omega}<\tilde C_{\omega}<\infty$ such that, with probability tending to one, $\tilde c_{\omega}\leq\hat{\omega}_i^{(k)}\leq \tilde C_{\omega},~i\in\hat{\mathcal{I}}_{k,3},~k=1,2,\ldots,K$.

We first control $(I)$. Let
$\mathcal{G}_1=\sigma\{\{(\boldsymbol{X}_i^{(k)},Y_i^{(k)}):i\in \mathcal{J}_{k,1}\cup \mathcal{J}_{k,2}\}_{k=0}^K\}$
be the $\sigma$-field generated by the samples used to construct the initial
estimators and density estimators. Conditional on
$(\mathcal{G}_1,\boldsymbol{X})$, the estimated weights and selection indicators are fixed, and only the independent errors in the final estimation samples remain random. Define
$\boldsymbol{U}_i^{(0)}=\boldsymbol{X}_i^{(0)}\psi_{\tau}(\varepsilon_i^{(0)})$ and $\boldsymbol{U}_i^{(k)}=\hat{\omega}_{i}^{(k)}\mathds{1}\{i\in \hat{\mathcal{I}}_{k,3}\}\boldsymbol{X}_{i}^{(k)}\psi_{\tau}(r_{i}^{(k)})$.
Then $(I)$ can be written as
\begin{align}
    (I)&=\frac{1}{N_3}\big[\sum_{i=1}^{n_{0,3}}\boldsymbol{X}_{i}^{(0)}\psi_{\tau}(\varepsilon_{i}^{(0)})+\sum_{k=1}^{K}\sum_{i=1}^{n_{k,3}}\hat{\omega}_{i}^{(k)}\mathds{1}\{i\in \hat{\mathcal{I}}_{k,3}\}\boldsymbol{X}_{i}^{(k)}\psi_{\tau}(r_{i}^{(k)})\big]\nonumber\\
    &=\underbrace{\frac{1}{N_3}\sum_{i=1}^{n_{0,3}}\boldsymbol{U}_i^{(0)}+\frac{1}{N_3}
    \sum_{k=1}^{K}\sum_{i=1}^{n_{k,3}}[\boldsymbol{U}_i^{(k)}-\mathbb{E}(\boldsymbol{U}_i^{(k)}\mid\mathcal{G}_1,\boldsymbol{X}]}_{\Lambda_1}\nonumber\\
    &~~~~+\underbrace{\frac{1}{N_3}\sum_{k=1}^{K}\sum_{i=1}^{n_{k,3}}\mathbb{E}(\boldsymbol{U}_i^{(k)}\mid\mathcal{G}_1,\boldsymbol{X})}_{\Lambda_2}.
    \label{T5-(I)-Lambda12}
\end{align}
For a fixed $j\in[p]$,
we have $\mathbb{E}(U_{ij}^{(0)}\mid X_{ij}^{(0)})=0$, $|U_{ij}^{(0)}|\leq\Psi_{n}$ and $|U_{ij}^{(k)}|\leq \tilde C_{\omega}\Psi_{n}$. 
By a straightforward derivation and Lemma \ref{lemma:selection-consistency}, we obtain
\begin{align*}
    N_3^2\tilde V_j:=&\sum_{i=1}^{n_{0,3}}Var(U_{ij}^{(0)}\mid\mathcal{G}_1, \boldsymbol{X})+\sum_{k=1}^{K}\sum_{i=1}^{n_{k,3}}Var(U_{ij}^{(k)}\mid\mathcal{G}_1, \boldsymbol{X})\\
    =&\sum_{i=1}^{n_{0,3}}\mathbb{E}((U_{ij}^{(0)})^2\mid\mathcal{G}_1, \boldsymbol{X})+\sum_{k=1}^{K}\sum_{i=1}^{n_{k,3}}\mathbb{E}((U_{ij}^{(k)})^2\mid\mathcal{G}_1, \boldsymbol{X})\\
    =&\sum_{i=1}^{n_{0,3}}\mathbb{E}((\boldsymbol{X}_{ij}^{(0)}\psi_{\tau}(\varepsilon_{i}^{(0)}))^2\mid\mathcal{G}_1, \boldsymbol{X})\\&+\sum_{k=1}^{K}\sum_{i=1}^{n_{k,3}}\mathbb{E}((\hat \omega_{i}^{(k)}\boldsymbol{X}_{ij}^{(k)}\psi_{\tau}(r_{i}^{(k)}))^2\mathds{1}\{i\in\hat{\mathcal{I}}_{k,3}\}\mid \mathcal{G}_1,\boldsymbol{X})\\
    \leq& n_{0,3}\Psi_{n}^2+\sum_{k=1}^{K}\sum_{i=1}^{n_{k,3}}\Psi_{n}^2\tilde C_{\omega}^2P(i\in\hat{\mathcal{I}}_{k,3}\mid\mathcal{G}_1,\boldsymbol{X})\\
    \leq& \Psi_{n}^2(1+\tilde C_{\omega}^2)(n_{0,3}+\sum_{k=1}^{K}\mathbb{E}(n_{\mathcal{I}_{k,3}}))\\
    =&\Psi_{n}^2(1+\tilde C_{\omega}^2)(n_{0}+\mathbb{E}(n_{\mathcal{I}}))/3.
\end{align*}
By Bernstein's inequality and Step 1 of the proof of Theorem \ref{oracletheorem}, we have
\begin{equation}
    \|\Lambda_1\|_{\infty}\leq\lambda_1\asymp\Psi_n\rho_{\mathcal{I}}\sqrt{\log p/(n_{0}+\mathbb{E}(n_{\mathcal{I}}))},
    \label{T5-Lambda1bound}
\end{equation}
where $\rho_{\mathcal{I}}=(n_{0}+\mathbb{E}(n_{\mathcal{I}}))/\sum_{k=0}^{K}n_{k}$.
By Lemma \ref{lemma:selection-consistency} and \eqref{T5hatwOpq1}, we obtain
\begin{align}
    \Lambda_2
    =&\frac{1}{N_3}\sum_{k=1}^{K}\sum_{i=1}^{n_{k,3}}(\mathbb{E}[(\hat{\omega}_i^{(k)}-\omega_i^{\circ,k})\mathds{1}\{i\in\hat{\mathcal{I}}_{k,3}\}\boldsymbol{X}_i^{(k)}\psi_{\tau}(r_i^{(k)})\mid\mathcal{G}_1,\boldsymbol{X}]\nonumber\\
    &+\mathbb{E}[\omega_i^{\circ,k}\mathds{1}\{i\in\hat{\mathcal{I}}_{k,3}\}\boldsymbol{X}_i^{(k)}\psi_{\tau}(r_i^{(k)})\mid\mathcal{G}_1,\boldsymbol{X}])\nonumber\\
    \leq&\frac{C_{1}\Psi_n}{N_3}\sum_{k=1}^{K}q_k\mathbb{E}(n_{\mathcal{I}_{k,3}})+\frac{1}{N_3}\sum_{k=1}^{K}\sum_{i=1}^{n_{k,3}}\mathbb{E}[\omega_i^{\circ,k}\mathds{1}\{i\in\hat{\mathcal{I}}_{k,3}\}\boldsymbol{X}_i^{(k)}\psi_{\tau}(r_i^{(k)})\mid\mathcal{G}_1,\boldsymbol{X}]
    \label{T5-Lambda2-bound1}
\end{align}
We next control $\mathbb{E}[\omega_i^{\circ,k}\mathds{1}\{i\in\hat{\mathcal{I}}_{k,3}\}\boldsymbol{X}_i^{(k)}\psi_{\tau}(r_i^{(k)})\mid\mathcal{G}_1,\boldsymbol{X}]$. Let $d_i^{(0,k)}=(\boldsymbol{X}_i^{(k)})^\top (\tilde{\boldsymbol{\beta}}^{(0)}-\boldsymbol{\beta}^{(0)})$. Since $\tilde{\varepsilon}_i^{(k)}+\tilde{\eta}_i^{(k)}=r_i^{(k)}-d_i^{(0,k)}$, a direct calculation gives
\begin{align}
    &\|\mathbb{E}[\omega_i^{\circ,k}\mathds{1}\{i\in\hat{\mathcal{I}}_{k,3}\}\boldsymbol{X}_i^{(k)}\psi_{\tau}(r_i^{(k)})\mid\boldsymbol{X}_i^{(k)}]\|_{\infty}
    \nonumber\\
    \leq&\frac{\Psi_n\mathds{1}\{|\tilde\eta_i^{(k)}|\leq M_k\}}{A}\big|\int_{-\tau A+d_i^{(0,k)}}^{(1-\tau)A+d_i^{(0,k)}}\psi_{\tau}(u)du\big|\nonumber\\
    \leq&\frac{|d_i^{(0,k)}|\Psi_n\mathds{1}\{|\tilde\eta_i^{(k)}|\leq M_k\}}{A}.
    \label{wi0bound}
\end{align}
On the one hand, Assumption 5 and $\Psi_n=O(1)$ imply 
\begin{equation}
    \max_{i\in\mathcal{J}_{k,3}}|d_i^{(0,k)}|\leq\Psi_n\|\tilde{\boldsymbol{\beta}}^{(0)}-\boldsymbol{\beta}^{(0)}\|_1=O_p(\Psi_n\gamma_0).
    \label{di0bound1}
\end{equation}
On the other hand,
\begin{align}
    |d_i^{(0,k)}|&=|(\boldsymbol{X}_i^{(k)})^\top (\tilde{\boldsymbol{\beta}}^{(0)}-\boldsymbol{\beta}^{(0)})|
    \nonumber\\&
    =|(\boldsymbol{X}_i^{(k)})^\top(\tilde{\boldsymbol{\beta}}^{(0)}-\tilde{\boldsymbol{\beta}}^{(k)}+\tilde{\boldsymbol{\beta}}^{(k)}-\boldsymbol{\beta}^{(k)}+\boldsymbol{\beta}^{(k)}-\boldsymbol{\beta}^{(0)})|\nonumber\\
    &\leq|\tilde{\eta}_i^{(k)}|+\Psi_n(\|\tilde{\boldsymbol{\beta}}^{(k)}-\boldsymbol{\beta}^{(k)}\|_1+\|\boldsymbol{\beta}^{(k)}-\boldsymbol{\beta}^{(0)}\|_1)
    \lesssim M_k+\Psi_n(\gamma_k+h_k),
    \label{di0bound2}
\end{align}
Combining \eqref{wi0bound}, \eqref{di0bound1} and \eqref{di0bound2}, we have 
\begin{align*}
    &\frac{1}{N_3}\sum_{k=1}^{K}\sum_{i=1}^{n_{k,3}}\|\mathbb{E}[\omega_i^{\circ,k}\mathds{1}\{i\in\hat{\mathcal{I}}_{k,3}\}\boldsymbol{X}_i^{(k)}\psi_{\tau}(r_i^{(k)})\mid\mathcal{G}_1,\boldsymbol{X}]\|_{\infty}\\
    \leq& \frac{\Psi_n}{AN_3}\sum_{k=1}^{K}\sum_{i=1}^{n_{k,3}}|d_i^{(0,k)}|\mathds{1}\{|\tilde\eta_i^{(k)}|\leq M_k\}\\
    \leq& \frac{\Psi_n}{AN_3}\sum_{k=1}^{K}\sum_{i=1}^{n_{k,3}}\{(\Psi_n\gamma_0)\wedge(M_k+\Psi_n(\gamma_k+h_k))\}\mathds{1}\{|\tilde\eta_i^{(k)}|\leq M_k\}
\end{align*}
By the condition in Lemma \ref{lemma:selection-consistency},
\begin{equation*}
    \pi_k := \sum_{i\in\mathcal{J}_{k}}\mathbb P ( L_{\tau} \le r_i^{(k)} \le U_{\tau}, \ |\eta_i^{(k)}| \le M_k )/n_k \ge c_{\pi}>0.
\end{equation*}
Therefore, it holds that
$\mathbb{E}(n_{I_{k,3}})=n_{k,3}\pi_k\geq c_\pi n_{k,3}.$
Since $\sum_{i=1}^{n_{k,3}}\mathds{1}\{|\tilde\eta_i^{(k)}|\leq M_k\}\leq n_{k,3}$, we have $\sum_{i=1}^{n_{k,3}}\mathds{1}\{|\tilde\eta_i^{(k)}|\leq M_k\}\leq c_{\pi}^{-1}\mathbb{E}(n_{\mathcal{I}_{k,3}})$. Thus, we have
\begin{align}
    &\frac{1}{N_3}\sum_{k=1}^{K}\sum_{i=1}^{n_{k,3}}\|\mathbb{E}[\omega_i^{\circ,k}\mathds{1}\{i\in\hat{\mathcal{I}}_{k,3}\}\boldsymbol{X}_i^{(k)}\psi_{\tau}(r_i^{(k)})\mid\mathcal{G}_1,\boldsymbol{X}]\|_{\infty}\nonumber\\
    \lesssim& \frac{\Psi_n}{AN_3}\sum_{k=1}^{K}\{(\Psi_n\gamma_0)\wedge(M_k+\Psi_n(\gamma_k+h_k))\}\mathbb{E}(n_{\mathcal{I}_{k,3}}).
    \label{T5-Lambda2-bound2}
\end{align}
Combining \eqref{T5-Lambda2-bound1} and \eqref{T5-Lambda2-bound2}, we have
\begin{equation}
    \|\Lambda_2\|_{\infty}\leq \lambda_2\asymp\Psi_n\rho_{\mathcal{I}}\sum_{k=1}^{K}\pi_{M_k}[(\Psi_n\gamma_0)\wedge(M_k+\Psi_n(\gamma_k+h_k))+q_k],
    \label{T5-Lambda2-boundend}
\end{equation}
where $\pi_{M_k}=\mathbb{E}(n_{\mathcal{I}_k})/\mathbb{E}(n_\mathcal{I})$. By \eqref{T5-(I)-Lambda12}, \eqref{T5-Lambda1bound} and \eqref{T5-Lambda2-boundend}, we consequently obtain
\begin{equation}
    \|(I)\|_{\infty}\leq\lambda_1+\lambda_2=\lambda/2.
    \label{T5-(I)boundend}
\end{equation}

\textbf{Step 2}. Since the integral remainder term in Knight's identity is nonnegative,
the same argument as that used to derive \eqref{IIupbound} gives $\hat{\boldsymbol{\delta}}_1^{(0)}\in\mathcal{E}(H_0,3)$ and $(II)\leq\frac{3\lambda}{2}\sqrt{s}\|\hat{\boldsymbol{\delta}}_1^{(0)}\|_2$.

To control $(II)$, we similarly use $(II)=\mathbb{E}[(II)\mid \mathcal{G}_1,\boldsymbol{X}]+(II)-\mathbb{E}[(II)\mid \mathcal{G}_1,\boldsymbol{X}]$. Moreover, by \eqref{T5hatwOpq1}, the uniform boundedness of the oracle weights on the selected region and $q_k=o(1)$, with probability tending to one, 
\begin{center}
$\hat{\omega}_i^{(k)}\geq\omega_i^{\circ,k}/2,~i\in\mathcal{I}_{k,3}^-,~k=1,\ldots,K.$
\end{center}
By Lemma \ref{lemma:selection-consistency}, $\mathcal{I}_{k,3}^{-}\subseteq\hat{\mathcal{I}}_{k,3}$ with probability tending to one and the integral remainder in Knight's identity is nonnegative, the source contribution to the conditional curvature can be restricted to $\mathcal{I}_{k,3}^{-}$ and lower bounded, up to a universal constant, by the corresponding oracle weighted curvature.  Therefore, for any $\boldsymbol{\delta}\in\mathcal{E}(H_0,3)$, applying Assumptions 3 and 6, the same local curvature argument as that used to derive \eqref{IIdownbound} yields
\begin{equation}
    \mathbb{E}[(II)\mid\mathcal{G}_1,\boldsymbol{X}]\geq\rho_{\mathcal{I}}\tilde \phi^2_{min}\|{\boldsymbol{\delta}}\|^2_2/2,
    \label{T5-EII-lowerbound}
\end{equation}
where $\tilde \phi^2_{min}:=\min(\kappa,A^{-1})\phi^2$. 
Next, we control the empirical fluctuation of $(II)$. Conditional on $(\mathcal{G}_1,\boldsymbol{X})$, the observations in the final estimation samples are independent. Using the Rademacher fixed-point argument and the conditional Gaussian maximal inequality in the same manner as in the proof of Theorem \ref{oracletheorem}, we obtain
\begin{equation*}
    |(II)-\mathbb{E}((II)|\mathcal{G}_1,\boldsymbol{X})|\leq \lambda\|\hat{\boldsymbol{\delta}}_1^{(0)}\|_1.
\end{equation*}
Since $\hat{\boldsymbol{\delta}}_1^{(0)}\in\mathcal{E}(H_0,3)$, we have $\|\hat{\boldsymbol{\delta}}_1^{(0)}\|_1\leq4\sqrt{s}\|\hat{\boldsymbol{\delta}}_1^{(0)}\|_2$. Consequently, we have
\begin{equation*}
    (II)\geq\rho_I\tilde{\phi}_{\min}^2\|\hat{\boldsymbol{\delta}}_1^{(0)}\|_2^2/2-C\lambda\sqrt{s}\|\hat{\boldsymbol{\delta}}_1^{(0)}\|_2.
\end{equation*}
Combining the above lower bound with 
\begin{center}
$(II)\leq3\lambda\sqrt{s}\|\hat{\boldsymbol{\delta}}_1^{(0)}\|_2/2,$
\end{center}
we obtain $\rho_\mathcal{I}\|\hat{\boldsymbol{\delta}}_1^{(0)}\|_2^2\lesssim\lambda\sqrt{s}\|\hat{\boldsymbol{\delta}}_1^{(0)}\|_2$. Therefore, $\|\hat{\boldsymbol{\delta}}_1^{(0)}\|_2^2\lesssim s\lambda^2/\rho_\mathcal{I}^2$. The same argument applies to $\hat{\boldsymbol{\delta}}_2^{(0)}$ and $\hat{\boldsymbol{\delta}}_3^{(0)}$. Finally, since $\lambda=2(\lambda_1+\lambda_2)$, $\lambda_1\asymp\Psi_n\rho_{\mathcal{I}}\sqrt{\log p/(n_{0}+\mathbb{E}(n_{\mathcal{I}}))}$, $\lambda_2\asymp\Psi_n\rho_{\mathcal{I}}\sum_{k=1}^{K}\pi_{M_k}[(\Psi_n\gamma_0)\wedge(M_k+\Psi_n(\gamma_k+h_k))+q_k]$ and $\hat{\boldsymbol{\beta}}_{\mathrm{SL-TL}}^{(0)}=\sum_{\ell=1}^3\hat{\boldsymbol{\beta}}_\ell^{(0)}/3$, we have
\begin{equation*}
    \|\hat{\boldsymbol{\beta}}_{\mathrm{SL\text{-}TL}}^{(0)}-\boldsymbol{\beta}^{(0)}\|^2_2=O_p\big\{\frac{s\mathrm{log}p}{n_0+\mathbb{E}(n_\mathcal{I})}+s\Psi_n^2(\sum_{k=1}^{K}\pi_{M_k}[(\Psi_n\gamma_0)\wedge(M_k+\Psi_n(\gamma_k+h_k))+q_k])^2\big\}.
\end{equation*}
Moreover, assume the following conditions: (i) For $k=1,2,\ldots,K$, $n_k\gg n_0$, such that $\gamma_k\ll\gamma_0$ and $q_k\ll\gamma_0$, (ii) $\Psi_n=O_p(1)$, and (iii) $M_k,k=1,2,\ldots,K$ are fixed. We have
\begin{equation*}
    \|\hat{\boldsymbol{\beta}}_{\mathrm{SL\text{-}TL}}^{(0)}-\boldsymbol{\beta}^{(0)}\|^2_2=O_p\big\{\frac{s\mathrm{log}p}{n_0+\mathbb{E}(n_\mathcal{I})}+s(\sum_{k=1}^{K}\pi_{M_k}[\gamma_0\wedge(M_k+h_k)])^2\big\}.
\end{equation*}
\end{proof}

\section{Table and Figure}
\label{supp:Tables and Figures}

\begin{table}[!htbp]
\centering
\caption{The list of nine target brain tissues, their abbreviations, and sample sizes.}
\label{tab:target_brain_tissues}
\begin{tabular}{cllr}
\toprule
ID & Tissue & Abbreviation & Sample size \\
\midrule
1 & Caudate basal ganglia
  & Caudate
  & 300 \\

2 & Cerebellar hemisphere
  & C. hemisphere
  & 277 \\

3 & Cerebellum
  & Cerebellum
  & 266 \\

4 & Cortex
  & Cortex
  & 270 \\

5 & Frontal cortex BA9
  & F. cortex
  & 269 \\

6 & Hippocampus
  & Hippocampus
  & 255 \\

7 & Hypothalamus
  & Hypothalamus
  & 257 \\

8 & Nucleus accumbens basal ganglia
  & N. accumbens
  & 285 \\

9 & Putamen basal ganglia
  & Putamen
  & 254 \\
\bottomrule
\end{tabular}
\end{table}


\begin{figure}[!htbp]
    \centering
    \includegraphics[width=0.8\textwidth]{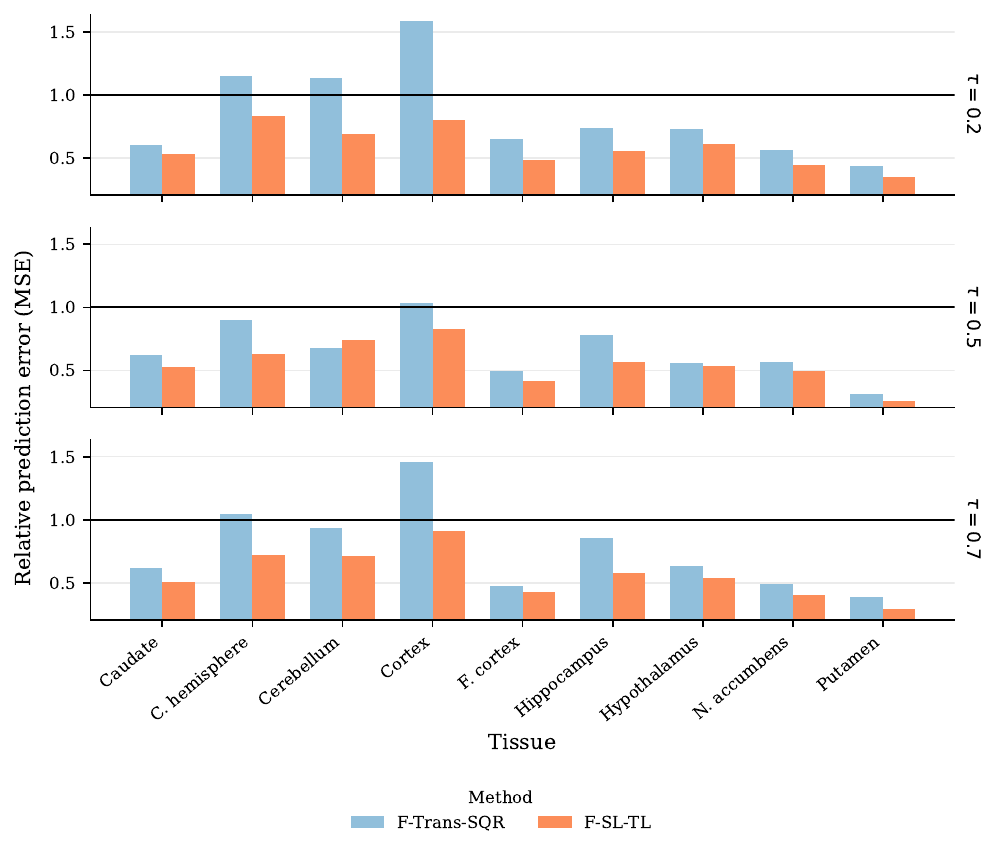}
    \caption{Relative mean squared prediction errors of F-Trans-SQR and F-SL-TL relative to that of Target-only across 9 different target tissues at quantile levels 0.2, 0.5 and 0.7. The horizontal line represents an RPE of one.}
\label{fig:gtex_mse}
\end{figure}

\newpage
\section{Algorithm}
\label{supp:alg}

\begin{algorithm}[!htbp]
\caption{Oracle Trans-SQR Method}
\label{alg:Trans-SQR}

\begin{algorithmic}[1]
\REQUIRE Target dataset \(\mathcal D^{(0)}\), source datasets
\(\{\mathcal D^{(k)}\}_{k=1}^{K}\), quantile level \(\tau\),
selection criterion $\mathcal{A}_{h}$, bandwidth parameters $h_{\omega}$ and $h_{\delta}$ and regularization parameters $\lambda_{\omega}$ and $\lambda_{\delta}$.

\medskip

\STATE \textbf{Transferring:} Estimate the pooled regression coefficient by
\[
\begin{aligned}
\hat{\boldsymbol\omega}_{\mathrm{s}}\in \arg\min_{\boldsymbol\omega\in\mathbb R^p}\big\{&\sum_{k\in\mathcal A_h\cup\{0\}}\sum_{i=1}^{n_k}\int_{-\infty}^{\infty}\rho_{\tau}(u)K((u+(\boldsymbol X_i^{(k)})^{\top}\boldsymbol\omega-Y_i^{(k)})/h_{\omega}) \\
&\quad /(N_{\mathcal A_h h_{\omega}})\,du +\lambda_{\omega}\|\boldsymbol\omega\|_1\big\}, 
\end{aligned}
\]
where $N_{\mathcal A_h}
:=n_0+n_{\mathcal{A}_h}$ and $K(\cdot)$ is the kernel function.

\medskip

\STATE \textbf{Debiasing:} Use the target data to estimate the correction vector
\[
\begin{aligned}
\hat{\boldsymbol\delta}_{\mathrm{S}}\in\arg\min_{\boldsymbol\delta\in\mathbb R^p}\big\{&\sum_{i=1}^{n_0}\int_{-\infty}^{\infty}\rho_{\tau}(u)K((u+(\boldsymbol X_i^{(0)})^{\top}\boldsymbol\delta-Y_i^{(0)}+(\boldsymbol X_i^{(0)})^{\top}\hat{\boldsymbol\omega})/h_{\delta}) \\
&\quad /(n_0h_{\delta})\,du+\lambda_{\delta}\|\boldsymbol\delta\|_1\big\}.
\end{aligned}
\]

\medskip

\ENSURE 
$\displaystyle
\hat{\boldsymbol{\beta}}_{\mathrm{S}}
=
\hat{\boldsymbol{\omega}}_{\mathrm{S}}
+
\hat{\boldsymbol{\delta}}_{\mathrm{S}}
$
\algorithmeqnumber{beta_oracletrans-sqr}
\end{algorithmic}

\end{algorithm}

\section{Simulation Settings for Density Estimator} 
\label{supp:Details}
For F-SL-TL, the error density in each source domain is estimated from the corresponding residuals in each cross-fitting rotation using the B-spline projection estimator in \eqref{Bsplinehat-pk}. We use cardinal B-splines of order
$r=4$ and set the smoothness parameter to $t=2$ for selecting
the resolution level. Following the resolution-level choice in
Lemma \ref{lemma hat p_k}, we choose $j_k$ such that
$2^{j_k}
\asymp
(
n_{k,2}/
{\log n_{k,2}}
)^{1/(2t+1)}.$ 
Specifically, in the numerical implementation, we set
$j_k
=
\max\{
0,\,
\operatorname{round}
[
\log_2
\{
(
n_{k,2}/
{\log n_{k,2}}
)^{1/(2t+1)}
\}
]
\}$. 
With $t=2$, this gives
$2^{j_k}
\asymp
(
n_{k,2}/
{\log n_{k,2}}
)^{1/5}$.
The spline basis is evaluated over the numerical interval
$[-8,8]$.
\end{document}